\documentclass[letterpaper,titlepage,12pt]{amsart}
\usepackage{amssymb,harvard,epsfig}

\usepackage{rotfloat}
\usepackage{setspace}
\usepackage{graphicx}
\usepackage{booktabs}
\usepackage{multirow}
\usepackage{amsmath}
\usepackage{amsfonts}
\usepackage{xspace}
\usepackage[mathcal]{eucal}
\usepackage[english]{babel}
\usepackage[babel,english=american]{csquotes}
\usepackage[T1]{fontenc}
\usepackage[latin9]{inputenc}
\usepackage{graphicx,amsmath,amssymb,epsfig}
\usepackage{amsfonts}
\usepackage{geometry}
\usepackage{setspace}
\usepackage{lscape}
\usepackage{caption}
\usepackage{rotating}

\newtheorem{theorem}{Theorem}

\newtheorem{condition}{Condition}

\newtheorem{corollary}{Corollary}

\newtheorem{lemma}{Lemma}

\usepackage{amsfonts}
\theoremstyle{plain}

\numberwithin{equation}{section}

\newcommand{\ind}{\overset{d}{\to}}

\newcommand{\inp}{\overset{p}{\to}}

\newcommand{\barr}{\begin{array}}
\newcommand{\earr}{\end{array}}
\newcommand{\ba}{\begin{align}}
\newcommand{\ea}{\end{align}}
\newcommand{\bs}{\begin{align}\begin{split}\nonumber}
\newcommand{\bsnumber}{\begin{align}\begin{split}}
\newcommand{\es}{\end{split}\end{align}}

\usepackage{amsmath}
\usepackage{amsfonts}
\usepackage{xspace}

\newcounter{myletter}
\newcounter{mynumber}

\newcommand{\be}{\begin{eqnarray}}
\newcommand{\ee}{\end{eqnarray}}
\newcommand{\bi}{\begin{itemize}}
\newcommand{\ei}{\end{itemize}}

\renewcommand{\hat}{\widehat}

\begin{document}

\title[ ]{Panel Data Models with Nonadditive Unobserved Heterogeneity: Estimation and Inference}
\author[ ]{Iv\'an Fern\'andez-Val$^\S$ \ \ Joonhwah Lee$^\ddag$}

\date{This version of \today. First version of April 2004. This paper is
based in part on the second chapter of
Fern\'andez-Val (2005)'s MIT PhD dissertation. We wish to thank Josh
Angrist, Victor Chernozhukov and Whitney Newey for encouragement and
advice. For suggestions and comments, we are grateful to Manuel
Arellano, Mingli Chen, the editor Elie Tamer, three anonymous referees and the
participants to the Brown and Harvard-MIT Econometrics seminar. We
thank Aju Fenn for providing us the data for the empirical example.
All remaining errors are ours. Fern\'andez-Val gratefully
acknowledges financial support from Fundaci\'on Caja Madrid,
Fundaci\'on Ram\'on Areces, and the National Science Foundation.
Please send comments or suggestions to ivanf@bu.edu (Iv\'an) or
jhlee82@mit.edu (Joonhwan)}

\thanks{ $\S$ Boston University, Department of Economics, 270 Bay State Road,Boston, MA
02215, ivanf@bu.edu.}
\thanks{ $\ddag$  Department of Economics, MIT, 50 Memorial Drive, Cambridge, MA 02142,
jhlee82@mit.edu.}

\begin{abstract}
This paper considers fixed effects estimation and inference in
linear and nonlinear panel data models with random coefficients and
endogenous regressors. The quantities of interest -- means,
variances, and other moments of the random coefficients -- are
estimated by cross sectional sample moments of GMM estimators
applied separately to the time series of each individual.  To deal with the
incidental parameter problem introduced by the noise
of the within-individual estimators in short panels, we develop bias corrections. These 
corrections are based on higher-order asymptotic expansions of the GMM estimators and produce improved point and interval estimates in
moderately long panels. Under
asymptotic sequences where the cross sectional and time series dimensions of the panel pass to infinity at the same rate, the uncorrected estimator has an asymptotic bias of the same order as the asymptotic variance. The bias corrections remove the bias without increasing variance. An
empirical example on cigarette demand based on Becker, Grossman and
Murphy (1994) shows significant heterogeneity in the price effect
across U.S. states.


\end{abstract}

\maketitle

\newpage


\section{Introduction} \label{s1}

This paper considers estimation and inference in linear and
nonlinear panel data models with random coefficients and endogenous
regressors. The quantities of interest are means, variances, and
other moments of the distribution of the random coefficients. In a state level panel
model of rational addiction, for example, we might be interested in
the mean and variance of the distribution of the price effect on cigarette consumption
across states, controlling for endogenous past and future consumptions. These
models pose important challenges in estimation and inference if the
relation between the regressors and random coefficients is left
unrestricted. Fixed effects methods based on GMM estimators applied
separately to the time series of each individual can be severely biased due to
the incidental parameter problem. The source of the bias is the
finite-sample bias of GMM if some of the regressors is endogenous or
the model is nonlinear in parameters, or nonlinearities if the parameter 
of interest is the variance or other high order moment of the random coefficients. Neglecting the heterogeneity and
 imposing fixed coefficients does not solve the problem, because the resulting
estimators are generally inconsistent for the mean of the random coefficients
(Yitzhaki, 1996, and Angrist, Graddy and Imbens,
2000).\footnote{Heckman and Vytlacil (2000) and Angrist (2004) find
sufficient conditions for fixed coefficient OLS and IV estimators to
be consistent for the average coefficient.}
Moreover, imposing fixed coefficients does not allow us to estimate 
other moments of the distribution of the random coefficients.

We introduce a  class of bias-corrected panel fixed effects
GMM estimators. Thus, instead of imposing
fixed coefficients, we estimate different coefficients for each
individual using the time series observations and correct for the resulting incidental parameter bias. For linear models, in addition to the bias correction, these estimators
differ from the standard fixed effects estimators in that both the
intercept and the slopes are different for each individual.
Moreover, unlike for the classical random coefficient estimators, they do not
rely on any restriction in the relationship between the regressors
and random coefficients; see Hsiao and Pesaran (2004) for a recent
survey on random coefficient models. This flexibility allows us to account for Roy (1951) type selection where the regressors are
decision variables with levels determined by their
returns. Linear models with Roy selection are commonly referred to as correlated random
coefficient models in the panel data literature.  In the presence of endogenous regressors, treating the random
coefficients as fixed effects is also convenient to overcome the identification
problems in these models pointed out by Kelejian (1974).

The most general models  we consider are semiparametric in the sense
that the distribution of the random coefficients is unspecified
and the parameters are identified from moment conditions. These
conditions can be nonlinear functions in parameters and variables,
accommodating both linear and nonlinear random coefficient models,
and allowing for the presence of time varying endogeneity in the
regressors not captured by the random coefficients. We use the
moment conditions to estimate the model parameters and other
quantities of interest via GMM methods applied separately to the time series of each individual. The resulting estimates can
be severely biased in short panels due to the incidental parameters
problem, which in this case is a consequence of the finite-sample
bias of GMM (Newey and Smith, 2004) and/or the nonlinearity of the
quantities of interest in the random coefficients. We develop
analytical corrections to reduce the bias.

To derive the bias corrections, we use higher-order expansions of the GMM estimators, extending the 
analysis in Newey and Smith (2004) for cross sectional estimators to
panel data estimators with fixed effects and serial dependence. If $n$ and $T$ denote 
the cross sectional and time series dimensions of the panel, the
corrections remove the leading term of the bias of order  $O(T^{-1})$, and center
the asymptotic distribution at the true parameter value under
sequences where $n$ and $T$ grow at the same rate. This approach is aimed to
perform well in econometric applications that use moderately long
panels, where the most important part of the bias is captured by the
first term of the expansion. Other previous studies that used a
similar approach for the analysis of linear and nonlinear fixed
effects estimators in panel data include, among others, Kiviet
(1995), Phillips and Moon (1999), Alvarez and Arellano (2003), Hahn
and Kuersteiner (2002), Lancaster (2002), Woutersen (2002),  Hahn and Newey (2004), and Hahn and
Kuersteiner (2011). See Arellano and Hahn
(2007) for a survey of this literature and additional references.

A first distinctive feature of our corrections is that
they can be used in overidentified models where the number of moment restrictions is greater than the
dimension of the parameter vector. This situation is common in
economic applications such as rational expectation models.
Overidentification complicates the analysis by introducing an
initial stage for estimating optimal weighting matrices to combine
the moment conditions, and precludes the use of the existing
methods. For example,  Hahn and
Newey's (2004) and Hahn and Kuersteiner's (2011) general bias reduction methods for nonlinear panel
data models do not cover optimal two-step GMM estimators. A second distinctive feature is that 
our results are specifically developed for models with multidimensional
nonadditive heterogeneity, whereas the previous studies focused
mostly on models with additive heterogeneity captured by an
scalar individual effect. Exceptions include Arellano and Hahn (2006) and
Bester and Hansen (2008), which also considered multidimensional
heterogeneity, but they focus on parametric likelihood-based panel
models with exogenous regressors. Bai (2009) analyzed related  linear panel models
with exogenous regressors and multidimensional interactive
individual effects. Bai's nonadditive heterogeneity allows for
interaction between individual effects and unobserved factors,
whereas the nonadditive heterogeneity that we consider allows for
interaction between individual effects and observed regressors.
A third distinctive feature of our analysis is the focus on moments of the 
distribution of the individual effects as  one of the main quantities of interest.


We illustrate the applicability of our methods with empirical and
numerical examples based on the cigarette demand application of
Becker, Grossman and Murphy (1994). Here, we estimate a linear
rational addictive demand model with state-specific coefficients for
price and common parameters for the other regressors using a panel
data set of U.S. states. We find that standard estimators that do
not account for non-additive heterogeneity by imposing a constant
coefficient for price can have important biases for the common parameters,
mean of the price coefficient and demand elasticities. The
analytical bias corrections are effective in removing the bias of
the estimates of the mean and standard deviation of the price
coefficient. Figure \ref{fig1} gives a preview of the empirical
results. It plots a normal approximation to the distribution of the
price effect based on uncorrected and bias corrected estimates of
the mean and standard deviation of the distribution of the price
coefficient. The figure shows that there is important heterogeneity
 in the price effect across states. The bias correction reduces by
more than 15\% the absolute value of the estimate of the mean effect
and by 30\% the estimate of the standard deviation.


Some of the results for the linear model are related to the recent literature on correlated random coefficient
panel models with fixed $T$. Graham and Powell (2008) gave
identification and estimation results for average effects. Arellano
and Bonhomme (2010) studied identification of the distributional
characteristics of the random coefficients in exogenous linear
models. None of these papers considered the case where some of the
regressors have time varying endogeneity not captured by the
random coefficients or the model is nonlinear. For nonlinear
models, Chernozhukov, Fern\'andez-Val, Hahn and Newey (2010)
considered identification and estimation of average and quantile treatment effects. Their
nonparametric and semiparametric bounds do not
require large-$T$, but they do not cover models with continuous
regressors and time varying endogeneity.

The rest of the paper is organized as follows. Section \ref{s2}
illustrates the type of models considered and discusses the nature
of the bias in two examples. Section \ref{s3} introduces the
general model and fixed effects GMM estimators. Section \ref{s4} derives the
asymptotic properties of the estimators. The bias corrections and
their asymptotic properties are given in Section \ref{s5}. Section
\ref{s6} describes the empirical and numerical examples. Section
\ref{s9} concludes with a summary of the main results. Additional
numerical examples, proofs and other technical details are given in
the online supplementary appendix Fern\'andez-Val and Lee (2012).


\section{Motivating examples} \label{s2}
In this section we describe in detail two simple examples to
illustrate the nature of the bias problem. The first example is a
linear correlated random coefficient model with endogenous
regressors. We show that averaging IV
estimators applied separately to the time series of each individual is biased for the mean of the random coefficients because of
the finite-sample bias of IV. The second example considers
estimation of the variance of the individual coefficients in a simple
setting without endogeneity. Here the sample variance of the estimators of the individual coefficients is biased because of the non-linearity of the variance operator in the individual
coefficients. The discussion in this section is heuristic leaving to
Section \ref{s4} the specification of precise regularity conditions
for the validity of the asymptotic expansions used.

\subsection{Correlated random coefficient model with endogenous regressors}
Consider the following panel model:
\begin{equation}\label{livmodel}
y_{it} = \alpha_{0i} + \alpha_{1i} x_{it} + \epsilon_{it}, \ (i = 1,
..., n; t = 1, ..., T);
\end{equation}
where $y_{it}$ is a response variable, $x_{it}$ is an observable
regressor, $\epsilon_{it}$ is an unobservable error term, and $i$
and $t$ usually index individual and time period,
respectively.\footnote{More generally, $i$ denotes a group index and
$t$ indexes the observations within the group. Examples of groups
include individuals, states, households, schools, or twins.} This is
a linear random coefficient model where the effect of the regressor
is heterogenous across individuals, but no restriction is imposed on
the distribution of the individual effect vector $\alpha_{i} :=
(\alpha_{0i}, \alpha_{1i})'$. The regressor can be correlated with
the error term and a valid instrument $(1,z_{it})$  is available for
$(1,x_{it})$, that is $E[\epsilon_{it} \mid \alpha_{i}] = 0$, $E[z_{it} \epsilon_{it} \mid \alpha_{i}] = 0$ and
$Cov[z_{it} x_{it} \mid \alpha_{i}] \neq 0$.  An important example of
this model is the panel version of the treatment-effect model
(Wooldridge, 2002 Chapter 10.2.3, and Angrist and Hahn, 2004). Here,
the objective is to evaluate the effect of a treatment ($D$) on an
outcome variable ($Y$). The average causal effect for each level of
treatment is defined as the difference between the potential outcome
that the individual would obtain with and without the treatment,
$Y_{d} - Y_{0}$. If individuals can choose the level of treatment,
potential outcomes and levels of treatment are generally correlated.
An instrumental variable $Z$ can be used to identify the causal
effect. If potential outcomes are represented as the sum of
permanent individual components and transitory individual-time
specific shocks, that is $Y_{jit} = Y_{ji} + \epsilon_{jit}$  for $j
\in \{0, 1\}$, then we can write this model as a special case of
(\ref{livmodel}) with $y_{it} = (1 - D_{it})Y_{0it} + D_{it}
Y_{1it}$, $\alpha_{0i} = Y_{0i}$, $\alpha_{1i} = Y_{1i} - Y_{0i}$,
$x_{it} = D_{it}$, $z_{it} = Z_{it}$, and $\epsilon_{it} = (1 -
D_{it}) \epsilon_{0it} + D_{it} \epsilon_{1it}$.

Suppose that we are ultimately interested in $\alpha_1 := E [\alpha_{1i}]$, the mean of the random slope coefficient. We could neglect the heterogeneity and run fixed effects OLS and IV
regressions in
\begin{equation*}
y_{it} = \alpha_{0i} + \alpha_{1} x_{it} + u_{it},
\end{equation*}
where $u_{it} = x_{it} (\alpha_{1i} - \alpha_{1}) + \epsilon_{it}$
in terms of the model (\ref{livmodel}). In this case, OLS and IV
estimate weighted means of the random coefficients in the
population; see, for example, Yitzhaki (1996) and Angrist and
Krueger (1999) for OLS, and Angrist, Graddy and Imbens (2000) for
IV. OLS puts more weight on individuals with higher variances of the
regressor because they give more information about the slope;
whereas IV weighs individuals in proportion to the variance of the
first stage fitted values because these variances reflect the amount
of information that the individuals convey about the part of the
slope affected by the instrument. These weighted means are generally 
different from the mean effect because the weights
can be correlated with the individual effects.

To see how these implicit OLS and IV weighting schemes affect the
estimand of the fixed-coefficient estimators, assume for simplicity that
the relationship between  $x_{it}$ and $z_{it}$ is linear, that is
$x_{it} = \pi_{0i} + \pi_{1i} z_{it} + \upsilon_{it},$
($\epsilon_{it}, \upsilon_{it}$) is normal conditional on ($z_{it},
\alpha_{i}, \pi_{i}$), $z_{it}$ is independent of
($\alpha_{i},\pi_{i}$), and ($\alpha_{i},\pi_{i}$) is normal, for
$\pi_{i} := (\pi_{0i}, \pi_{1i})'$. Then, the probability limits of
the OLS and IV estimators are\footnote{The limit of the IV estimator
is obtained from a first stage equation that imposes also fixed
coefficients, that is $x_{it} = \pi_{0i} + \pi_{1} z_{it} + w_{it},$ where $w_{it} = z_{it}(\pi_{1i} - \pi_1) + \upsilon_{it}$.
When the first stage equation is different for each individual, the
limit of the IV estimator is
\begin{equation*}
\alpha_{1}^{IV} = \alpha_{1} + 2 E [\pi_{1i}] Cov[\alpha_{1i},
\pi_{1i}]/\{E [\pi_{1i}]^{2} + Var[\pi_{1i}]\}.
\end{equation*}
See Theorems 2 and 3 in Angrist and Imbens (1995) for a related
discussion.}
\begin{eqnarray*}
\alpha_{1}^{OLS} &=&  \alpha_{1} + \{Cov[\epsilon_{it},
\upsilon_{it}]
+ 2 E [\pi_{1i} ] Var[z_{it}] Cov[\alpha_{1i}, \pi_{1i}]\}/Var[x_{it}],\\
\alpha_{1}^{IV} &=& \alpha_{1} + Cov[\alpha_{1i}, \pi_{1i}]/E
[\pi_{1i}].
\end{eqnarray*}
These expressions show that the OLS estimand differs from the
average coefficient in presence of endogeneity, i.e. non zero correlation
between the individual-time specific error terms, or whenever
 the random coefficients are correlated;
while the IV estimand differs from the average coefficient only in
the latter case.\footnote{This feature of the IV estimator is also
pointed out in Angrist, Graddy and Imbens (1999), p. 507.} In the
treatment-effects model, there exists correlation between the error
terms in presence of endogeneity bias and correlation between the
individual effects arises under Roy-type selection, i.e., when
individuals who experience a higher permanent effect of the
treatment are relatively more prone to accept the offer of
treatment. Wooldridge (2005) and Murtazashvile and Wooldridge (2005)
give sufficient conditions for consistency of standard OLS and IV
fixed effects estimators. These conditions amount to
$Cov[\epsilon_{it}, \upsilon_{it}]=0$ and $Cov[x_{it}, \alpha_{1i} |
\alpha_{i0}] = 0$.

Our proposal is to estimate the mean coefficient from separate
time series estimators for each individual. This strategy consists of
running OLS or IV for each individual, and then
estimating the population moment of interest by the corresponding
sample moment of the individual estimators. For example, the mean of the random
slope coefficient in the population is estimated by the sample average
of the OLS or IV slopes. These sample moments converge to the population moments of interest 
as number of individuals $n$ and time periods $T$ grow. However, since a different coefficient is estimated for each individual, the asymptotic distribution
of the sample moments can have asymptotic bias due to the incidental parameter
problem (Neyman and Scott, 1948).

To illustrate the nature of this bias, consider the estimator of the
mean coefficient $\alpha_{1}$ constructed from individual time series IV
 estimators. In this case the incidental parameter
problem is caused by the finite-sample bias of IV. This can be
explained using some expansions. Thus, assuming independence across
$t$, standard higher-order asymptotics gives (e.g. Rilstone et. al.,
1996), as $T \rightarrow \infty$
\begin{equation*}
\sqrt{T}(\hat{\alpha}_{1i}^{IV} - \alpha_{1i}) = \frac{1}{\sqrt{T}}
\sum_{t=1}^{T} \psi_{it}  + \frac{1}{\sqrt{T}} \beta_{i} +
o_{P}(T^{-1/2}),
\end{equation*}
where $\psi_{it} = E[\tilde z_{it} \tilde x_{it} \mid \alpha_i,
\pi_i]^{-1} \tilde z_{it} \epsilon_{it}$ is the influence function
of IV, $\beta_{i} = - E [\tilde z_{it} \tilde x_{it} \mid
\alpha_{i},\pi_{i}]^{-2} \linebreak E [\tilde z_{it}^{2} \tilde
x_{it} \epsilon_{it} \mid \alpha_{i},\pi_{i} ]$ is the higher-order
bias of IV  (see, e.g., Nagar, 1959, and Buse, 1992),  and the
variables with tilde are in deviation from their individual means,
e.g., $\tilde z_{it} = z_{it} - E[z_{it} \mid \alpha_{i},\pi_{i}]$.
In the previous expression the first order asymptotic distribution
of the individual estimator is centered at the truth since
$\sqrt{T}(\hat{\alpha}_{1i}^{IV} - \alpha_{1i}) \to_d N(0,
\sigma_{i}^{2})$ as $T \rightarrow \infty$, where $\sigma_{i}^{2} =
E[\tilde z_{it} \tilde x_{it} \mid \alpha_{i},\pi_{i}]^{-2} E[
\tilde z_{it}^{2} \epsilon_{it}^{2} \mid \alpha_{i},\pi_{i}].$

Let $\hat{\alpha}_1 = n^{-1} \sum_{i=1}^n \hat{\alpha}_{1i}^{IV}$, the sample average of the IV estimators.
The asymptotic distribution of $\hat{\alpha}_1$ is not centered around
$\alpha_1$ in short panels or more precisely under asymptotic
sequences where $T/\sqrt{n} \to 0$. To see this, consider the expansion for
$\hat{\alpha}_1$
\begin{equation*}
\sqrt{n}(\hat{\alpha}_1 - \alpha_1) = \frac{1}{\sqrt{n}}
\sum_{i=1}^{n} (\alpha_{1i} - \alpha_{1}) + \frac{1}{\sqrt{n}}
\sum_{i=1}^{n} (\hat{\alpha}_{1i}^{IV} - \alpha_{1i}).
\end{equation*}
The first term is the standard influence function for a sample mean
of known elements. The second term comes from the estimation of the
individual elements inside the sample mean.  Assuming independence
across $i$ and combining the previous expansions,
\begin{equation*}
\sqrt{n}(\hat{\alpha}_1 - \alpha_1) =  \underset{=
O_{P}(1)}{\underbrace{\frac{1}{\sqrt{n}} \sum_{i=1}^{n} (\alpha_{1i}
- \alpha_{1})}}  +
\underset{=O_{P}(1/\sqrt{T})}{\underbrace{\frac{1}{\sqrt{T}}
\frac{1}{\sqrt{nT}} \sum_{i=1}^{n} \sum_{t=1}^{T} \psi_{it}}} +
\underset{= O(\sqrt{n}/T)}{\underbrace{\frac{\sqrt{n}}{T}
\frac{1}{n} \sum_{i=1}^{n} \beta_{i}}} + o_{P}\left( 1 \right).
\end{equation*}
This expression shows that the bias term dominates the
asymptotic distribution of $\hat{\alpha}_1$ in short panels under
sequences where $T/\sqrt{n} \to 0$. Averaging reduces the order of
the variance of $\hat{\alpha}_{1i}^{IV}$, without affecting the
order of its bias. In this case the  estimation of the
random coefficients has no first order effect in the asymptotic variance of
$\hat{\alpha}_1$ because the second term is of smaller order than the first term.

A potential drawback of the individual by individual time series estimation
 is that it might more be sensitive to weak identification
problems than fixed coefficient pooled estimation.\footnote{We thank a referee for pointing out this issue.} In the random
coefficient model, for example, we require that $E[\tilde z_{it}
\tilde x_{it} \mid \alpha_{i},\pi_{i}] = \pi_{1i} \neq 0$ with probability one,  i.e., for all
the individuals, whereas fixed
coefficient IV only requires that this condition holds on average,
i.e., $E[\pi_{1i}] \neq 0$. The individual estimators are therefore
more sensitive than traditional pooled estimators to weak
instruments problems. On the other hand, individual by individual
estimation relaxes the exogeneity condition  by conditioning on
additive and non-additive time invariant heterogeneity, i.e, $E[\tilde z_{it}
\epsilon_{it} \mid \alpha_{i},\pi_{i}] = 0$. Traditional
fixed effects estimators only condition on additive time invariant
heterogeneity. A formal treatment of these identification issues is
beyond the scope of this paper.

\subsection{Variance of individual coefficients}
Consider the panel model:
\begin{equation*}
y_{it} = \alpha_i + \epsilon_{it}, \ \epsilon_{it} \mid \alpha_i \sim (0,
\sigma_{\epsilon}^2), \ \alpha_i \sim (\alpha, \sigma_{\alpha}^2), \
(t = 1, ..., T; i = 1, ..., n);
\end{equation*}
where $y_{it}$ is an outcome variable of interest, which can be
decomposed in an individual effect $\alpha_i$ with mean $\alpha$ and
variance $\sigma_{\alpha}^2$, and an error term $\epsilon_{it}$ with
zero mean and variance $\sigma_{\epsilon}^2$ conditional on
$\alpha_i$. The parameter of interest is $\sigma_{\alpha}^2 = Var[\alpha_i]$ and its
fixed effects estimator is
\begin{equation*}
\hat{\sigma}_{\alpha}^2 = (n-1)^{-1} \sum_{i=1}^n (\hat{\alpha}_i -
\hat{\alpha})^2,
\end{equation*}
where $\hat{\alpha}_i = T^{-1} \sum_{t=1}^T y_{it}$ and
$\hat{\alpha} = n^{-1} \sum_{i=1}^n \hat{\alpha}_i$.

Let $\varphi_{\alpha_i}  = (\alpha_i - \alpha)^2 -
\sigma_{\alpha}^2$ and $\varphi_{\epsilon_{it}} = \epsilon_{it}^2 -
\sigma_{\epsilon}^2$. Assuming independence across $i$ and $t$, a
standard asymptotic expansion gives, as $n,T \to \infty$,
\begin{equation*}
\sqrt{n} (\hat{\sigma}_{\alpha}^2 - \sigma_{\alpha}^2) =
\underset{=O_P(1)}{\underbrace{\frac{1}{\sqrt{n}} \sum_{i=1}^n
\varphi_{\alpha_i}}} +
\underset{=O_P(1/\sqrt{T})}{\underbrace{\frac{1}{\sqrt{T}}
\frac{1}{\sqrt{nT}} \sum_{i=1}^n \sum_{t=1}^T
\varphi_{\epsilon_{it}}}} +
\underset{=O(\sqrt{n}/T)}{\underbrace{\frac{\sqrt{n}}{T}
\sigma_{\epsilon}^2}} + o_P ( 1 ).
\end{equation*}
The first term corresponds to the influence function of the sample
variance if the $\alpha_i$'s were known. The second term comes from
the estimation of the $\alpha_i$'s. The third term is a bias term that comes from the
nonlinearity of the variance in $\hat{\alpha}_i$. The bias term dominates
the expansion in short panels under sequences where $T/\sqrt{n} \to
0$. As in the previous example, the estimation of the $\alpha_i$'s
has no first order affect in the asymptotic variance since the second
term is of smaller order than the first term.


\section{The Model and Estimators} \label{s3}
We consider a general model with a finite number of moment
conditions $d_g$. To describe it, let the data be denoted by $z_{it}$
$(i = 1, \ldots, n; t = 1, \ldots, T)$. We assume that $z_{it}$ is
independent over $i$ and stationary and strongly mixing over $t$.
Also, let $\theta$ be a $d_{\theta}$--vector of common parameters,
$\{ \alpha_{i} :  1 \leq i \leq n\}$ be a sequence of $d_{\alpha}$--vectors
with the realizations of the individual effects, and $g(z;\theta,\alpha_{i})$ be an $d_g$--vector of functions, where $d_g \geq d_{\theta}+d_{\alpha}$.\footnote{We impose that some
of the parameters are common for all the individuals to help
preserve degrees of freedom in estimation of short panels with many
regressors. An order condition for this model is that the
number of individual specific parameters $d_{\alpha}$ has to be less than the
time dimension $T$.} The model has true parameters $\theta_{0}$ and
$\{ \alpha_{i0} : 1 \leq i \leq n\}$, satisfying the  moment conditions
\begin{equation*}
E \left[ g(z_{it};\theta_{0},\alpha_{i0}) \right] = 0,   \  (t = 1, ..., T; i = 1,
...,n),
\end{equation*}
where $E[ \cdot ]$ denotes conditional expectation with respect to
 the distribution of $z_{it}$ conditional on the individual effects.

Let  $\bar E[ \cdot ]$ denote the expectation taken with respect to
the distribution of the individual effects. In the previous model,
the ultimate quantities of interest are smooth functions of
parameters and observations, which in some cases could be the
parameters themselves,
\begin{equation*}
\zeta = \bar E  E[ \zeta_i(z_{it}; \theta_{0}, \alpha_{i0})],
\end{equation*}
if $\bar E E |\zeta_i(z_{it}; \theta_{0}, \alpha_{i0})| < \infty$,
or moments or other smooth functions of the individual effects
\begin{equation*}
\mu = \bar E [ \mu(\alpha_{i0}) ],
\end{equation*}
if $\bar E | \mu(\alpha_{i0}) | < \infty$. In the correlated random
coefficient example, $g(z_{it};\theta_{0},\alpha_{i0}) =
z_{it}(y_{it} - \alpha_{0i0} - \alpha_{1i0}x_{it})$, $\theta =
\emptyset$, $d_{\theta} = 0$, $d_{\alpha} = 2$, and $\mu(\alpha_{i0}) = \alpha_{1i0}$.
In the variance of the random coefficients example,
$g(z_{it};\theta_{0},\alpha_{i0}) = (y_{it} - \alpha_{0i0})$,
$\theta = \emptyset$, $d_{\theta} = 0$, $d_{\alpha} = 1 $ , and $\mu(\alpha_{i0}) =
(\alpha_{1i0} - \bar E [\alpha_{1i0}])^2$.

Some more notation, which will be extensively used in the
definition of the estimators and in the analysis of their
asymptotic properties, is the following
\begin{eqnarray*}
\Omega_{ji}(\theta, \alpha_{i}) &:=& E
[g(z_{it};\theta,\alpha_{i})g(z_{i,t-j};\theta,\alpha_{i})'], \ \ j \in \{0, 1, 2, ... \}, \\
G_{\theta_{i}} (\theta, \alpha_{i}) &:=& E [G_{\theta} (z_{it};
\theta, \alpha_i)] = E \ [ \partial
g(z_{it};\theta,\alpha_{i}) / \partial \theta'], \\
G_{\alpha_{i}} (\theta, \alpha_i) &:=& E [G_{\alpha} (z_{it};\theta,
\alpha_i)] = E \ [ \partial g(z_{it};\theta,\alpha_{i})/ \partial
\alpha_{i}'],
\end{eqnarray*}
where superscript $'$ denotes transpose and higher-order derivatives
will be denoted by adding subscripts. Here $\Omega_{ji}$ is the
covariance matrix between the moment conditions for individual
$i$ at times $t$ and $t-j$, and $G_{\theta_{i}}$ and $G_{\alpha_{i}}$ are
time series average derivatives of these conditions. Analogously, for
sample moments
\begin{eqnarray*}
\hat{\Omega}_{ji}(\theta, \alpha_{i}) &:=& T^{-1} \sum_{t=j+1}^{T}
g(z_{it};\theta,\alpha_{i})g(z_{i,t-j};\theta,\alpha_{i})', \ \ j \in \{0, 1, ..., T-1\}, \\
\hat{G}_{\theta_{i}} (\theta, \alpha_{i}) &:=& T^{-1} \sum_{t=1}^{T}
G_{\theta} (z_{it}; \theta, \alpha_i) = T^{-1} \sum_{t=1}^{T}
\partial g(z_{it};\theta,\alpha_{i})/\partial
\theta', \\
\hat{G}_{\alpha_{i}} (\theta, \alpha_i) &:=& T^{-1} \sum_{t=1}^{T}
G_{\alpha} (z_{it};\theta, \alpha_i) = T^{-1} \sum_{t=1}^{T}
\partial g(z_{it};\theta,\alpha_{i})/\partial \alpha_{i}'.
\end{eqnarray*}
In the sequel, the arguments of the expressions will be omitted
when the functions are evaluated at the true parameter values
$(\theta_{0}', \alpha_{i0}')'$, e.g., $g(z_{it})$ means $g(z_{it};
\theta_{0}, \alpha_{i0})$.

In cross-section and time series models, parameters defined from
moment conditions are usually estimated using the two-step GMM
estimator of Hansen (1982). To describe how to adapt this method to
panel models with fixed effects, let $\hat{g}_{i}(\theta,
\alpha_{i}) := T^{-1} \sum_{t=1}^{T} g(z_{it};\theta, \alpha_{i})$,
and let $(\tilde{\theta}',\{ \tilde{\alpha}_{i}' \}_{i = 1}^{n})'$
be some preliminary one-step FE-GMM estimator, given by
$(\tilde{\theta}',\{\tilde{\alpha}_{i}' \}_{i = 1}^{n})' =$
\linebreak $\arg \inf_{\{(\theta', \alpha_{i}' )' \in \Upsilon\}_{i
= 1}^{n} } \sum_{i=1}^{n} \hat{g}_{i}(\theta, \alpha_{i})'$ $
\hat{W}_{i}^{-1} $ $\hat{g}_{i}(\theta, \alpha_{i})$, where $\Upsilon \subset \mathbb{R}^{d_{\theta}+d_{\alpha}}$
denotes the parameter space, and $\{\hat{W}_{i} : 1 \leq i \leq n$\} is a
sequence of positive definite symmetric $d_g \times d_g$ weighting
matrices. The two-step FE-GMM estimator is the solution
to the following program
\begin{equation*}
(\hat{\theta}',\{\hat{\alpha}_{i}' \}_{i = 1}^{n})' = \text{arg}
\inf_{\{(\theta', \alpha_{i}' )' \in \Upsilon \}_{i = 1}^{n} }
\sum_{i=1}^{n} \hat{g}_{i}(\theta, \alpha_{i})'
\hat{\Omega}_{i}(\tilde{\theta},\tilde{\alpha_{i}})^{-1}
\hat{g}_{i}(\theta, \alpha_{i}),
\end{equation*}
where $\hat{\Omega}_{i}(\tilde{\theta},\tilde{\alpha_{i}})$ is an
estimator of the optimal weighting matrix  for
individual $i$
$$
\Omega_{i} = \Omega_{0i} + \sum_{j = 1}^{\infty}  (\Omega_{ji} + \Omega_{ji}').
$$
 To facilitate the asymptotic analysis, in the estimation of the optimal weighting matrix we assume that 
$g(z_{it};\theta_{0},\alpha_{i0})$ is a
martingale difference sequence with respect to the sigma algebra
 $\sigma(\alpha_{i}, z_{i,t-1}, z_{i,t-2}, ...)$, so that $
\Omega_{i} = \Omega_{0i}$ and $\hat{\Omega}_{i}(\tilde{\theta},\tilde{\alpha_{i}}) = \hat{\Omega}_{0i}(\tilde{\theta},\tilde{\alpha_{i}}) $. This
assumption holds in rational expectation models. We do not impose this assumption to derive the limiting distribution of the one-step 
FE-GMM estimator. 


For the subsequent analysis of the asymptotic properties of the
estimator, it is convenient to consider the concentrated or profile
problem. This problem is a two-step procedure. In the first step the
program is solved for the individual effects, given the value of the
common parameter $\theta$. The First Order Conditions (FOC) for this
stage, reparametrized conveniently as in Newey and Smith (2004), are
the following
\begin{eqnarray*}
\hat{t}_{i}(\theta, \hat{\gamma}_{i}(\theta)) = - \left(%
\begin{array}{c}
  \hat{G}_{\alpha_{i}} (\theta, \hat{\alpha}_{i}(\theta))'
\hat{\lambda}_{i}(\theta) \\
  \hat{g}_{i}(\theta,\hat{\alpha}_{i}(\theta)) +
\hat{\Omega}_{i}(\tilde{\theta},\tilde{\alpha}_{i})
\hat{\lambda}_{i}(\theta) \\
\end{array}%
\right) = 0, \text{   }( i = 1, ..., n),
\end{eqnarray*}
where $\lambda_{i}$ is a $d_g$--vector of individual Lagrange
multipliers for the moment conditions, and $\gamma_{i} :=
(\alpha_{i}', \lambda_{i}')'$ is an extended $(d_{\alpha}+d_g)$--vector
of individual effects. Then, the solutions to the previous equations
are plugged into the original problem, leading to the following
first order conditions for $\theta$, $\hat{s}(\hat{\theta}) = 0$, where
\begin{equation*}
\hat{s}(\theta) = n^{-1} \sum_{i=1}^{n}
\hat{s}_{i}(\theta, \hat{\gamma}_{i}(\theta)) = -
n^{-1} \sum_{i=1}^{n} \hat{G}_{\theta_{i}} ( \theta,
\hat{\alpha}_{i}(\theta))' \hat{\lambda}_{i} (\theta),
\end{equation*}
is the profile score function for $\theta$.\footnote{In
the original parametrization, the FOC can be written as
\begin{equation*}
n^{-1} \sum_{i=1}^{n} \hat{G}_{\theta_{i}} ( \hat{\theta},
\hat{\alpha}_{i}(\hat{\theta}))'
\hat{\Omega}_{i}(\tilde{\theta},\tilde{\alpha}_{i})^{-}
\hat{g}_{i}(\theta,\hat{\alpha}_{i}(\theta))  = 0,
\end{equation*}
where the superscript $^{-}$ denotes a generalized inverse.} 

Fixed effects estimators of smooth functions of parameters and observations are constructed using the
plug-in principle, i.e. $\hat \zeta = \hat \zeta(\hat \theta)$ where
$$
\hat \zeta(\theta) = (nT)^{-1} \sum_{i=1}^n \sum_{t=1}^T \zeta(z_{it}; \theta, \hat \alpha_i(\theta)).
$$ 
Similarly, moments of the individual effects are estimated by $\hat \mu = \hat \mu (\hat \theta),$ where
$$
\hat \mu(\theta) = n^{-1} \sum_{i=1}^n \mu(\hat \alpha_i(\theta)). 
$$


\section{Asymptotic Theory for FE-GMM Estimators}
\label{s4}
In this section we analyze the properties of one-step and two-step
FE-GMM estimators in large samples. 
We show
consistency and derive the asymptotic distributions for estimators
of individual effects, common parameters and other quantities of interest under sequences where
both $n$ and $T$ pass to infinity with the sample size. We establish results separately for one-step and two-step
estimators because the former are derived under less restrictive assumptions.

We make the following assumptions to show uniform consistency of
the FE-GMM one-step estimator:
\begin{condition}[Sampling and asymptotics]
\label{cond1} (i)  For each $i$,  conditional on $\alpha_i,$ $z_{i} :=
\left\{z_{it}: 1 \leq t \leq T\right\}$ is a stationary mixing sequence
of random vectors with strong mixing coefficients $a_{i}(l) =
\sup_{t} \sup_{A\in\mathcal{A}_{t}^{i}, D\in\mathcal{D}_{t+l}^{i}}$ $
\left| P(A \cap D) - P(A)P(D)\right|$, where $\mathcal{A}_{t}^{i} = \sigma(\alpha_{i}, z_{it}, z_{i,t-1}, ... ) $ and $\mathcal{D}%
_{t}^{i} = \sigma(\alpha_{i}, z_{it}, z_{i,t+1}, ... ),$ such that  $\sup_{i} \left|
a_{i}(l) \right| \leq C a^{l}$ for some $0 < a < 1$
and some $C > 0$; (ii) $ \left\{(z_{i}, \alpha_i): 1 \leq i \leq n\right\}$ are
independent and identically distributed across $i$; (iii) $n,T \to \infty $ such that $n/T \to \kappa^2$, where $0 < \kappa^2 < \infty$; and (iv) $\dim \left[
g(\cdot; \theta, \alpha_{i}) \right] = d_g < \infty$.
\end{condition}

For a matrix or vector $A$, let $|A|$ denote the Euclidean norm, that is $|A|^{2} = trace [A A']$.
\begin{condition}[Regularity and identification]
\label{cond2}(i) The vector of moment functions $g ( \cdot ;\theta
,\alpha) = (g_{1} \left( \cdot ;\theta ,\alpha \right), ..., g_{d_g}
\left( \cdot ;\theta ,\alpha \right))' $ is continuous in $\left(
\theta ,\alpha \right) \in \Upsilon $; (ii) the parameter space
$\Upsilon $ is a compact, convex subset of $\mathbb{R}^{d_{\theta}+d_{\alpha}}$;
(iii) $dim\left( \theta ,\alpha \right) = d_{\theta}+d_{\alpha} \leq d_g$; (iv) there exists a function $%
M\left( z_{it}\right) $ such that $\left| g_{k} \left( z_{it};\theta
,\alpha _{i}\right) \right| \leq M\left( z_{it}\right) $, $\left|
\partial g_{k} \left( z_{it};\theta ,\alpha _{i}\right) / \partial
\left( \theta
,\alpha _{i}\right) \right| \leq M\left( z_{it}\right) $, for $k = 1, ...,d_g$, and $\sup_{i}E%
\left[ M\left( z_{it}\right) ^{4+\delta}\right] <\infty $ for some $\delta > 0$; and (v) there exists  a deterministic sequence of symmetric finite positive definite matrices $\{W_{i}: 1 \leq i \leq n \}$ such that
$\sup_{1 \leq i \leq n} |\hat W_{i} - W_i | \to_P 0,$ 
and, for each $\eta
>0$ 
\begin{equation*}
\inf_{i}\left[ Q_{i}^{W} \left( \theta _{0},\alpha _{i0}\right)
-\sup_{\left\{ \left( \theta ,\alpha \right) :\left| \left( \theta
,\alpha \right) -\left( \theta _{0},\alpha _{i0}\right) \right|
>\eta \right\} } Q_{i}^{W} \left( \theta ,\alpha \right)
\right] >0,
\end{equation*}
where
\begin{equation*}
Q_{i}^{W} \left( \theta ,\alpha _{i}\right) := - g_{i} \left( \theta
,\alpha _{i}\right)' W_{i}^{-1} g_{i} \left( \theta ,\alpha
_{i}\right), \ \ g_{i} \left( \theta ,\alpha _{i}\right) := E \left[
\hat{g}_{i} \left( \theta ,\alpha _{i}\right) \right].
\end{equation*}
\end{condition}

Conditions \ref{cond1}(i)-(ii) impose cross sectional independence, but
allow for weak time series dependence as in Hahn and Kuersteiner
(2011). Conditions \ref{cond1}(iii)-(iv) describe the asymptotic
sequences that we consider where $T$ and $n$ grow at the same rate with the sample size,
whereas the number of moments $d_g$ is fixed. Condition \ref{cond2} adapts standard 
assumptions of the GMM literature to guarantee
the identification of the parameters based on time series variation
for all the individuals, see Newey and McFadden (1994). The
dominance and moment conditions in \ref{cond2}(iv) 
are used to establish uniform consistency of the estimators of the
individual effects. 

\begin{theorem}[Uniform consistency of one-step estimators] \label{th1} Suppose that Conditions
\ref{cond1} and \ref{cond2} hold. Then, for any
$\eta > 0$
\begin{equation*}
\Pr \left( \left| \tilde{\theta} - \theta_{0} \right| \geq \eta
\right) = o(T^{-1}),
\end{equation*}
where $\tilde{\theta} = \arg \max_{ \left\{ \left(\theta,\alpha_{i}
\right) \in \Upsilon \right\}_{i=1}^{n}
 } \frac{1}{n} \sum_{i=1}^{n}
\hat{Q}_{i}^{W} (\theta, \alpha_{i})$  and $\hat{Q}_{i}^{W} \left(
\theta ,\alpha _{i}\right) :=  - \hat{g}_{i} \left( \theta ,\alpha
_{i}\right)' \hat W_{i}^{-1} \hat{g}_{i} \left( \theta ,\alpha
_{i}\right)$. Also, for any $\eta > 0$
\begin{equation*}
\Pr \left( \sup_{1 \leq i \leq n} \left| \tilde{\alpha}_{i} -
\alpha_{i0} \right| \geq \eta \right) = o\left( T^{-1} \right) \text{ and } \Pr \left( \sup_{1 \leq i \leq n} \left| \tilde{\lambda}_{i} \right| \geq \eta \right)
= o\left( T^{-1} \right),
\end{equation*}
where $\tilde{\alpha}_{i} = \arg \max_{\alpha}
\hat{Q}_{i}^{W}(\tilde{\theta},\alpha)$ and $ \tilde{\lambda}_{i} = -
\hat W_{i}^{-1} \hat{g}_{i}(\tilde{\theta},\tilde{\alpha}_{i})$.
\end{theorem}

Let $\Sigma_{\alpha_{i}}^{W} := \left( G_{\alpha_{i}}' W_{i}^{-1}
G_{\alpha_{i}} \right)^{-1}$, $H_{\alpha_{i}}^{W} :=
\Sigma_{\alpha_{i}}^{W} G_{\alpha_{i}}' W_{i}^{-1}$,
$P_{\alpha_{i}}^{W} := W_{i}^{-1} - W_{i}^{-1} G_{\alpha_{i}}
H_{\alpha_{i}}^{W}$, $J_{si}^{W} := G_{\theta_{i}}'
P_{\alpha_{i}}^{W} G_{\theta_{i}}$ and $J^{W}_{s} := \bar{E}
[J_{si}^{W}]$. We use the following additional assumptions to derive
the limiting distribution of the one-step estimator:
\begin{condition}[Regularity]
\label{cond3} (i) For each $i$, $(\theta_{0}, \alpha_{i0}) \in int
\left[ \Upsilon \right]$;  and (ii)  $J_{s}^{W}$ is finite positive definite, and $\{
G_{\alpha_{i}}' W_{i}^{-1} G_{\alpha_{i}}: 1 \leq i \leq n \}$ is a sequence of finite positive definite matrices, where $\{W_i : 1 \leq i \leq n \}$ is the sequence of matrices of  Condition \ref{cond2}(v).
\end{condition}

\begin{condition}[Smoothness]
\label{cond5} (i) There exists a function $M\left( z_{it}\right) $ such
that, for $k = 1, ...,d_g,$
\begin{equation*}
\left| \partial ^{d_{1} + d_{2}} g_{k} \left( z_{it};\theta ,\alpha
_{i}\right) / \partial \theta ^{d_{1}}\partial \alpha
_{i}^{d_{2}}\right| \leq M\left( z_{it}\right), \qquad 0\leq
d_{1}+d_{2}\leq 1,\ldots ,5,
\end{equation*}
and $\sup_{i}E\left[ M\left( z_{it}\right) ^{5(d_{\theta}+d_{\alpha}+6)/(1-10v)+\delta}\right] <\infty,$ for some $\delta > 0$ and $0 < v < 1/10;$ and (ii) there exists $\xi_{i}(z_{it})$ such that
$\hat W_i = W_i + \sum_{t=1}^T \xi_i(z_{it})/T + R_i^W/T,$ where $max_{i} |R_{i}^W| = o_P(T^{1/2}),$  $E[\xi_i(z_{it})] = 0,$ and $\sup_{i} E[|\xi_i(z_{it})|^{20/(1-10v)+\delta}] <\infty,$ for some $\delta > 0$ and $0 < v < 1/10.$ 
\end{condition}
Condition \ref{cond3} is the panel data analog to the standard asymptotic 
normality condition for GMM with cross sectional data,
see  Newey and McFadden (1994).  Condition \ref{cond5} is similar to Condition 4 in
Hahn and Kuersteiner (2011), and guarantees
the existence of higher order expansions for the GMM estimators and
the uniform convergence of their remainder terms.

Let $G_{\alpha \alpha_i} := (G_{\alpha \alpha_{i, 1}}', \ldots,
G_{\alpha \alpha_{i,q}}')',$ where $G_{\alpha \alpha_{i, j}} = E [
\partial G_{\alpha_i}(z_{it}) / \partial \alpha_{i,j}],$
and $G_{\theta \alpha_i} := (G_{\theta \alpha_{i, 1}}', \ldots,
G_{\theta \alpha_{i,q}}')',$ where $G_{\theta \alpha_{i, j}} = E [
\partial G_{\theta_i}(z_{it}) / \partial \alpha_{i,j}]$. The symbol
$\otimes$ denotes kronecker product of matrices,  $I_{d_{\alpha}}$ a $d_{\alpha} \times
d_{\alpha}$ identity matrix, $e_{j}$ a unitary $d_g$--vector with 1 in
row $j$, and $P_{\alpha_{i},j}^{W}$ the $j$-th column of
$P_{\alpha_{i}}^{W}$. Recall that the extended individual effect is
$\gamma_i = (\alpha_i', \lambda_i')'$.

\begin{lemma} [Asymptotic expansion for one-step
estimators of individual effects] \label{lemma1} Under
Conditions \ref{cond1}, \ref{cond2}, \ref{cond3}, and \ref{cond5},
\begin{equation}
\sqrt{T} (\tilde{\gamma}_{i0} - \gamma_{i0}) = \tilde{\psi}_{i}^{W}
+ T^{-1/2} Q_{1i}^{W} + T^{-1} R_{2i}^{W},
\end{equation}
where $\tilde{\gamma}_{i0} := \tilde{\gamma}_{i}(\theta_{0})$,
$$
 \tilde{\psi}_{i}^W = -  \left(
                          \begin{array}{cc}
                            H_{\alpha_i}^W \\
                            P_{\alpha_i}^W \\
                          \end{array}
                        \right) T^{-1/2} \sum_{t=1}^T g(z_{it}) \ind N(0, V_i^{W}),
$$
$n^{-1/2} \sum_{i=1}^{n} \tilde{\psi}_{i}^{W} \ind N(0,
\bar{E}[V_i^{W}]),$ $n^{-1} \sum_{i=1}^n Q_{1i}^{W} \inp \bar{E} [B_{\gamma_{i}}^{W}]$, $B_{\gamma_{i}}^{W} = B_{\gamma_{i}}^{W,I} +
B_{\gamma_{i}}^{W,G} + B_{\gamma_{i}}^{W,1S}$, $\sup_{1 \leq i \leq n} R_{2i}^{W} = o_{P}(\sqrt{T})$, for
\begin{footnotesize}
\begin{eqnarray*}
V_i^{W} &=& \left(%
\begin{array}{cc}
  H_{\alpha_{i}}^{W}  \\
  P_{\alpha_{i}}^{W}  \\
\end{array}%
\right) \Omega_{i} \left(H_{\alpha_{i}}^{W'}, P_{\alpha_{i}}^{W}\right),
 \\
B_{\gamma_{i}}^{W,I} &=& \left(%
\begin{array}{c}
  B_{\alpha_{i}}^{W,I} \\
  B_{\lambda_{i}}^{W,I} \\
\end{array}%
\right) = \left(%
\begin{array}{c}
  H_{\alpha_{i}}^{W}  \\
  P_{\alpha_{i}}^{W}  \\
\end{array}%
\right) \left(\sum_{j=-\infty}^{\infty} E \left[ G_{\alpha_{i}}(z_{it}) H_{\alpha_{i}}^{W} g(z_{i,t-j}) \right] -    \sum_{j=1}^{d_{\alpha}} G_{\alpha \alpha_{i,j}} H_{\alpha_{i}}^{W} \Omega_{i} H_{\alpha_{i}}^{W'}/2\right), \\
B_{\gamma_{i}}^{W,G} &=& \left(%
\begin{array}{c}
  B_{\alpha_{i}}^{W,G} \\
  B_{\lambda_{i}}^{W,G} \\
\end{array}%
\right) = \left(%
\begin{array}{c}
  -\Sigma_{\alpha_{i}}^{W}
      \\
  H_{\alpha_{i}}^{W'}\\
\end{array}%
\right)\sum_{j=-\infty}^{\infty}  E \left[ G_{\alpha_{i}}(z_{it})' P_{\alpha_{i}}^{W} g(z_{i,t-j})\right], \\
B_{\gamma_{i}}^{W,1S} &=& \left(%
\begin{array}{c}
  B_{\alpha_{i}}^{W,1S} \\
  B_{\lambda_{i}}^{W,1S} \\
\end{array}%
\right) = \left(%
\begin{array}{c}
  \Sigma_{\alpha_{i}}^{W} \\
  -  H_{\alpha_{i}}^{W'}  \\
\end{array}%
\right) \left(\sum_{j=1}^{d_{\alpha}} G_{\alpha \alpha_{i,j}}' P_{\alpha_{i}}^{W} \Omega_{i}
  H_{\alpha_{i}}^{W'}/2
  +    \sum_{j=1}^{d_g}
   G_{\alpha \alpha_{i}}' (I_{d_{\alpha}} \otimes e_{j}) H_{\alpha_{i}}^{W} \Omega_{i} P_{\alpha_{i},j}^{W}/2\right),    \\ 
& & \quad  \quad \quad \quad \quad \quad + 
\left(%
\begin{array}{c}
   H_{\alpha_{i}}^{W}  \\
  P_{\alpha_{i}}^{W}   \\
\end{array}%
\right)\sum_{j=-\infty}^{\infty}  E \left[ \xi_i(z_{it}) P_{\alpha_{i}}^{W} g(z_{i,t-j}) \right]. 
\end{eqnarray*}
\end{footnotesize}
\end{lemma}

\smallskip

\begin{theorem} [Limit distribution of one-step estimators
of common parameters] \label{th3} Under Conditions \ref{cond1},
\ref{cond2}, \ref{cond3} and \ref{cond5}, 
$$
\sqrt{nT}(\tilde{\theta}-\theta_{0}) 
\ind - (J_{s}^{W})^{-1} N \left( \kappa B_{s}^{W}, V_{s}^{W} \right),
$$
where
\begin{equation*}
J_{s}^{W} = \bar{E} \left[ G_{\theta_{i}}'P_{\alpha_{i}}^{W}
G_{\theta_{i}} \right], V_{s}^{W} = \bar{E} \left[
G_{\theta_{i}}'P_{\alpha_{i}}^{W} \Omega_{i} P_{\alpha_{i}}^{W}
G_{\theta_{i}} \right],  B_{s}^{W} = \bar{E} \left[ B_{si}^{W,B} +
B_{si}^{W,C} + B_{si}^{W,V} \right],
\end{equation*}
and
\begin{eqnarray*}
\begin{array}{ll}
&B_{si}^{W,B} = -  G_{\theta_{i}}' \left( B_{\lambda_{i}}^{W,I} +
B_{\lambda_{i}}^{W,G} +  B_{\lambda_{i}}^{W,1S} \right),
B_{si}^{W,C} = \sum_{j=-\infty}^{\infty}  E[G_{\theta_{i}}(z_{it})'
P_{\alpha_{i}}^{W}
g_{i}(z_{i,t-j})],   \\
&B_{si}^{W,V} = - \sum_{j=1}^{d_{\alpha}} G_{\theta \alpha_{i,j}}'
P_{\alpha_{i}}^{W} \Omega_{i} H_{\alpha_{i}}^{W'}/2 - \sum_{j=1}^{d_g}
 G_{\theta \alpha_{i}}' (I_{d_{\alpha}} \otimes  e_{j}) H_{\alpha_{i}}^{W} \Omega_{i}
P_{\alpha_{i},j}/2.
\end{array}
\end{eqnarray*}
The expressions for $B_{\lambda_{i}}^{W,I}$,
$B_{\lambda_{i}}^{W,G}$, and $B_{\lambda_{i}}^{W,1S}$ are given in
Lemma \ref{lemma1}.
\end{theorem}

The source of the bias is the non-zero expectation of the profile score of $\theta$ at the true parameter value, due to the
substitution of the unobserved individual effects by sample
estimators. These estimators converge to their true parameter value at
a rate $\sqrt{T},$ which is slower than $\sqrt{nT}$, the rate of
convergence of the estimator of the common parameter. Intuitively, the rate for $\widetilde \gamma_{i0}$ is $\sqrt{T}$ because only the $T$ observations for individual $i$ convey information about $\gamma_{i0}$. In nonlinear
and dynamic models, the slow convergence of the estimator of the individual
effect introduces bias in the estimators of the rest of parameters. The expression of this bias
 can be explained with an expansion of the score around the true
value of the individual effects\footnote{Using the notation
introduced in Section \ref{s3}, the  score is
\begin{equation*}
\hat{s}^{W}(\theta_{0}) = n^{-1} \sum_{i=1}^{n}
\hat{s}_{i}^{W}(\theta_{0}, \tilde{\gamma}_{i0}) = - n^{-1}
\sum_{i=1}^{n} \hat{G}_{\theta_{i}} ( \theta_{0},
\tilde{\alpha}_{i0})' \tilde{\lambda}_{i0},
\end{equation*}
where $\tilde{\gamma}_{i0} = (\tilde{\alpha}_{i0}',
\tilde{\lambda}_{i0}')$ is the solution to
\begin{eqnarray*}
\hat{t}_{i}^{W}(\theta_{0}, \tilde{\gamma}_{i0}) = - \left(%
\begin{array}{c}
  \hat{G}_{\alpha_{i}} (\theta_{0}, \tilde{\alpha}_{i0})'
\tilde{\lambda}_{i0} \\
  \hat{g}_{i}(\theta_{0},\tilde{\alpha}_{i0}) +
W_{i} \tilde{\lambda}_{i0} \\
\end{array}%
\right) = 0.
\end{eqnarray*}
 }
\begin{eqnarray*}
E \left[ \hat{s}_{i}^{W} (\theta_{0}, \tilde{\gamma}_{i0}) \right]
&=& E \left[ \hat{s}_{i}^{W} \right] +  E \left[ \hat{s}_{\gamma
i}^{W} \right]' E \left[ \tilde{\gamma}_{i0} - \gamma_{i0} \right] +
E \left[ (\hat{s}_{\gamma i}^{W} - E \left[ \hat{s}_{\gamma i}^{W}
\right]
)' (\tilde{\gamma}_{i0} - \gamma_{i0}) \right] \notag\\
&+&  E \left[ \sum_{j=1}^{d_{\alpha}+d_g} (\tilde{\gamma}_{i0,j} -
\gamma_{i0,j}) E \left[ \hat{s}_{\gamma \gamma i}^{W} \right]
(\tilde{\gamma}_{i0} - \gamma_{i0})
\right]/2 + o(T^{-1}) \notag\\
&=&  0 +  B_{s}^{W,B}/T + B_{s}^{W,C}/T + B_{s}^{W,V}/T + o(T^{-1}).
\end{eqnarray*}
This expression shows that the bias has the same three components as in the
MLE case, see Hahn and Newey (2004). The first component,
$B_{s}^{W,B}$, comes from the higher-order bias of the estimator of
the individual effects. The second component, $B_{s}^{W,C}$, is a
correlation term and is present because individual effects and
common parameters are estimated using the same observations. The
third component, $B_{s}^{W,V}$, is a variance term. The bias of the
individual effects, $B_{s}^{W,B}$, can be further decomposed in
three terms corresponding to the asymptotic bias for a GMM estimator
with the optimal score, $B_{\lambda}^{W,I}$, when $W$ is used as the
weighting function; the bias arising from estimation of
$G_{\alpha_{i}}$, $B_{\lambda}^{W,G}$; and the bias arising from not
using an optimal weighting matrix, $B_{\lambda}^{W,1S}$.

We use the following condition to show the consistency of
the two-step FE-GMM estimator: 
\begin{condition}[Smoothness, regularity, and martingale]
\label{cond4}(i) There exists a function $%
M\left( z_{it}\right) $ such that $\left| g_{k} \left( z_{it};\theta
,\alpha _{i}\right) \right| \leq M\left( z_{it}\right) $, $\left|
\partial g_{k} \left( z_{it};\theta ,\alpha _{i}\right) / \partial
\left( \theta
,\alpha _{i}\right) \right| \leq M\left( z_{it}\right) $, for $k = 1, ...,d_g$, and $\sup_{i}E%
\left[ M\left( z_{it}\right) ^{10(d_{\theta}+d_{\alpha}+6)/(1-10v)+\delta}\right] <\infty,$ for some $\delta > 0$ and $0 < v < 1/10;$  (ii) $\{
\Omega_{i}:  1 \leq i \leq n \}$ is a sequence of finite positive definite
matrices; and (iii)  for each $i$, $g (z_{it} ;\theta_0 ,\alpha_{i0})$ is a martingale
difference sequence with respect to
$\sigma(\alpha_{i}, z_{i,t-1}, z_{i,t-2}, \ldots)$.
\end{condition}

Conditions \ref{cond4}(i)-(ii) are used to establish the uniform
consistency of the estimators of the individual weighting matrices. Condition \ref{cond4}(iii) is convenient to simplify the
expressions of the optimal weighting matrices. It holds, for example, in rational expectation models
that commonly arise in economic applications.


\begin{theorem} [Uniform consistency of two-step estimators]
\label{th4}Suppose that Conditions \ref{cond1}, \ref{cond2},
\ref{cond3} and \ref{cond4} hold. Then, for any $\eta > 0$
\begin{equation*}
\Pr \left( \left| \hat{\theta} - \theta_{0} \right| \geq \eta
\right) = o \left( T^{-1} \right),
\end{equation*}
where $\hat{\theta} = \arg \max_{ \{ (\theta',\alpha_{i}')
\}_{i=1}^{n} \in \Upsilon}  \sum_{i=1}^{n}
\hat{Q}_{i}^{\Omega} (\theta, \alpha_{i})$ and $\hat{Q}_{i}^{\Omega}
\left( \theta ,\alpha _{i}\right) :=  - \hat{g}_{i} \left( \theta
,\alpha _{i}\right)' \hat{\Omega}_{i}(\tilde{\theta},
\tilde{\alpha}_{i})^{-1} \hat{g}_{i} \left( \theta ,\alpha
_{i}\right)$. Also, for any $\eta > 0$
\begin{equation*}
\Pr \left( \sup_{1 \leq i \leq n} \left| \hat{\alpha}_{i} -
\alpha_{0} \right| \geq \eta \right) = o \left( T^{-1} \right) \text{ and }  \Pr \left( \sup_{1 \leq i \leq n} \left| \hat{\lambda}_{i} \right| \geq \eta \right) =
o \left( T^{-1} \right),
\end{equation*}
where $\hat{\alpha}_{i} = \arg \max_{\alpha}
\hat{Q}_{i}^{\Omega}(\hat{\theta},\alpha)$ and 
$\hat{g}_{i}(\hat{\theta},\hat{\alpha}_{i}) +
 \hat{\Omega}_{i} (\tilde{\theta},
\tilde{\alpha}_{i}) \hat{\lambda}_{i} = 0$.
\end{theorem}

We replace Condition \ref{cond5} by the following condition to
obtain the limit  distribution of the two-step estimator:
\begin{condition}[Smoothness]
\label{cond7} There exists some $M\left( z_{it}\right) $ such
that, for $k = 1, ...,d_g$
\begin{equation*}
\left| \partial ^{d_{1} + d_{2}} g_{k} \left( z_{it};\theta
,\alpha _{i}\right) / \partial \theta ^{d_{1}}\partial \alpha
_{i}^{d_{2}}\right| \leq M\left( z_{it}\right) \qquad 0\leq
d_{1}+d_{2}\leq 1,\ldots ,5,
\end{equation*}
and $\sup_{i}E\left[ M\left( z_{it}\right)  ^{10(d_{\theta}+d_{\alpha}+6)/(1-10v)+\delta}\right] <\infty,$ for some $\delta > 0$ and $0 < v < 1/10$.
\end{condition}
Condition \ref{cond7} guarantees the existence of higher order
expansions for the estimators of the weighting matrices and uniform
convergence of their remainder terms. Conditions \ref{cond4} and
\ref{cond7} are stronger versions of conditions \ref{cond2}(iv),
\ref{cond2}(v) and \ref{cond5}. They are presented separately
because they are only needed when there is a first stage where the
weighting matrices are estimated.

Let $\Sigma_{\alpha_{i}} := \left( G_{\alpha_{i}}' \Omega_{i}^{-1}
G_{\alpha_{i}} \right)^{-1}$, $H_{\alpha_{i}} := \Sigma_{\alpha_{i}}
G_{\alpha_{i}}' \Omega_{i}^{-1}$, and $P_{\alpha_{i}} :=
\Omega_{i}^{-1} - \Omega_{i}^{-1} G_{\alpha_{i}} H_{\alpha_{i}}$.
\begin{lemma} [Asymptotic expansion for two-step estimators of individual effects] \label{lemma2} Under
the Conditions \ref{cond1}, \ref{cond2}, \ref{cond3}, \ref{cond5}, and
\ref{cond4},
\begin{equation}
\sqrt{T} (\hat{\gamma}_{i0} - \gamma_{i0}) = \tilde{\psi}_{i} +
T^{-1/2} B_{\gamma_{i}} + T^{-1} R_{2i},
\end{equation}
where $\hat{\gamma}_{i0} := \hat{\gamma}_{i}(\theta_{0})$,
$$
 \tilde{\psi}_{i} = -  \left(
                          \begin{array}{cc}
                            H_{\alpha_i} \\
                            P_{\alpha_i}\\
                          \end{array}
                        \right) T^{-1/2} \sum_{t=1}^T g(z_{it}) \ind N(0, V_i),
$$
$n^{-1/2} \sum_{i=1}^{n} \tilde{\psi}_{i} \ind N(0, \bar{E}[V_i]),$
$B_{\gamma_{i}} = B_{\gamma_{i}}^{I} + B_{\gamma_{i}}^{G} +
B_{\gamma_{i}}^{\Omega} + B_{\gamma_{i}}^{W}$, $\sup_{1 \leq i \leq n} R_{2i} = o_{P}(\sqrt{T})$, with, for $\Omega_{\alpha_{i,j}} =
\partial \Omega_{\alpha_{i}}/\partial \alpha_{i,j}$,
\begin{footnotesize}\begin{eqnarray*}
V_i &=&  \text{diag} \left(
  \Sigma_{\alpha_{i}},
  P_{\alpha_{i}} 
\right),
\\
B_{\gamma_{i}}^{I} &=& \left(%
\begin{array}{c}
  B_{\alpha_{i}}^{I} \\
  B_{\lambda_{i}}^{I} \\
\end{array}%
\right) = \left(%
\begin{array}{c}
  H_{\alpha_{i}}  \\
P_{\alpha_{i}} \\
\end{array}%
\right) \left(- \sum_{j=1}^{d_{\alpha}} G_{\alpha \alpha_{i,j}}
  \Sigma_{\alpha_{i}}/2
  + E \left[ G_{\alpha_{i}}(z_{it}) H_{\alpha_{i}} g(z_{i,t-j}) \right]\right), \\
B_{\gamma_{i}}^{G} &=& \left(%
\begin{array}{c}
  B_{\alpha_{i}}^{G} \\
  B_{\lambda_{i}}^{G} \\
\end{array}%
\right) =  \left(%
\begin{array}{c}
  -\Sigma_{\alpha_{i}} \\
  H_{\alpha_{i}}'  \\
\end{array}%
\right) \sum_{j=0}^{\infty}  E \left[ G_{\alpha_{i}}(z_{it})'P_{\alpha_{i}}g(z_{i,t-j}) \right] , \\
B_{\gamma_{i}}^{\Omega} &=& \left(%
\begin{array}{c}
  B_{\alpha_{i}}^{\Omega} \\
  B_{\lambda_{i}}^{\Omega} \\
\end{array}%
\right) = \left(%
\begin{array}{c}
  H_{\alpha_{i}}  \\
  P_{\alpha_{i}}  \\
\end{array}%
\right)   \sum_{j=0}^{\infty} E[g(z_{it}) g(z_{it})'P_{\alpha_{i}} g(z_{i,t-j})], \\
B_{\gamma_{i}}^{W} &=& \left(%
\begin{array}{c}
  B_{\alpha_{i}}^{W} \\
  B_{\lambda_{i}}^{W} \\
\end{array}%
\right) = \left(%
\begin{array}{c}
   H_{\alpha_{i}}  \\
   P_{\alpha_{i}}  \\
\end{array}%
\right) \sum_{j=1}^{d_{\alpha}} \Omega_{\alpha_{i,j}} \left( H_{\alpha_{i, j}}^{W'} - H'_{\alpha_{i,j}} \right).
\end{eqnarray*}\end{footnotesize}
\end{lemma}

\begin{theorem} [Limit distribution for two-step estimators
of common parameters] \label{th6} Under the Conditions \ref{cond1},
\ref{cond2}, \ref{cond3}, \ref{cond5}, \ref{cond4} and \ref{cond7},
\begin{equation*}
    \sqrt{nT}(\hat{\theta}-\theta_{0}) \ind  -  J_{s}^{-1}  N\left( \kappa B_{s}, J_{s} \right),
\end{equation*}
where $J_{s} = \bar{E} \left[ G_{\theta_{i}}'P_{\alpha_{i}}
G_{\theta_{i}} \right],$  $B_{s} = \bar{E} \left[ B_{si}^{B} +
B_{si}^{C} \right],$ $ B_{si}^{B} = - G_{\theta_{i}}' \left[
B_{\lambda_{i}}^{I}
 + B_{\lambda_{i}}^{G} + B_{\lambda_{i}}^{\Omega} + B_{\lambda_{i}}^{W}  \right]$,
 $B_{si}^{C} =  \sum_{j=0}^{\infty}  E \left[ G_{\theta_{i}}(z_{it})'
P_{\alpha_{i}} g(z_{i,t-j}) \right]$. The expressions for
$B_{\lambda_{i}}^{I}$, $B_{\lambda_{i}}^{G}$,
$B_{\lambda_{i}}^{\Omega}$ and $B_{\lambda_{i}}^{W}$ are given in
Lemma \ref{lemma2}.
\end{theorem}

Theorem \ref{th6} establishes that one iteration of the GMM
procedure not only improves asymptotic efficiency by reducing the
variance of the influence function, but also removes the variance
and non-optimal weighting matrices components from the bias. The
higher-order bias of the estimator of the individual effects,
$B_{\lambda}^{B}$, now has four components, as in Newey and Smith
(2004). These components correspond to the asymptotic bias for a GMM
estimator with the optimal score, $B_{\lambda}^{I}$; the bias
arising from estimation of $G_{\alpha_{i}}$, $B_{\lambda}^{G}$; the
bias arising from estimation of $\Omega_{i}$,
$B_{\lambda}^{\Omega}$; and the bias arising from the choice of the
preliminary first step estimator, $B_{\lambda}^{W}$. An
additional iteration of the GMM estimator removes the term
$B_{\lambda}^{W}$.

The general procedure for deriving the asymptotic distribution of the FE-GMM
estimators consists of several expansions. First, we derive
higher-order asymptotic expansions for the estimators of the
individual effects, with the common parameter fixed at its true
value $\theta_{0}$. Next, we obtain the asymptotic distribution for the profile
score of the common parameter at $\theta_0$ using the expansions of the estimators
of the individual effects. Finally, we derive the asymptotic distribution of estimator for
the common parameter multiplying the asymptotic distribution of the score 
by the limit profile Jacobian matrix. This
procedure is detailed in the online appendix Fern\'andez-Val and Lee
(2012). Here we characterize the asymptotic bias in a linear
correlated random coefficient model with endogenous regressors.
Motivated by the numerical and empirical examples that follow, we
consider a model where only the variables with common
parameter are endogenous and allow for the moment conditions not to
be martingale difference sequences. 

\noindent \textbf{Example: Correlated random coefficient model with
endogenous regressors.} We consider a simplified version of the
models in the empirical and numerical examples. The notation is the
same as in the theorems discussed above. The moment condition is
$$g(z_{it}; \theta, \alpha_i) = w_{it}(y_{it}-x_{1it}'\alpha_{i}-x_{2it}'\theta),$$
where $w_{it} = (x_{1it}', w_{2it}')'$ and $z_{it}=(x_{1it}',
x_{2it}', w_{2it}', y_{it})'$. That is, only the regressors with
common coefficients are endogenous. Let $\epsilon_{it} = y_{it} -
x_{1it}'\alpha_{i0}-x_{2it}'\theta_0$. To simplify the expressions
for the bias, we assume that $\epsilon_{it} \mid w_i, \alpha_i \sim
i.i.d. (0, \sigma_{\epsilon}^2)$ and $E[x_{2it} \epsilon_{i,t-j}
\mid w_i, \alpha_i] = E[x_{2it} \epsilon_{i,t-j}],$ for $w_{i} =
(w_{i1}, ..., w_{iT})'$ and $j \in \{0, \pm 1, \ldots \}$. Under
these conditions, the optimal weighted matrices are proportional to
$E[w_{it}w_{it}'],$ which do not depend on $\theta_0$ and
$\alpha_{i0}$. We can therefore obtain the optimal GMM estimator in
one step using the sample averages $T^{-1} \sum_{t=1}^T w_{it} w_{it}'$
to estimate the optimal weighting matrices.

In this model, it is straightforward to see that the estimators of
the individual effects have no bias, that is $B_{\gamma_i}^{W,I} =
B_{\gamma_i}^{W,G} = B_{\gamma_i}^{W,1S} = 0$. By linearity of the
first order conditions in $\theta$ and $\alpha_{i},$ 
$B_{si}^{W,V} = 0.$ The only source of bias is the correlation
between the estimators of $\theta$ and $\alpha_i.$ After some
straightforward but tedious algebra, this bias simplifies to
$$
B_{si}^{W,C} = - (d_g - d_{\alpha}) \sum_{j = - \infty}^{\infty} E[x_{2it}
\epsilon_{i,t-j}].
$$
For the
limit Jacobian, we find
$$
J_s^W = \bar{E} \left\{ E[\tilde{x}_{2it} \tilde{w}_{2it}']
E[\tilde{w}_{2it} \tilde{w}_{2it}']^{-1} E[\tilde{w}_{2it}
\tilde{x}_{2it}'] \right\},
$$
where variables with tilde indicate residuals of population linear
projections of the corresponding variable on $x_{1it},$ for example
$\tilde{x}_{2it}= x_{2it} - E[x_{2it} x_{1it}'] E[x_{1it}
x_{1it}']^{-1} x_{1it}$. The expression of the bias is
\begin{equation}\label{bias: linearIV}
\mathcal{B}(\theta_0)  = - (d_g - d_{\alpha}) (J_s^W)^{-1} \bar{E}
 \sum_{j = - \infty}^{\infty} E[\tilde{x}_{2it} (\tilde{y}_{i,t-j} - \tilde{x}_{2i,t-j}'\theta_0)].
\end{equation}

\bigskip

In random coefficient models the ultimate quantities of interest are
often functions of the data, model parameters and individual
effects.  The following corollaries characterize the asymptotic distributions of the fixed effects estimators of these quantities.  The first corollary applies to averages of
functions of the data and individual effects such as average partial effects and average
derivatives in nonlinear models, and average elasticities in linear
models with variables in levels. Section 6 gives an example
of these elasticities. The second corollary applies to averages of smooth
functions of the individual effects including means, variances
and other moments of the distribution of these effects. Sections
2 and 6 give examples of these functions.    We state the results only for estimators constructed from  two-step estimators of the common parameters and individual effects. Similar results apply to estimators constructed from one-step estimators. Both corollaries follow
from Lemma \ref{lemma2} and Theorem \ref{th6} by the delta method.
 
\begin{corollary}[Asymptotic distribution for fixed effects averages] \label{cor1} Let
$\zeta(z;\theta,\alpha_{i})$ be a twice continuously differentiable
function in its second and third argument, such that $\inf_i Var[\zeta(z_{it})]  > 0,$ $\bar{E}E[\zeta(z_{it})^2] < \infty,$ $\bar{E}E|\zeta_{\alpha}(z_{it})|^2 < \infty,$ and $\bar{E}E|\zeta_{\theta}(z_{it})|^2 < \infty,$ where the subscripts on $\zeta$ denote partial derivatives. Then, under
the conditions of Theorem \ref{th6}, for some deterministic sequence $r_{nT} \to \infty$ such that $r_{nT}= O(\sqrt{nT}),$
\begin{equation*}
r_{nT} (\hat{\zeta} - \zeta - B_{\zeta}/T) \ind N(0, V_{\zeta}),
\end{equation*}
where $ \zeta = \bar{E}  E \left[ \zeta(z_{it}) \right],$ 
$$
B_{\zeta} = \bar{E} E \left[
- \sum_{j = 0}^{\infty} \zeta_{\alpha_i} (z_{it})' H_{\alpha_{i}} g(z_{i,t-j}) +
 \zeta_{\alpha_i} (z_{it})' B_{\alpha_{i}}  + \sum_{j=1}^{d_{\alpha}} \zeta_{\alpha \alpha
_{i,j}} (z_{it})'
\Sigma_{\alpha_{i}}/2 - \zeta_{\beta} (z_{it})'J_s^{-1}B_s \right],
$$
for $B_{\alpha_{i}}  = B_{\alpha_{i}}^I  + B_{\alpha_{i}}^G  +  B_{\alpha_{i}}^{\Omega}  + B_{\alpha_{i}}^W,$ and for $r^2 = \lim_{n,T \to \infty} r_{nT}^2/(nT),$
\begin{footnotesize}\begin{equation*}
 V_{\zeta}
=  \bar{E} \Bigg \{ r^2 E \left[  \zeta_{\alpha_i}(z_{it})' \Sigma_{\alpha_{i}}
\zeta_{\alpha_i}(z_{it}) +
\zeta_{\theta}(z_{it})' J_{s}^{-1}
\zeta_{\theta}(z_{it}) \right]   +   \lim_{n,T \to \infty} \frac{r_{nT}^2}{n} E \left[  \left( \frac{1}{T}\sum_{t = 1}^{T}  (\zeta(z_{it}) - \zeta) \right)^2 \right] \Bigg \}.
\end{equation*}\end{footnotesize}
\end{corollary}

\begin{corollary}[Asymptotic distribution for smooth functions of individual effects] \label{cor2} Let
$\mu(\alpha_{i})$ be a twice differentiable function such that $\bar{E}[ \mu(\alpha_{i0})^2] < \infty$ and $\bar{E}| \mu_{\alpha}(\alpha_{i0})|^2 < \infty,$ where the subscripts on $\mu$ denote partial derivatives.  Then, under the
conditions of Theorem \ref{th6}
\begin{equation*}
\sqrt{n}(\hat{\mu} - \mu) \ind N(\kappa B_{\mu}, V_{\mu}),
\end{equation*}
where $\mu = \bar{E} \left[ \mu(\alpha_{i0}) \right],$
$$
B_{\mu} =  \bar{E} 
 \left[ \mu_{\alpha_i} (\alpha_{i0})' B_{\alpha_{i}}
+ \sum_{j=1}^{d_{\alpha}} \mu_{\alpha \alpha _{i,j}}
(\alpha_{i0})' \Sigma_{\alpha_{i}}/2
\right],
$$
for $B_{\alpha_{i}}  = B_{\alpha_{i}}^I  + B_{\alpha_{i}}^G  +  B_{\alpha_{i}}^{\Omega}  + B_{\alpha_{i}}^W,$ and $V_{\mu} =
\bar{E} \left[ (\mu(\alpha_{i0}) - \mu)^{2} \right].$ 
\end{corollary}
The convergence rate $r_{nT}$ in Corollary \ref{cor1} depends on the function $\zeta(z; \theta, \alpha_i).$ For example, $r_{nT} = \sqrt{nT}$ for functions that do not depend on $\alpha_i$ such as $\zeta(z; \theta, \alpha_i) = c'\theta$, where $c$ is a known $d_{\theta}$ vector. In general, $r_{nT} = \sqrt{n}$ for functions that depend on $\alpha_i$. In this case  $r^2 = 0$  and the first two terms of $V_{\zeta}$ drop out. Corollary \ref{cor2} is an important special case of Corollary \ref{cor1}. We present it separately because the asymptotic bias and variance have simplified expressions.


\section{Bias Corrections} \label{s5}
The FE-GMM estimators of common parameters, while consistent, have
bias in the asymptotic distributions under sequences where $n$ and $T$
grow at the same rate. These sequences provide a good approximation
to the finite sample behavior of the estimators in empirical
applications where the time dimension is moderately large. The
presence of bias invalidates any asymptotic inference because the
bias is of the same order as the  variance. In
this section we describe bias correction methods to adjust the asymptotic
distribution of the FE-GMM estimators of the common parameter and
smooth functions of the data, model parameters and individual
effects. All the corrections considered are analytical. Alternative
corrections based on variations of Jackknife can be implemented
using the approaches described in Hahn and Newey (2004) and Dhaene
and Jochmans (2010).\footnote{Hahn, Kuersteiner and Newey (2004)
show that analytical, Bootstrap, and Jackknife bias corrections
methods are asymptotically equivalent up to third order for MLE. We
conjecture that the same result applies to GMM estimators, but the
proof is beyond the scope of this paper.}

We consider three analytical methods that differ in whether the bias
is corrected from the estimator or from the first order conditions, and in whether the correction is one-step or iterated for
methods that correct the bias from the estimator. All these methods
reduce the order of the asymptotic bias without increasing the
asymptotic variance.  They are based on analytical estimators of the bias of the profile score $B_{s}$ and the 
profile Jacobian matrix $J_{s}$. Since these quantities include cross sectional and time series means $\bar{E}$ and $E$ evaluated 
at the true parameter values for the common parameter and individual effects,  they are estimated by 
the corresponding cross sectional and time series averages evaluated at the FE-GMM estimates.  Thus, for any function of the data, common parameter and individual effects  $f_{it}(\theta, \alpha_i),$ let $\hat f_{it}(\theta) = f_{it}(\theta, \hat \alpha_i(\theta)),$ $\hat f_i(\theta) = \hat{E} [\hat f_{it}(\theta)] = T^{-1} \sum_{t=1}^T \hat f_{it}(\theta)$ and
$\hat f(\theta) = \widehat{\bar{E}}[\hat f_i(\theta)] = n^{-1} \sum_{i=1}^n  \hat f_i(\theta)$. Next, define $\hat \Sigma_{\alpha_i}(\theta) = [\hat G_{\alpha_i}(\theta)' \hat \Omega_{i}^{-1} \hat G_{\alpha_i}(\theta)]^{-1},$ $\hat H_{\alpha_i}(\theta) = \hat \Sigma_{\alpha_i}(\theta) \hat G_{\alpha_i}(\theta)' \hat \Omega_{i}^{-1},$ and $\hat P_{\alpha_i}(\theta) = \hat \Omega_{i}^{-1} \hat G_{\alpha_i}(\theta) \hat H_{\alpha_i}(\theta).$ To simplify the presentation, we only give explicit formulas for FE-GMM three-step estimators in the main text. We give the expressions for one and two-step estimators in the Supplementary Appendix. Let
$$
\hat{\mathcal{B}}(\theta) = - \hat J_s(\theta)^{-1}\hat B_s(\theta), \ \ \hat B_s(\theta) = \widehat{\bar{E}}[\hat B_{si}^B(\theta) + \hat{B}_{si}^C(\theta)], \ \ \hat J_s(\theta) = \widehat{\bar{E}}[\hat G_{\theta_i}(\theta)' \hat P_{\alpha_i} (\theta) \hat G_{\theta_i}(\theta)],
$$
where  $\hat B_{si}^{B}(\theta) = - \hat G_{\theta_{i}}(\theta)' [
\hat B_{\lambda_{i}}^{I}(\theta)
 + \hat B_{\lambda_{i}}^{G}(\theta) + \hat B_{\lambda_{i}}^{\Omega}(\theta)+ \hat B_{\lambda_{i}}^{W}(\theta)]$,
\begin{footnotesize}\begin{eqnarray*}
\hat B_{\lambda_{i}}^{I}(\theta) &=& -  \hat P_{\alpha_{i}}(\theta) \sum_{j=1}^{d_{\alpha}} \hat G_{\alpha\alpha_{i,j}}(\theta)
\hat \Sigma_{\alpha_{i}}(\theta)/2
   + \hat P_{\alpha_{i}}(\theta) \sum_{j=0}^{\ell} T^{-1} \sum_{t=j+1}^T \hat G_{\alpha_{it}}(\theta) \hat H_{\alpha_{i}}(\theta) \hat g_{i,t-j}(\theta),\\
\hat B_{\lambda_{i}}^{G}(\theta) &=&  \hat H_{\alpha_{i}}(\theta)' \sum_{j=0}^{\infty}  T^{-1} \sum_{t=j+1}^T  \hat G_{\alpha_{it}}(\theta)' \hat P_{\alpha_{i}}(\theta) \hat g_{i,t-j}(\theta),   \\
\hat B_{\lambda_{i}}^{\Omega}(\theta) &=&   \hat P_{\alpha_{i}}(\theta) \sum_{j=0}^{\ell}  T^{-1} \sum_{t=j+1}^T \hat g_{it}(\theta) \hat g_{it}(\theta)'\hat P_{\alpha_{i}}(\theta) \hat g_{i,t-j}(\theta),
\end{eqnarray*}\end{footnotesize}and $\hat{B}_{si}^C(\theta) = T^{-1} \sum_{j=0}^{\ell} \sum_{t=j+1}^T \hat G_{\theta_{it}}(\theta)'
\hat P_{\alpha_{i}}(\theta) \hat g_{i,t-j}(\theta).$ In the previous expressions, the spectral time series averages 
that involve an infinite number of terms are trimmed.  The trimming parameter $\ell$ is a positive bandwidth that need to be chosen such
that $\ell \to \infty$ and $\ell/T \to 0$ as $T \to \infty$ (Hahn
and Kuersteiner, 2011)

The one-step correction of the estimator
subtracts an estimator of the expression of the asymptotic bias from the estimator of the common
parameter. Using the expressions defined above evaluated at $\hat \theta$,  the
bias-corrected estimator is 
\begin{equation}\label{bc}
\hat{\theta}^{BC} = \hat{\theta} -
\hat{\mathcal{B}}(\hat{\theta})/T.
\end{equation}
This bias correction is straightforward to implement because it only
requires one  optimization.  The iterated
correction is equivalent to solving the nonlinear equation
\begin{equation} \label{ibc}
\hat{\theta}^{IBC} = \hat{\theta} -
\hat{\mathcal{B}}(\hat{\theta}^{IBC})/T.
\end{equation}
When $\theta + \hat{\mathcal{B}}(\theta)$ is invertible in
$\theta$, it is possible to obtain a closed-form solution to the
previous equation.\footnote{See MacKinnon and Smith (1998) for a
comparison of one-step and iterated bias correction methods.}
Otherwise, an iterative procedure is needed. The score
bias-corrected estimator is the solution to the following estimating
equation
\begin{equation}\label{sbc}
\hat{s} (\hat{\theta}^{SBC}) -  \hat{B}_{s}
(\hat{\theta}^{SBC})/T = 0.
\end{equation}
This procedure, while computationally more intensive, has the
attractive feature that both estimator and bias are obtained
simultaneously. Hahn and Newey (2004) show that fully iterated
bias-corrected estimators solve approximated bias-corrected first
order conditions. IBC and SBC are equivalent if the first order
conditions are linear in $\theta$.

\noindent \textbf{Example: Correlated random coefficient model with
endogenous regressors.} The previous methods can be illustrated in
the correlated random coefficient model example in Section 4. Here,
the fixed effects GMM estimators have closed forms:
$$
\hat{\alpha}_i(\theta) = \left(\sum_{t=1}^T
x_{1it}x_{1it}'\right)^{-1} \sum_{t=1}^T x_{1it} (y_{it} -
x_{2it}'\theta),
$$
and
$$
\hat{\theta} = (\hat{J}_s^W)^{-1} \sum_{i=1}^n \left[\sum_{t=1}^T
\tilde{x}_{2it} \tilde{w}_{2it}' \left(\sum_{t=1}^T \tilde{w}_{2it}
\tilde{w}_{2it}'\right)^{-1} \sum_{t=1}^T \tilde{w}_{2it}
\tilde{y}_{it}\right],
$$
where $\hat{J}_s^W = \sum_{i=1}^n [\sum_{t=1}^T \tilde{x}_{2it}
\tilde{w}_{2it}' (\sum_{t=1}^T \tilde{w}_{2it}
\tilde{w}_{2it}')^{-1} \sum_{t=1}^T \tilde{w}_{2it}
\tilde{x}_{2it}'],$ and variables with tilde now indicate residuals
of sample linear projections of the corresponding variable on
$x_{1it},$ for example $\tilde{x}_{2it}= x_{2it} - \sum_{t=1}^T
x_{2it} x_{1it}' (\sum_{t=1}^T x_{1it} x_{1it}')^{-1} x_{1it}$.

We can estimate the bias of $\hat{\theta}$ from the analytic formula
in expression (\ref{bias: linearIV}) replacing population by sample
moments and $\theta_0$ by $\hat{\theta},$ and trimming the number of
terms in the spectral expectation,
$$
\widehat{\mathcal{B}}(\hat{\theta}) = - (d_g - d_{\alpha}) (\hat{J}_s^W)^{-1}
\sum_{i=1}^n \sum_{j = -\ell}^{\ell} \sum_{t= \max(1,
j+1)}^{\min(T,T+j)} \tilde{x}_{2it} (\tilde{y}_{i,t-j} -
\tilde{x}_{2i,t-j}'\hat{\theta}).
$$ 
The one-step bias corrected estimates of the
common parameter $\theta$ and the average of the individual
parameter $\alpha := E[\alpha_{i}]$ are
$$
\hat{\theta}^{BC} =
\hat{\theta} -  \hat{\mathcal{B}}(\hat{\theta})/T, \qquad
\hat{\alpha}^{BC}= n^{-1}
\sum_{i=1}^n\hat{\alpha}_{i}(\hat{\theta}^{BC}).
$$
The iterated bias correction estimator can be derived analytically
by solving
\begin{equation*}
\hat{\theta}^{IBC} = \hat{\theta} -
\hat{\mathcal{B}}(\hat{\theta}^{IBC})/T,
\end{equation*}
which has closed-form solution
\begin{footnotesize}\begin{multline*}
\hat{\theta}^{IBC} = \left[ I_{d_{\theta}} + (d_g - d_{\alpha})(\hat{J}_s^W)^{-1}
\sum_{i=1}^n \sum_{j = -\ell}^{\ell} \sum_{t= \max(1,
j+1)}^{\min(T,T+j)} \tilde{x}_{2it} \tilde{x}_{2i,t-j}'
/(nT^2)\right]^{-1} \times \\ \left[ \hat{\theta} + (d_g -
d_{\alpha})(\hat{J}_s^W)^{-1} \sum_{i=1}^n \sum_{j = -\ell}^{\ell} \sum_{t=
\max(1, j+1)}^{\min(T,T+j)} \tilde{x}_{2it} \tilde{y}_{i,t-j}
/(nT^2) \right].
\end{multline*}\end{footnotesize}The score bias correction is the same as the iterated correction
because the first order conditions are linear in $\theta$.

The bias correction methods described above yield normal asymptotic
distributions centered at the true parameter value for panels where
$n$ and $T$ grow at the same rate with the sample size. This result is formally stated in Theorem \ref{th7},
which establishes that all the methods are asymptotically
equivalent, up to first order.


\begin{theorem} [Limit distribution of bias-corrected
FE-GMM] \label{th7} Assume that $\sqrt{nT}(\hat B_{s}(\overline{\theta}) -
B_{s})/T \inp 0$ and $\sqrt{nT}(\hat J_{s}(\overline{\theta}) - J_{s})/T
\inp 0$, for some $\overline \theta = \theta_0 + O_P((nT)^{-1/2})$. Under
Conditions \ref{cond1}, \ref{cond2}, \ref{cond3}, \ref{cond5},
\ref{cond4} and \ref{cond7}, for $C \in \left\{ BC, SBC, IBC
\right\}$
\begin{equation}
\sqrt{nT}(\hat{\theta}^{C} - \theta_{0}) \ind N \left( 0, J_{s}^{-1}
\right),
\end{equation}
where $\hat{\theta}^{BC}$, $\hat{\theta}^{IBC}$ and
$\hat{\theta}^{SBC}$ are defined in (\ref{bc}), (\ref{ibc}) and
(\ref{sbc}), and $J_{s} = \bar{E} [G_{\theta_{i}}'P_{\alpha_{i}}
G_{\theta_{i}}]$.
\end{theorem}
The convergence condition for the estimators of $B_{s}$ and
$J_{s}$ holds for sample analogs evaluated at the initial FE-GMM
one-step or two-step estimators if the trimming sequence is chosen such that $\ell \to
\infty$ and $\ell/T \to 0$ as $T \to \infty$. Theorem \ref{th7} also
shows that all the bias-corrected estimators considered are
first-order asymptotically efficient, since their variances achieve
the semiparametric efficiency bound for the common parameters in
this model, see Chamberlain (1992).

The following corollaries give bias corrected estimators for averages of the data and individual effects and for moments of the individual effects, together with the limit distributions of these estimators and consistent estimators of their asymptotic variances.   To construct the corrections, we use bias corrected estimators of the common parameter. The corollaries then follow from Lemma \ref{lemma2} and Theorem \ref{th7} by the delta method. We use the same notation as in the estimation of the bias of the common parameters above to denote the  estimators of the components of the bias and variance.
 
\begin{corollary}[Bias correction for fixed effects averages] \label{c1th5} Let
$\zeta(z;\theta,\alpha_{i})$ be a twice continuously differentiable
function in its second and third argument,  such that $\inf_i Var[\zeta(z_{it})] > 0,$ $\bar{E}E[\zeta(z_{it})^2] < \infty,$ $\bar{E}E[\zeta_{\alpha}(z_{it})^2] < \infty,$ and $\bar{E}E|\zeta_{\theta}(z_{it})|^2 < \infty.$  For $C \in \left\{ BC, SBC, IBC \right\}$, let $\hat{\zeta}^{C} = \hat{\zeta}(\hat{\theta}^{C}) - \hat{B}_{\zeta}(\hat{\theta}^{C})/T$ where
\begin{eqnarray*}
\hat{B}_{\zeta}(\theta) &=& \widehat{\bar E} \left[
\sum_{j = 0}^{\ell} \frac{1}{T} \sum_{t=  j+1}^{T}
\hat \zeta_{\alpha_{it}} (\theta)'
 \hat{\tilde{\psi}}_{\alpha_{i,t-j}}(\theta) +
 \hat \zeta_{\alpha_{i}} (\theta)'
 \hat{B}_{\alpha_{i}}(\theta)  + \sum_{j=1}^{d_{\alpha}} \hat \zeta_{\alpha \alpha
_{i,j}} (\theta)'
\hat{\Sigma}_{\alpha_{i}}(\theta)/2 \right],
\end{eqnarray*}
where  $\ell$ is a positive bandwidth such that
$\ell \to \infty$ and $\ell/T \to 0$ as $T \to \infty$. Then, under
the conditions of Theorem \ref{th7}
\begin{equation*}
r_{nT}(\hat{\zeta}^{C} - \zeta) \ind N(0, V_{\zeta}),
\end{equation*}
where $r_{nT}$, $ \zeta ,$ and $ V_{\zeta}$ are defined in Corollary \ref{cor1}.
Also,  for any $\bar \theta = \theta_0 + O_P((nT)^{-1/2})$ and $\bar{\zeta} = \zeta + O_P(r_{nT}^{-1})$, 
\begin{eqnarray*}
\hat{V}_{\zeta} &=&  \frac{r_{nT}^2}{nT} \widehat{\bar E}  \Bigg \{ \widehat{E}[ \hat \zeta_{\alpha_{it}}(\bar{\theta})' \hat{\Sigma}_{\alpha_{i}}(\bar{\theta})
\hat \zeta_{\alpha_{it}}(\bar{\theta}) + \hat \zeta_{\theta_{it}}(\bar{\theta})'
\hat{J}_{s}(\bar{\theta})^{-1} \hat \zeta_{\theta_{it}}(\bar{\theta}) ] 
 + T  \left(\widehat{E} [\hat \zeta_{it}(\bar{\theta}) -
\bar{\zeta}]\right)^2  \Bigg \}
\end{eqnarray*}
is a consistent estimator for $V_{\zeta}$.
\end{corollary}

\begin{corollary}[Bias correction for smooth functions of individual effects] \label{c2th5} Let
$\mu(\alpha_{i})$ be a twice differentiable function such that $\bar{E}[ \mu(\alpha_{i0})^2] < \infty$ and $\bar{E}| \mu_{\alpha}(\alpha_{i0})|^2 < \infty$. For $C \in
\left\{ BC, SBC, IBC \right\}$, let $\hat{\mu}^{C} = \widehat{\bar E}[
\hat \mu_i(\hat{\theta}^{C}) ]-  \hat{B}_{\mu}(\hat{\theta}^{C})/T,$ where $\hat \mu_i(\theta) = \mu(\hat{\alpha_{i}}(\theta)),$
and $\hat{B}_{\mu}(\theta) = \widehat{\bar E}[ \hat \mu_{\alpha_i} (\theta)' \hat{B}_{\alpha_{i}}(\theta)
+ \sum_{j=1}^{d_{\alpha}} \hat \mu_{\alpha \alpha _{i,j}}(\theta)' \hat{\Sigma}_{\alpha_{i}}(\theta)/2].$ Then, under the
conditions of Theorem \ref{th7}
\begin{equation*}
\sqrt{n}(\hat{\mu}^{C} - \mu) \ind N(0, V_{\mu}),
\end{equation*}
where $\mu = \bar{E} \left[ \mu(\alpha_{i0}) \right]$ and $V_{\mu} =
\bar{E} \left[ (\mu(\alpha_{i0}) - \mu)^{2} \right].$ Also, for any $\bar \theta = \theta_0 + O_P((nT)^{-1/2})$ and $\bar{\mu} = \mu + O_P(n^{-1/2})$,
\begin{equation}\label{eq: std_errors}
\hat{V}_{\mu} = \widehat{\bar{E}} \left[\{\hat \mu_i(\bar{\theta}) - \bar{\mu}\}^{2} +
\hat \mu_{\alpha_i}(\bar{\theta})'
\hat{\Sigma}_{\alpha_{i}}(\bar{\theta})
\hat \mu_{\alpha_i}(\bar{\theta})/T \right],
\end{equation}
is a  consistent estimator for $V_{\mu}$. The second term in (\ref{eq: std_errors}) is included to improve the finite sample
properties of the estimator in short panels.
\end{corollary}


\section{Empirical example} \label{s6}
We illustrate the new estimators with an empirical example based on
the classical cigarette demand study of Becker, Grossman and Murphy
(1994) (BGM hereafter). Cigarettes are addictive goods. To account
for this addictive nature, early cigarette demand studies included
lagged consumption as explanatory variables (e.g., Baltagi and
Levin, 1986).  This approach, however, ignores that rational or
forward-looking consumers take into account the effect of  today's
consumption decision on future consumption decisions. Becker and
Murphy (1988) developed a model of rational addiction where expected
changes in future prices affect the current consumption. BGM empirically tested
this model using a linear structural demand function based
on quadratic utility assumptions. The demand function includes both
future and past consumptions as determinants of current demand, and
the future price affects the current demand only through the future
consumption. They found that the effect of future consumption on
current consumption is significant, what they took as evidence in
favor of the rational model.

Most of the empirical studies in this literature use yearly
state-level panel data sets. They include fixed effects to control
for additive heterogeneity at the state-level and use leads and lags
of cigarette prices and taxes as instruments for leads and lags of
consumption. These studies, however, do not consider possible
non-additive heterogeneity in price elasticities or sensitivities
across states. There are multiple reasons why there may be
heterogeneity in the price effects across states correlated with the
price level. First, the considerable differences in income,
industrial, ethnic and religious composition at inter-state level
can translate into different tastes and policies toward cigarettes.
Second, from the perspective of the theoretical model developed by
Becker and Murphy (1988), the price effect is a function of the
marginal utility of wealth that varies across states and depends on
cigarette prices. If the price effect is heterogenous and correlated
with the price level, a fixed coefficient specification may produce
substantial bias in estimating the average elasticity of cigarette
consumption because the between variation of price is much larger
than the within variation. Wangen (2004) gives additional
theoretical reasons against a fixed coefficient specification for
the demand function in this application.

We consider the following linear specification for the demand
function
\begin{equation}\label{eq: rc}
C_{it}=\alpha_{0i} +
\alpha_{1i}P_{it}+\theta_{1}C_{i,t-1}+\theta_{2}C_{i,t+1}+X_{it}'\delta+\epsilon_{it},
\end{equation}
where $C_{it}$ is cigarette consumption in state $i$ at time $t$
measured by per capita sales in packs; $\alpha_{0i}$ is an additive
state effect; $\alpha_{1i}$ is a state specific price coefficient;
$P_{it}$ is the price in 1982-1984 dollars; and $X_{it}$ is a vector
of covariates which includes income, various measures of incentive
for smuggling across states, and year dummies. We estimate the model
parameters using OLS and IV methods with both fixed coefficient for
price  and random coefficient  for price. The data set, consisting
of an unbalanced panel of 51 U.S. states over the years 1957 to
1994, is the same as in Fenn, Antonovitz and Schroeter (2001).  The
set of instruments for $C_{i,t-1}$ and $C_{i,t+1}$ in the IV
estimators is the same as in specification 3 of BGM and includes
$X_{it}$, $P_{it}$, $P_{i,t-1}$, $P_{i,t+1}$, $Tax_{it}$,
$Tax_{i,t-1}$, and $Tax_{i,t+1}$, where $Tax_{it}$ is the state
excise tax for cigarettes in 1982-1984 dollars.

Table 1 reports estimates of coefficients and demand elasticities.
We focus on the coefficients of the key variables, namely $P_{it}$,
$C_{i,t-1}$ and $C_{i,t+1}$. Throughout the table, FC refers to the
fixed coefficient specification with $\alpha_{1i} = \alpha_1$ and RC
refers to the random coefficient specification in equation (\ref{eq:
rc}). BC and IBC refer to estimates after bias correction and
iterated bias correction, respectively. Demand elasticities are
calculated using the expressions in Appendix A of BGM. They are
functions of $C_{it}$,$P_{it}$, $\alpha_{1i}$, $\theta_{1}$ and
$\theta_{2}$, linear in $\alpha_{1i}$. For random coefficient
estimators, we report the mean of individual elasticities, i.e.
$$\hat{\zeta}_h = \frac{1}{nT}\sum_{i=1}^{n} \sum_{t=1}^T \zeta_{h}(z_{it}; \hat \theta, \hat \alpha_i),$$
where $\zeta_{h}(z_{it}; \theta, \alpha_i) = \partial \log C_{it(h)}
/
\partial \log P_{it(h)}$ are price elasticities at different time horizons
$h$. Standard errors for the elasticities are obtained by the delta
method as described in Corollaries \ref{c1th5} and \ref{c2th5}. For bias-corrected RC estimators
the standard errors use bias-corrected estimates of $\theta$ and $\alpha_i$.

As BGM, we find that OLS  estimates substantially differ from their
IV counterparts. IV-FC underestimates the elasticities relative to
IV-RC. For example, the long-run elasticity estimate is  $-0.70$
with IV-FC, whereas it is $-0.88$ with IV-RC. This difference is
also pronounced for short-run elasticities, where the IV-RC
estimates are more than 25 percent larger than the IV-FC estimates.
We observe the same pattern throughout the table for every
elasticity. The bias comes from both the estimation of the common
parameter $\theta_{2}$ and the mean of the individual specific
parameter $E[\alpha_{1i}]$. The bias corrections increase the
coefficient of future consumption $C_{i,t+1}$ and reduce the
absolute value of the mean of the price coefficient. Moreover, they
have significant impact on the estimator of dispersion of the price
coefficient. The uncorrected estimates of the standard deviation are
more than $20\%$ larger than the bias corrected counterparts. In the
online appendix Fern\'andez-Val and Lee (2012), we show through a
Monte-Carlo experiment calibrated to this empirical example, that
the bias is generally large for dispersion parameters and the bias
corrections are effective in reducing this bias. As a consequence of
shrinking the estimates of the dispersion of $\alpha_{1i}$, we
obtain smaller standard errors for the estimates of $E[\alpha_{1i}]$
throughout the table. In the Monte-Carlo experiment, we also find
that this correction in the standard errors provides improved
inference.

\section{Conclusion} \label{s9}
This paper introduces a new class of fixed effects GMM estimators
for panel data models with unrestricted nonadditive heterogeneity and
endogenous regressors.
Bias correction methods are developed because these estimators
suffer from the incidental parameters problem. Other estimators
based on moment conditions, like the class of GEL estimators, can be
analyzed using a similar methodology.  An attractive alternative
framework for estimation and inference in random coefficient models
is a flexible Bayesian approach. It would be interesting to explore
whether there are connections between moments of posterior
distributions in the Bayesian approach and the fixed effects
estimators considered in the paper. Another interesting extension
would be to find bias reducing priors in the GMM framework similar
to the ones characterized by Arellano and Bonhomme (2009) in the MLE
framework. We leave these extensions to future research.



\begin{figure}[h]

\begin{center}

\centering\epsfig{figure=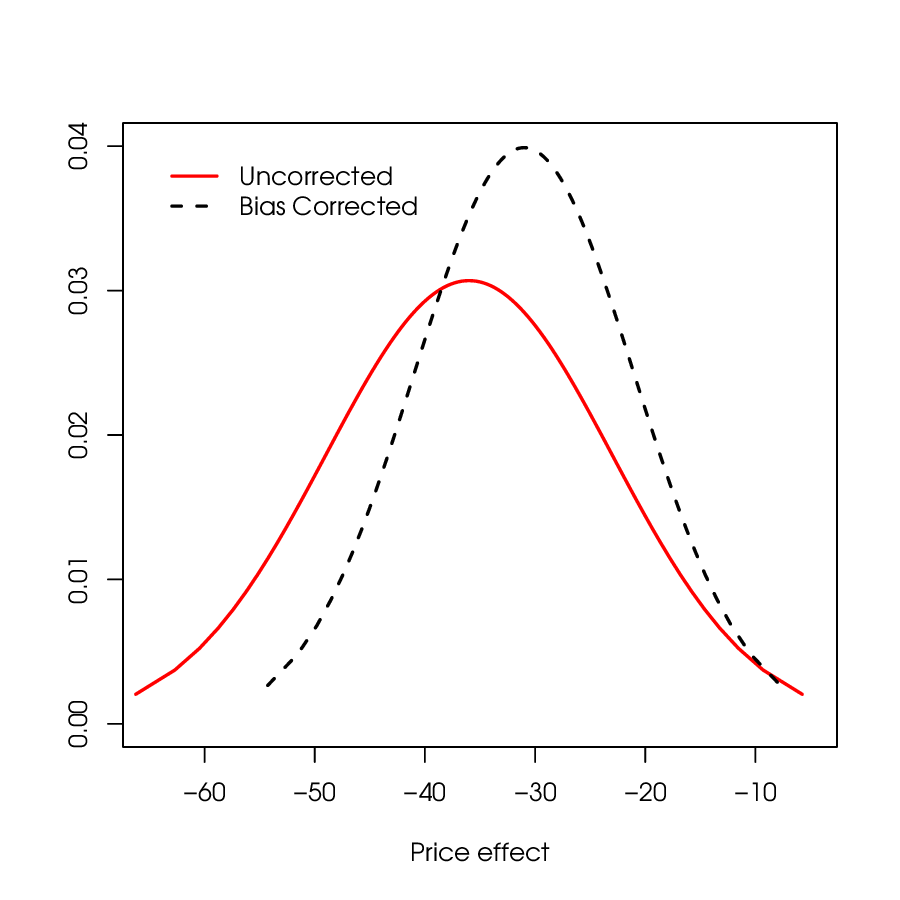,width=6in,height=6in}

\caption{\label{fig1} Normal approximation to the distribution of
price effects using uncorrected (solid line) and bias corrected
(dashed line) estimates of the mean and standard deviation of the
distribution of price effects. Uncorrected estimates of the mean and
standard deviation are -36 and 13, bias corrected estimates are -31
and 10.}

\end{center}

\end{figure}



\begin{table}
\begin{centering}
\textbf{Table 1: Estimates of Rational Addiction Model for Cigarette
Demand}
\par\end{centering}

\begin{center}
\begin{tabular}{ccccccccc}
\toprule\toprule \noalign{\vskip\doublerulesep}
   & OLS-FC  & IV-FC  & \multicolumn{3}{c}{OLS-RC} & \multicolumn{3}{c}{IV-RC}\tabularnewline[\doublerulesep]
\noalign{\vskip\doublerulesep}
 \cmidrule(rl){4-6}\cmidrule(rl){7-9} &  &  & NBC  & BC  & IBC  & NBC  & BC  & IBC\tabularnewline[\doublerulesep]
\hline 
 &  &  &  &  &  &  &  & \tabularnewline
\multicolumn{9}{l}{Coefficients} \tabularnewline
 &  &  &  &  &  &  &  & \tabularnewline
(Mean) $P_{t}$  & -9.58  & -34.10  & -13.49  & -13.58  &
-13.26  & -36.39  & -31.26  & -31.26 \tabularnewline
 & (1.86) & (4.10) & (3.55) & (3.55) & (3.55) & (4.85) & (4.62) & (4.64)\tabularnewline
 &  &  &  &  &  &  &  & \tabularnewline
(Std. Dev.) $P_{t}$  &  &  & 4.35  & 4.22  & 4.07  & 12.86  & 10.45  &
10.60 \tabularnewline
 &  &  & (0.98) & (1.02) & (1.03) & (2.35) & (2.13) & (2.15)\tabularnewline
 &  &  &  &  &  &  &  & \tabularnewline
$C_{t-1}$  & 0.49  & 0.45  & 0.48  & 0.48  & 0.48  & 0.44  & 0.44  &
0.45 \tabularnewline
 & (0.01) & (0.06) & (0.04) & (0.04) & (0.04) & (0.04) & (0.04) & (0.04)\tabularnewline
 &  &  &  &  &  &  &  & \tabularnewline
$C_{t+1}$  & 0.44  & 0.17  & 0.44  & 0.43  & 0.44  & 0.23  & 0.29  &
0.27 \tabularnewline
 & (0.01) & (0.07) & (0.04) & (0.04) & (0.04) & (0.05) & (0.05) & (0.05)\tabularnewline
 &  &  &  &  &  &  &  & \tabularnewline
\hline
 &  &  &  &  &  &  &  & \tabularnewline
\multicolumn{9}{l}{Price elasticities} \tabularnewline
 &  &  &  &  &  &  &  & \tabularnewline
Long-run  & -1.05  & -0.70  & -1.30  & -1.31  & -1.28  & -0.88  &
-0.91  & -0.90 \tabularnewline
 & (0.24) & (0.12) & (0.28) & (0.28) & (0.28) & (0.09) & (0.10) & (0.10)\tabularnewline
 &  &  &  &  &  &  &  & \tabularnewline
Own Price  & -0.20  & -0.32  & -0.27  & -0.27  & -0.27  & -0.38  &
-0.35  & -0.35 \tabularnewline (Anticipated)  & (0.04) & (0.04) &
(0.06) & (0.06) & (0.06) & (0.04) & (0.04) & (0.04)\tabularnewline
 &  &  &  &  &  &  &  & \tabularnewline
Own Price  & -0.11  & -0.29  & -0.15  & -0.16  & -0.15  & -0.33  &
-0.29  & -0.29 \tabularnewline (Unanticipated)  & (0.02) & (0.03) &
(0.04) & (0.04) & (0.04) & (0.04) & (0.04) & (0.04)\tabularnewline
 &  &  &  &  &  &  &  & \tabularnewline
Future Price  & -0.07  & -0.05  & -0.10  & -0.10  & -0.09  & -0.09
& -0.10  & -0.09 \tabularnewline (Unanticipated)  & (0.01) & (0.03)
& (0.02) & (0.02) & (0.02) & (0.02) & (0.02) & (0.02)\tabularnewline
 &  &  &  &  &  &  &  & \tabularnewline
Past Price  & -0.08  & -0.14  & -0.11  & -0.11  & -0.10  & -0.16  &
-0.15  & -0.15 \tabularnewline (Unanticipated)  & (0.01) & (0.02) &
(0.03) & (0.02) & (0.03) & (0.02) & (0.02) & (0.02)\tabularnewline
 &  &  &  &  &  &  &  & \tabularnewline
Short-Run  & -0.30  & -0.35  & -0.41  & -0.41  & -0.40  & -0.44  &
-0.44  & -0.43 \tabularnewline
 & (0.05) & (0.06) & (0.12) & (0.12) & (0.12) & (0.06) & (0.06) & (0.06)\tabularnewline
\bottomrule\bottomrule
\end{tabular}
\end{center}
{\footnotesize
\begin{flushleft}
RC/FC refers to random/fixed coefficient model. NBC/BC/IBC refers to
no bias-correction/bias correction/iterated bias correction
estimates.\\Note: \it{Standard errors are in parenthesis.}
\end{flushleft}}

\end{table}

\clearpage
\newpage

\newpage
\setcounter{page}{1}
\begin{center}
\textbf{Supplementary Appendix to Panel Data Models with Nonadditive Unobserved Heterogeneity: Estimation and Inference}\\
\normalsize{Iv\'an Fern\'andez-Val and Joonhwan Lee}\\
\normalsize{\today}
\end{center}

\bigskip

\footnotesize

 This supplement to the paper ``Panel Data Models with Nonadditive Unobserved Heterogeneity: Estimation and Inference'' provides
additional numerical examples and the proofs of the main results.
It is organized in seven appendices. Appendix A contains a Monte Carlo
simulation calibrated to the empirical example of the paper. Appendix B gives the proofs of the consistency of
the one-step and two-step FE-GMM estimators. Appendix C includes the derivations of the
asymptotic distribution of one-step and two-step FE-GMM estimators. Appendix D
provides the derivations of the asymptotic distribution of bias corrected FE-GMM estimators. Appendix E and 
Appendix F contain the characterization of the stochastic expansions for the estimators of the individual effects and the scores.
Appendix G includes the expressions for the
scores and their derivatives.

Throughout the appendices $O_{uP}$ and $o_{uP}$ will denote uniform
orders in probability. For example, for a sequence of random
variables $\{\xi_{i}: 1 \leq i \leq n \}$, $\xi_{i} = O_{uP}(1)$ means
$\sup_{1 \leq i \leq n} \xi_{i} = O_{P}(1)$ as $n \to \infty$, and $\xi_{i} =
o_{uP}(1)$ means $\sup_{1 \leq i \leq n} \xi_{i} = o_{P}(1)$ as $n \to \infty$. It can
be shown  that the usual algebraic properties for $O_P$ and  $o_{P}$ orders also
apply to the uniform orders  $O_{uP}$ and $o_{uP}$. Let $e_{j}$ denote a $1 \times
d_g$ unitary vector with a one in position $j$. For a matrix $A$,
$|A|$ denotes Euclidean norm, that is $|A|^{2} = trace [A A']$. HK
refers to Hahn and Kuersteiner (2011). 

\begin{appendix}

\appendix


\section{Numerical example} \label{s7}

We design a Monte Carlo experiment to closely match the cigarette
demand empirical example in the paper. In particular, we consider the following
linear model with common and individual specific parameters:
\begin{eqnarray*}
C_{it} &=& \alpha_{0i}+\alpha_{1i}P_{it}+\theta_{1}C_{i,t-1}+\theta_{2}C_{i,t+1}+\psi\epsilon_{it},\\
P_{it} &=& \eta_{0i}+\eta_{1i}Tax_{it}+u_{it},  \ \
(i=1,2,\ldots,n,\, t=1,2,\ldots,T);
\end{eqnarray*}
where $\{(\alpha_{ji}, \eta_{ji}) :  1 \leq i \leq n\}$ is i.i.d.
bivariate normal with mean $(\mu_j, \mu_{\eta_j})$, variances
$(\sigma_j^2, \sigma_{\eta_j}^2)$, and correlation $\rho_j$, for $j
\in \{0,1\},$ independent across $j$; $\{u_{it}: 1 \leq t \leq T, 1 \leq  i \leq n \}$ is i.i.d
$N(0,\sigma_{u}^{2})$; and $\{\epsilon_{it}: 1 \leq t  \leq T, 1 \leq  i \leq n
\}$ is i.i.d. standard normal. We fix the values of $Tax_{it}$ to
the values in the data set. All the parameters other than $\rho_1$
and $\psi$ are calibrated to the data set. Since the panel is
balanced for only $1972$ to $1994$, we set $T=23$ and generate
balanced panels for the simulations. Specifically, we consider
\begin{eqnarray*}
& & n = 51,\, T=23;\,\mu_{0}=72.86,\,\mu_{1}=-31.26,\,\mu_{\eta_0}=0.81,\,\mu_{\eta_1}=0.13,\ \
\sigma_{0}  = 18.54,\,\sigma_{1}=10.60,\,\sigma_{\eta_0}=0.14,\\ 
& &\sigma_{\eta_1}=2.05,\,\sigma_{u}=0.15,\,\theta_{1}=0.45,\,\theta_{2}=0.27,\ \
\rho_0=-0.17,\,\rho_1  \in
\{0,\,0.3,\,0.6,\,0.9\},\,\psi\in\{2,\,4,\,6\}.
\end{eqnarray*}
In the empirical example, the estimated values of $\rho_1$ and
$\psi$ are close to $0.3$ and $5$, respectively.

Since the model is dynamic with leads and lags of the dependent
variable on the right hand side, we construct the series of $C_{it}$
by solving the difference equation following BGM. The stationary
part of the solution is
$$C_{it}=\frac{1}{\theta_{1}\phi_{1}(\phi_{2}-\phi_{1})}\sum_{s=1}^{\infty}\phi_{1}^{s}h_{i}(t+s)+\frac{1}{\theta_{1}\phi_{2}(\phi_{2}-\phi_{1})}\sum_{s=0}^{\infty}\phi_{2}^{-s}h_{i}(t-s)$$
where
\begin{equation*}
h_{i}(t)  = \alpha_{0i}+\alpha_{1i}P_{i,t-1}+\psi\epsilon_{i,t-1}, \
\phi_{1}
 =
 \frac{1-(1-4\theta_{1}\theta_{2})^{1/2}}{2\theta_{1}},\,\phi_{2}=\frac{1+(1-4\theta_{1}\theta_{2})^{1/2}}{2\theta_{1}}.
 \end{equation*}
In our specification, these values are $\phi_1 = 0.31$ and $\phi_2 =
1.91$. The parameters that we vary across the experiments are
$\rho_1$ and $\psi$. The parameter $\rho_1$ controls the degree of
correlation between $\alpha_{1i}$ and $P_{it}$ and determines the
bias caused by using fixed coefficient estimators. The parameter
$\psi$ controls the degree of endogeneity in $C_{i,t-1}$ and
$C_{i,t+1}$, which determines the bias of OLS and the incidental
parameter bias of random coefficient IV estimators. Although $\psi$
is not an ideal experimental parameter because it is the variance of
the error, it is the only free parameter that affects the
endogeneity of $C_{i,t-1}$ and $C_{i,t+1}$. In this design we cannot
fully remove the endogeneity of $C_{i,t-1}$ and $C_{i,t+1}$ because
of the dynamics.

In each simulation, we estimate the parameters with standard fixed
coefficient OLS and IV with additive individual effects (FC) , and the FE-GMM OLS and IV
estimators with the individual specific coefficients (RC). For IV,
we use the same set of instruments as in the empirical example. We
report results only for the common coefficient $\theta_{2}$, and the
mean and standard deviation of the individual-specific coefficient
$\alpha_{1i}$. Throughout the tables, $Bias$ refers to the mean of
the bias across simulations; $SD$ refers to the standard deviation
of the estimates; $SE/SD$ denotes the ratio of the average standard
error to the standard deviation; and $p;.05$ is the rejection
frequency of a two-sided test with nominal level of $0.05$ that the parameter is equal to its true value. For bias-corrected RC
estimators the standard errors are calculated using bias corrected estimates of the common parameter and individual effects.

Table A.1 reports the results for the estimators of $\theta_{2}$. We
find significant biases in all the OLS estimators relative to the standard deviations of these estimators.
The bias of OLS grows with $\psi$.  The IV-RC estimator has bias unless $\rho_1=0$, that is unless there is no correlation
between $\alpha_{1i}$ and $P_{it},$ and its test shows size distortions due to the bias and underestimation in the standard errors. IV-RC
estimators have no bias in every configuration and their tests display much
smaller size distortions than for the other estimators. The bias
corrections preserve the bias and inference properties of the RC-IV
estimator.

Table A2 reports similar results for the estimators of the mean of
the individual specific coefficient $\mu_{1}=\bar{E}[\alpha_{1i}]$. We
find substantial biases for OLS and IV-FC estimators. RC-IV displays
some bias, which is removed by the corrections in some
configurations. The bias corrections provide significant
improvements in the estimation of standard errors. IV-RC standard
errors overestimate the dispersion by more than $15\%$ when $\psi$
is greater than $2$, whereas IV-BC or IV-IBC estimators
have SE/SD ratios close to $1$. As a result bias corrected
estimators show smaller size distortions. This improvement comes
from the bias correction in the estimates of the dispersion of
$\alpha_{1i}$ that we use to construct the standard errors. The bias
of the estimator of the dispersion is generally large, and is
effectively removed by the correction. We can see more evidence on
this phenomenon in Table A3.

Table A3 shows the results for the estimators of the standard
deviation of the individual specific coefficient
$\sigma_{1}=\bar{E}[(\alpha_{1i} - \mu_1)^2]^{1/2}$. As noted above, the bias corrections
are relevant in this case. As $\psi$ increases, the bias grows in
orders of $\psi$. Most of bias is removed by the correction even
when $\psi$ is large. For example, when $\psi=6$, the bias of IV-RC
estimator is about $4$ which is larger than two times its standard
deviation. The correction reduces the bias to about $0.5$, which is
small relative to the standard deviation. Moreover, despite the
overestimation in the standard errors, there are important size
distortions for IV-RC estimators for tests on $\sigma_1$ when $\psi$
is large. The bias corrections bring the rejection frequencies close
to their nominal levels.

Overall, the calibrated Monte-Carlo experiment confirms that the
IV-RC estimator with bias correction provides improved estimation
and inference for all the parameters of interest for the model
considered in the empirical example.


\section{Consistency of One-Step and Two-Step FE-GMM Estimator}

\begin{lemma}
\label{la8}Suppose that the Conditions \ref{cond1} and \ref{cond2} hold. Then, for every $\eta > 0$
$$
\Pr \left\{ \sup_{1 \leq i \leq n} \sup_{(\theta,\alpha) \in
\Upsilon} \left| \hat{Q}_{i}^{W} (\theta, \alpha) - Q_{i}^{W}
(\theta, \alpha) \right| \geq \eta \right\} = o(T^{-1}),
$$
and
$$
\sup_{\alpha} \left| Q_{i}^{W}(\theta,\alpha) -
Q_{i}^{W}(\theta',\alpha) \right| \leq C \cdot E[M(z_{it})]^{2}
\left| \theta - \theta' \right|
$$
for some constant $C > 0$.
\end{lemma}
\begin{proof}
First, note that
\begin{multline*}
\left| \hat{Q}_{i}^{W} (\theta, \alpha) - Q_{i}^{W} (\theta, \alpha)
\right| \leq \left|  \hat{g}_{i}(\theta, \alpha)' W_{i}^{-1} \hat{g}_{i}(\theta, \alpha) - g_{i}(\theta, \alpha)' W_{i}^{-1} g_{i}(\theta, \alpha) \right| +  \left| \hat{g}_{i}(\theta, \alpha)' (\hat W_{i}^{-1} - W_{i}^{-1}) \hat{g}_{i}(\theta, \alpha) \right| \\ 
\leq \left| [ \hat{g}_{i}(\theta, \alpha) -
g_{i}(\theta,\alpha) ]' W_{i}^{-1} [ \hat{g}_{i}(\theta, \alpha) -
g_{i}(\theta,\alpha) ] \right| 
+  2 \cdot \left| g_{i}(\theta,\alpha)' W_{i}^{-1} [
\hat{g}_{i}(\theta, \alpha) -
g_{i}(\theta,\alpha) ] \right|  \\
+  \left| [\hat{g}_{i}(\theta, \alpha) - g_{i}(\theta,\alpha)]' (\hat W_{i}^{-1} - W_{i}^{-1})  [\hat{g}_{i}(\theta, \alpha) - g_{i}(\theta,\alpha)] \right| 
+ 2 \left| [\hat{g}_{i}(\theta, \alpha) - g_{i}(\theta,\alpha)]' (\hat W_{i}^{-1} - W_{i}^{-1}) g_{i}(\theta,\alpha) \right| 
\\ + \left| g_{i}(\theta,\alpha)' (\hat W_{i}^{-1} - W_{i}^{-1}) g_{i}(\theta,\alpha) \right| 
\leq d_g^{2} \max_{1 \leq k \leq d_g} \left|  \hat{g}_{k,i}(\theta,
\alpha) - g_{k,i}(\theta,\alpha) 
\right|^{2} \left| W_{i} \right|^{-1} \\
+ 2 d_g^{2} \sup_{1 \leq i \leq n} E[ M(z_{it}) ] \left| W_{i}
\right|^{-1} \max_{1 \leq k \leq d_g} \left| \hat{g}_{k,i}(\theta,
\alpha) - g_{k,i}(\theta,\alpha)  \right| + o_P\left(\max_{1 \leq k \leq d_g} \left| \hat{g}_{k,i}(\theta,
\alpha) - g_{k,i}(\theta,\alpha) 
\right|\right),
\end{multline*}
where we use that $\sup_{1 \leq i \leq n} |\hat W_i - W_i| = o_P(1)$. Then, by Condition \ref{cond2}, we can apply Lemma 4 of HK to $|\hat{g}_{k,i}(\theta,\alpha)-g_{k,i}(\theta,\alpha)|$ to obtain the first part.

The second part follows from
\begin{eqnarray*}
\left| Q_{i}^{W}(\theta,\alpha) - Q_{i}^{W}(\theta',\alpha) \right|
&\leq& \left|
g_{i}(\theta, \alpha) ' W_{i}^{-1} [g_{i}(\theta, \alpha) - g_{i}(\theta',\alpha) ] \right| \notag + \left| [g_{i}(\theta,\alpha) - g_{i}(\theta',\alpha)]'
W_{i}^{-1} g_{i}(\theta', \alpha) \right|
\notag \\
&\leq& 2 \cdot d_g^{2} E[M(z_{it})]^{2} \left| W_{i} \right|^{-1}
\left| \theta - \theta' \right|.
\end{eqnarray*}
\end{proof}

\subsection{Proof of Theorem 1}
\begin{proof}
\textbf{Part I: Consistency of $\tilde{\theta}$.} For any $\eta>0$, let $\varepsilon := \inf_{i}[ Q^{W}_{i}( \theta
_{0},\alpha _{i0}) -\sup_{\left\{ \left( \theta ,\alpha
\right) :\left| \left( \theta ,\alpha \right) -\left( \theta
_{0},\alpha _{i0}\right) \right|
>\eta \right\} } Q_{i}^{W}( \theta ,\alpha)] >0$ as defined in Condition \ref{cond2}. Using the standard argument for consistency of extremum estimator, as in Newey and McFadden (1994), with probability $1 -
o(T^{-1})$
\begin{eqnarray*}
\max_{\left| \theta -\theta _{0}\right| >\eta ,\alpha _{1},\ldots
,\alpha _{n}}n^{-1}\sum_{i=1}^{n}\hat{Q}_{i}^{W} \left( \theta
,\alpha _{i}\right) &<& n^{-1}\sum_{i=1}^{n}\hat{Q}_{i}^{W} \left( \theta _{0},\alpha
_{i0}\right) -\frac{1}{3}\varepsilon ,
\end{eqnarray*}
by definition of $\varepsilon$ and Lemma \ref{la8}. Thus, by continuity of $\hat{Q}_{i}^{W}$ and the definition of the lefthand side
above, we conclude that $\Pr \left[ \left| \tilde{\theta} -\theta _{0}\right| \geq \eta \right] = o\left( T^{-1} \right)$.

\textbf{Part II: Consistency of $\tilde{\alpha}_{i}$.}   By Part I and Lemma \ref{la8},
\begin{equation}
\Pr \left[ \sup_{1 \leq i \leq n}\sup_{\alpha }\left| \hat{Q}_{i}^{W}
\left( \tilde{\theta},\alpha \right) - Q_{i}^{W} \left( \theta
_{0},\alpha \right) \right| \geq \eta \right] = o(T^{-1}) \  \ \
\label{uniform-in-hat-theta}
\end{equation}
for any $\eta >0$. Let
\begin{equation*}
\varepsilon := \inf_{i} \left[ Q_{i}^{W} \left( \theta _{0},\alpha
_{i0}\right) -\sup_{\left\{ \alpha _{i}:\left| \alpha _{i}-\alpha
_{i0}\right| >\eta \right\} } Q_{i}^{W} \left( \theta _{0},\alpha
_{i}\right) \right]
>0.
\end{equation*}
Condition on the event
\begin{equation*}
\left\{ \sup_{1 \leq i \leq n}\sup_{\alpha }\left| \hat{Q}_{i}^{W}
\left( \tilde{\theta},\alpha \right) - Q_{i}^{W} \left( \theta
_{0},\alpha \right) \right| \leq \frac{1}{3}\varepsilon \right\},
\end{equation*}
which has a probability equal to $%
1- o\left( T^{-1} \right)  $ by (\ref{uniform-in-hat-theta}). Then
\begin{equation*}
\max_{\left| \alpha _{i}-\alpha _{i0}\right| >\eta
}\hat{Q}_{i}^{W}\left( \tilde{\theta},\alpha _{i}\right)
<\max_{\left| \alpha _{i}-\alpha _{i0}\right| >\eta } Q_{i}^{W}
\left( \theta _{0},\alpha _{i}\right) +\frac{1}{3}\varepsilon
<Q_{i}^{W} \left( \theta _{0},\alpha _{i0}\right)
-\frac{2}{3}\varepsilon <\hat{Q}_{i}^{W} \left( \tilde{\theta},\alpha _{i0}\right) -\frac{1}{3}%
\varepsilon.
\end{equation*}
This is inconsistent with $\hat{Q}_{i}^{W} \left( \tilde{%
\theta },\tilde{\alpha}_{i}\right) \geq \hat{Q}_{i}^{W} \left(
\tilde{\theta},\alpha _{i0}\right) $, and therefore, $\left|
\tilde{\alpha}_{i}-\alpha _{i0}\right| \leq \eta $ with probability $1 - o(T^{-1})$ for every $i$.
\\
\textbf{Part III: Consistency of $\tilde{\lambda}_{i}$.} First, note
that
\begin{eqnarray*}
\left| \tilde{\lambda}_{i} \right| &=& \left|\hat W_{i}^{-1}
\hat{g}_{i}(\tilde{\theta},\tilde{\alpha}_{i}) \right| \leq d_g
\left|\hat W_{i} \right|^{-1} \max_{1 \leq k \leq d_g} \left(\left|
\hat{g}_{k,i} (\tilde{\theta},\tilde{\alpha}_{i}) - g_{k,i}
(\tilde{\theta},\tilde{\alpha}_{i})\right| + \left|
g_{k,i} (\tilde{\theta},\tilde{\alpha}_{i}) \right|\right) \\
&\leq& d_g \left|\hat W_{i} \right|^{-1} \max_{1 \leq k \leq d_g} 
\sup_{(\theta,\alpha_{i}) \in \Upsilon} \left| \hat{g}_{k,i}
(\theta,\alpha_{i}) - g_{k,i} (\theta,\alpha_{i})\right| 
\\
&+& d_g \left|\hat W_{i} \right|^{-1} M(z_{it}) \left| \tilde{\theta} -
\theta_{0} \right| + d_g \left|\hat W_{i} \right|^{-1} M(z_{it}) \left|
\tilde{\alpha}_{i} - \alpha_{i0} \right|.
\end{eqnarray*}
Then, the result follows because $\sup_{ 1 \leq i \leq n} |\hat W_i - W_i| = o_P(1)$ and $\{W_i: 1 \leq i \leq n\}$ are positive definite by Condition \ref{cond2}, 
$\max_{1 \leq k \leq d_g} 
\sup_{(\theta,\alpha_{i}) \in \Upsilon} \left| \hat{g}_{k,i}
(\theta,\alpha_{i}) - g_{k,i} (\theta,\alpha_{i})\right|  = o_P(1)$ by Lemma 4 in HK,  and $\left| \tilde{\theta} -
\theta_{0} \right| = o_P(1)$ and $\sup_{1 \leq i \leq n} \left|
\tilde{\alpha}_{i} - \alpha_{i0} \right| = o_P(1)$ by Parts I and II.
\end{proof}


\subsection{Proof of Theorem 3}
\begin{proof}
First, assume that Conditions \ref{cond1}, \ref{cond2}, \ref{cond3} and \ref{cond4} hold.
The proofs are exactly the same as that of Theorem \ref{th1} using the uniform convergence of the criterion function.

To establish the uniform convergence of the criterion function as in
Lemma \ref{la8}, we need
\begin{equation*}
\sup_{1 \leq i \leq n} \left|
\hat{\Omega}_{i}(\tilde{\theta},\tilde{\alpha}_{i}) -
\Omega_{i}(\theta_{0},\alpha_{i0}) \right| =
o_P(1),
\end{equation*}
along with an extended version of the continuous mapping theorem for $o_{uP}$.
This can be shown by noting that
\begin{eqnarray*}
\left| \hat{\Omega}_{i}(\tilde{\theta},\tilde{\alpha}_{i}) -
\Omega_{i}(\theta_{0},\alpha_{i0}) \right| \leq \left|
\hat{\Omega}_{i}(\tilde{\theta},\tilde{\alpha}_{i}) -
\Omega_{i}(\tilde{\theta},\tilde{\alpha}_{i}) \right| + \left|
\Omega_{i}(\tilde{\theta},\tilde{\alpha}_{i}) -
\Omega_{i}(\theta_{0},\alpha_{i0}) \right| \\
\leq \left| \hat{\Omega}_{i}(\tilde{\theta},\tilde{\alpha}_{i}) -
\Omega_{i}(\tilde{\theta},\tilde{\alpha}_{i}) \right| + d_g^{2} E
\left[ M(z_{it})^{2} \right] \left|
(\tilde{\theta},\tilde{\alpha}_{i}) - (\theta_{0},\alpha_{i0})
\right|.
\end{eqnarray*}
The convergence follows by the consistency of $\tilde\theta$ and $\tilde{\alpha}_i$'s, and the application of Lemma 2 of HK to
$g_{k}(z_{it};\theta,\alpha_i) g_{l}(z_{it};\theta,\alpha_i)$ using that $\left| g_{k}(z_{it};\theta,\alpha_i) g_{l}(z_{it};\theta,\alpha_i) \right| \leq M(z_{it})^2$.
\end{proof}


\section{Asymptotic Distribution of One-step and Two-step FE-GMM Estimator}
\label{appendixB}

\subsection{Some Lemmas}

\begin{lemma}
\label{lb4} Assume that Condition \ref{cond1} holds.  Let $h(z_{it};\theta,\alpha_{i})$ be a
function such that (i) $h(z_{it};\theta,\alpha_{i})$ is 
continuously differentiable in $(\theta,\alpha_{i}) \in \Upsilon \subset \mathbb{R}^{d_{\theta}+d_{\alpha}}$;
(ii) $\Upsilon$ is convex; (iii) there exists a function
$M(z_{it})$ such that $\left| h(z_{it};\theta,\alpha_{i}) \right|
\leq M(z_{it})$ and $\left| \partial h(z_{it};\theta,\alpha_{i}) / \partial (\theta, \alpha_i) \right|
\leq M(z_{it})$ with $E \left[ M(z_{it})^{5(d_{\theta}+d_{\alpha}+6)/(1-10v) + \delta}
\right] < \infty$ for some $\delta > 0$ and $0 < v < 1/10$.  Define $\hat{H}_{i}(\theta,\alpha_{i}) :=
T^{-1} \sum_{t=1}^{T} h(z_{it};\theta,\alpha_{i})$, and
$H_{i}(\theta,\alpha_{i}) := E \left[ \hat{H}_{i}(\theta,\alpha_{i})
\right]$. Let
\begin{equation*}
\alpha_{i}^{*}  = \arg \max_{\alpha_i} \hat{Q}_{i}^{W} (\theta^{*},
\alpha_i),
\end{equation*}
such that $\alpha_{i}^{*} - \alpha_{i0} = o_{uP}(T^{a_{\alpha}})$
and $\theta^{*} - \theta_{0} = o_{P} (T^{a_{\theta}})$, with
$-2/5 \leq a \leq 0$, for $a =
\max(a_{\alpha},a_{\theta})$.  Then, for any $\overline{\theta}$ between $\theta^{*}$ and
$\theta_{0}$, and $\overline{\alpha}_{i}$ between $\alpha_{i}^{*}$ and
$\alpha_{i0}$, 
\begin{equation*}
\sqrt{T}[\hat{H}_{i}(\overline{\theta},\overline{\alpha}_{i}) -
 H_{i}(\overline{\theta},\overline{\alpha}_{i})] = o_{uP}(T^{1/10}), \ \  \hat{H}_{i}(\overline{\theta},\overline{\alpha}_{i}) - H_{i}(\theta_{0},\alpha_{i0})  = o_{uP}(T^{a}).
\end{equation*}
\end{lemma}
\begin{proof}
The first statement follows from Lemma 2 in HK. The second statement follows by the first statement and the conditions of the Lemma by a mean value expansion since 
\begin{eqnarray*}
\left| \hat{H}_{i}(\overline{\theta},\overline{\alpha}_{i}) -
 H_{i}(\theta_{0},\alpha_{i0})  \right|
 & \leq &   \underset{= o_{uP}(T^{a})}{\underbrace{\left| \overline{\theta} - \theta_{0} \right|}}
 \underset{= O_{uP}(1)}{\underbrace{\left| \frac{1}{T} \sum_{t=1}^{T} M(z_{it})
 \right|}}
 + \underset{= o_{uP}(T^{a})}{\underbrace{\left| \overline{\alpha}_{i} - \alpha_{i0} \right| }}
 \underset{= O_{uP}(1)}{\underbrace{\left| \frac{1}{T} \sum_{t=1}^{T} M(z_{it})
 \right|}}
 \\ &+& \underset{= o_{uP}\left(T^{- 2/5}\right)}{\underbrace{\left|   \hat{H}_{i}\left(\theta_{0} ,\alpha _{i0} \right)
- H_{i}\left( \theta _{0},\alpha _{i0}\right)
\right|}} = o_{uP}(T^{a}).
\end{eqnarray*}
\end{proof}


\begin{lemma}
\label{lb6} Assume that Conditions \ref{cond1}, \ref{cond2},
\ref{cond3} and \ref{cond5} hold. Let $\hat{t}_{i}^{W}(\theta, \gamma_{i})$ denote the first stage GMM score of the fixed effects, that
is
\begin{equation*}
\hat{t}_{i}^{W}(\theta, \gamma_{i}) =  - \left(%
\begin{array}{c}
  \hat{G}_{\alpha_{i}}(\theta, \alpha_{i})'\lambda_{i} \\
  \hat{g}_{i}(\theta,\alpha_{i}) + \hat W_{i} \lambda_{i} \\
\end{array}%
\right),
\end{equation*}
where $\gamma_{i} = (\alpha_{i}',\lambda_{i}')'$, 
$\hat{s}_{i}^{W} (\theta, \gamma_{i})$ denote the
one-step GMM score for the common parameter, that is
\begin{equation*}
\hat{s}_{i}^{W} (\theta, \gamma_{i}) =  -
\hat{G}_{\theta
i}(\theta,\alpha_{i})'\lambda_{i},
\end{equation*}
and  $\tilde{\gamma}_{i}(\theta)$ be such that
$\hat{t}_{i}^{W}(\theta, \tilde{\gamma}_{i}(\theta)) = 0$.

Let $\hat{T}_{i,j}^{W} (\theta, \gamma_{i})$ denote $\partial
\hat{t}_{i}^{W} (\theta, \gamma_{i}) / \partial \gamma_{i}' \partial \gamma_{i,j}$ and
$\hat{M}_{i,j}^{W} (\theta, \gamma_{i})$ denote $\partial
\hat{s}_{i}^{W} (\theta, \gamma_{i}) / \partial
\gamma_{i}' \partial \gamma_{i,j}$, for some $0 \leq j \leq d_g + d_{\alpha},$  where $\gamma_{i,j}$ is the $j$th element of $\gamma_{i}$ and $j=0$ denotes no second derivative. Let
$\hat{N}_{i}^{W} (\theta, \gamma_{i})$ denote $\partial
\hat{t}_{i}^{W}(\theta, \gamma_{i})/\partial \theta'$ and
$\hat{S}_{i}^{W} (\theta, \gamma_{i})$ denote $ \partial
\hat{s}_{i}^{W} (\theta, \gamma_{i}) / \partial \theta'$. Let
$(\tilde{\theta},\tilde{\gamma}_{1}, \ldots, \tilde{\gamma}_{n})$ be the one-step
GMM estimator. Then, for any $\overline{\theta}$ between
$\tilde{\theta}$ and $\theta_{0}$, and $\overline{\gamma}_{i}$
between $\tilde{\gamma}_{i}$ and $\gamma_{i0}$, 
\begin{equation*}
\hat{T}_{i,j}^{W} (\overline{\theta},\overline{\gamma}_{i}) -
T_{i,j}^{W} = o_{uP} \left( 1 \right),  \ \hat{M}_{i,j}^{W}
(\overline{\theta},\overline{\gamma}_{i}) - M_{i,j}^{W} = o_{uP}
\left( 1 \right), \ \hat{N}_{i}^{W}
(\overline{\theta},\overline{\gamma}_{i}) -
 N_{i}^{W}  = o_{uP} \left( 1 \right), \
\hat{S}_{i}^{W} (\overline{\theta},\overline{\gamma}_{i}) -
 S_{i}^{W}  = o_{uP} \left( 1 \right).
\end{equation*}
Also, for any $\overline{\gamma}_{i0}$ between $\gamma_{i0}$ and $\tilde \gamma_{i0} = \tilde \gamma_{i}(\theta_0),$
\begin{eqnarray*}
& &\sqrt{T} \hat{t}_{i}^{W} (\theta_0, \overline{\gamma}_{i0}) = o_{uP} \left(
T^{1/10} \right), \sqrt{T} \left( \hat{T}_{i,j}^{W} (\theta_0, \overline{\gamma}_{i0}) -
T_{i,j}^{W} \right) = o_{uP} \left( T^{1/10} \right), \\
& & \sqrt{T} \left( \hat{M}_{i,j}^{W} (\theta_0, \overline{\gamma}_{i0}) -
M_{i,j}^{W} \right) = o_{uP} \left( T^{1/10} \right), \
\end{eqnarray*}
\end{lemma}

\begin{proof}
 The first set of  results follows by inspection of the scores and their derivatives (the expressions
 are given in Appendix \ref{appendixF}),  uniform consistency of $\tilde{\gamma}_{i}$
 by Theorem \ref{th1} and application of the first part of Lemma \ref{lb4} to $\theta^{*} = \tilde{\theta}$ and $\alpha_{i}^{*} = \tilde{\alpha}_{i}$ with $a=0$.

The following steps are used to prove the second set of result.  By Lemma
\ref{lb4}, 
\begin{equation*}
\sqrt{T} \hat{t}_{i}^{W}  = o_{uP} \left(
T^{1/10} \right), \ \   \hat{T}_{i}^{W}
(\theta_0, \overline{\gamma}_{i0}) -
 T_{i}^{W} = o_{uP} \left(
1 \right)
\end{equation*}
where $\overline{\gamma}_{i0}$ is between $\tilde{\gamma}_{i0}$ and
$\gamma_{i0}$. Then, a mean value expansion of the FOC of
$\tilde{\gamma}_{i0},$  $\hat{t}_{i}^{W}( \theta_{0}, \tilde \gamma_{i0}) = 0$, around $\tilde \gamma_{i0} = \gamma_{i0}$ gives 
\begin{eqnarray*}
\sqrt{T} (\tilde{\gamma}_{i0} - \gamma_{i0}) &=& -
\underset{=O_{u}(1)}{\underbrace{\left( T_{i}^{W} \right)^{-1}}}
\underset{= o_{uP}(T^{1/10})}{\underbrace{\sqrt{T}
\hat{t}_{i}^{W}}} -
\underset{=O_{u}(1)}{\underbrace{\left( T_{i}^{W} \right)^{-1}}}
\underset{=o_{uP}(1)}{\underbrace{\left(
\hat{T}_{i}^{W}(\theta_0, \overline{\gamma}_{i0}) - T_{i}^{W} \right)}}
\sqrt{T} (\tilde{\gamma}_{i0} - \gamma_{i0}) \\
&=& o_{uP}(T^{1/10}) + o_{uP}\left(\sqrt{T} (\tilde{\gamma}_{i0} -
\gamma_{i0})\right),
\end{eqnarray*}
by Condition \ref{cond3} and the previous result. Therefore,
\begin{equation*}
\left( 1 + o_{uP}(1) \right) \sqrt{T} (\tilde{\gamma}_{i0} -
\gamma_{i0}) = o_{uP}(T^{1/10}) \Rightarrow \sqrt{T}
(\tilde{\gamma}_{i0} - \gamma_{i0}) = o_{uP}(T^{1/10}).
\end{equation*}

Given this uniform rate for $\tilde{\gamma}_{i0}$, the desired
result can be obtained by applying the second part of Lemma
\ref{lb4} to $\theta^{*} = \theta_{0}$ and $\alpha_{i}^{*} =
\tilde{\alpha}_{i0}$ with $a = -2/5$.
\end{proof}


\subsection{Proof of Theorem 2}

\begin{proof}
By a mean value expansion of the FOC for $\tilde{\theta}$ around $\tilde \theta = \theta_0,$
\begin{eqnarray*}
0 = \hat s^{W}(\tilde{\theta}) = \hat s^{W}(\theta_{0}) + \frac{d
\hat s^{W}(\overline{\theta})}{d \theta'} (\tilde{\theta} -
\theta_{0}),
\end{eqnarray*}
where $\overline{\theta}$ lies between $\tilde{\theta}$ and
$\theta_{0}$.

\textbf{Part I:} Asymptotic limit of $d
\hat s^{W}(\overline{\theta})/ d \theta'$. Note that
\begin{eqnarray}
\frac{d \hat s^{W}(\overline{\theta})}{d \theta'} &=& \frac{1}{n}
\sum_{i=1}^{n} \frac{d \hat s_{i}^{W}(\overline{\theta},
\tilde{\gamma}_{i}(\overline{\theta}))}{d \theta'}, \notag \\
\frac{d \hat s_{i}^{W}(\overline{\theta},
\tilde{\gamma}_{i}(\overline{\theta}))}{d \theta'} &=&
\frac{\partial \hat s_{i}^{W}(\overline{\theta},
\tilde{\gamma}_{i}(\overline{\theta}))}{\partial \theta'} +
\frac{\partial \hat s_{i}^{W}(\overline{\theta},
\tilde{\gamma}_{i}(\overline{\theta}))}{\partial \gamma_{i}'}
\frac{\partial \tilde{\gamma}_{i} (\overline{\theta})}{\theta'}.
\label{fth3_1}
\end{eqnarray}
By  Lemma \ref{lb6},
\begin{equation*}
\frac{\partial \hat s_{i}^{W}(\overline{\theta},
\tilde{\gamma}_{i}(\overline{\theta}))}{\partial \theta'} =
S_{i}^{W} + o_{uP}(1) , \ \ \frac{\partial
\hat s_{i}^{W}(\overline{\theta},
\tilde{\gamma}_{i}(\overline{\theta}))}{\partial \gamma_{i}'} =
M_{i}^{W} + o_{uP}(1).
\end{equation*}
Then, differentiation of the FOC for $\tilde{\gamma}_{i}(\overline{\theta})$, \
$\hat{t}_{i}^{W}
(\overline{\theta}, \tilde{\gamma}_{i}(\overline{\theta})) = 0$, with
respect to $\theta$ and $\tilde{\gamma}_{i}$ gives
\begin{eqnarray*}
\hat{T}_{i}^{W} (\overline{\theta}, \tilde{\gamma}_{i}(\overline{\theta})) \frac{\partial \tilde{\gamma}_{i}
(\overline{\theta})}{\partial \theta'} + \hat{N}_{i}^{W}
(\overline{\theta}, \tilde{\gamma}_{i}(\overline{\theta})) = 0,
\end{eqnarray*}
By repeated application of Lemma \ref{lb6} and
Condition \ref{cond3}, 
\begin{eqnarray*}
\frac{\partial \tilde{\gamma}_{i} (\overline{\theta})}{\partial
\theta'} = - \left( T_{i}^{W} \right)^{-1} N_{i}^{W} + o_{uP}(1).
\end{eqnarray*}
Finally, replacing the expressions for the components in
(\ref{fth3_1}) and using the formulae for the derivatives, which are
provided in the Appendix \ref{appendixF},
\begin{equation}
\frac{d \hat s^{W}(\overline{\theta})}{d \theta'} = \frac{1}{n} \sum_{i=1}^{n} G_{\theta_{i}}'
P_{\alpha_{i}}^{W} G_{\theta_{i}} + o_{P}(1) = J_{s}^{W} + o_{P}(1),
\ J_{s}^{W} = \bar E [ G_{\theta_{i}}'
P_{\alpha_{i}}^{W} G_{\theta_{i}}] .
\label{fth3_2} \end{equation}

\textbf{Part II:} Asymptotic Expansion for $\tilde{\theta} -
\theta_{0}$. By (\ref{fth3_2}) and
Lemma \ref{lg3}, which states the stochastic expansion of $\sqrt{nT}
\hat{s}^{W}(\theta_{0})$, 
\begin{eqnarray*}
0 &=& \underset{O_{P}(1)}{\underbrace{\sqrt{nT} \hat s^{W}
(\theta_{0})}} + \underset{O_{P}(1)}{\underbrace{\frac{d
\hat s^{W}(\overline{\theta})}{d \theta'}}} \sqrt{nT} (\tilde{\theta}
- \theta_{0}).
\end{eqnarray*}
Therefore, $\sqrt{nT} (\tilde{\theta} - \theta_{0}) = O_{P}(1),$ and by part I, Lemma \ref{lg3} and Condition \ref{cond3},
$$\sqrt{nT}(\tilde{\theta}-\theta_{0}) \ind - (J_{s}^{W})^{-1} N \left( \kappa B_{s}^{W}, V_{s}^{W} \right).$$

\end{proof}

\setcounter{theorem}{2}




\subsection{Proof of Theorem 4}
\begin{proof}
Applying Lemma \ref{lb4} with a minor modification, along with
Condition \ref{cond5}, we can prove an exact counterpart to Lemma
\ref{lb6} for the two-step GMM score for the fixed effects
\begin{equation*}
\hat{t}_{i}(\theta, \gamma_{i}) = \hat{t}_{i}^{\Omega}(\theta, \gamma_{i}) + \hat{t}_{i}^{R}(\theta, \gamma_{i}),
\end{equation*}
where the expressions of $\hat{t}_{i}^{\Omega}$ and
$\hat{t}_{i}^{R}$ are given in the Appendix \ref{appendixF}, and for
the two-step score of the common parameter
\begin{equation*}
\hat{s}_{i} (\theta, \hat{\gamma}_{i}(\theta)) =  - \hat{G}_{\theta
i}(\theta,\hat{\alpha}_{i}(\theta))'\hat{\lambda}_{i}(\theta),
\end{equation*}
The only difference arises due to the term
$\hat{t}_{i}^{R}(\theta, \gamma_{i})$, which involves
$\hat{\Omega}_{i}(\tilde{\theta},\tilde{\alpha}_{i}) - \Omega_{i}$.
Lemma \ref{le3} shows that $\sqrt{T}(\hat{\Omega}_{i}(\tilde{\theta},\tilde{\alpha}_{i}) - \Omega_{i}) = o_{uP}(T^{1/10}),$
so that a result similar to Lemma \ref{lb6} holds for the
two-step scores.

Thus, we can make the same argument as in the proof of Theorem 2 using
the stochastic expansion of $\sqrt{nT} \hat{s}(\theta_{0})$
given  in Lemma \ref{li1}.
\end{proof}



\section{Asymptotic Distribution of Bias-Corrected Two-Step GMM
Estimator}\label{appendixC}

\subsection{Some Lemmas}


\begin{lemma}
\label{le2} Assume that Conditions \ref{cond1}, \ref{cond2},
\ref{cond3}, \ref{cond5}  and \ref{cond4} hold. Let
$\hat{t}_{i}(\theta, \gamma_{i})$ denote the two-step GMM score for
the fixed effects,  $\hat{s}_{i} (\theta,
\gamma_{i})$ denote the two-step GMM score for the
common parameter, and  $\hat{\gamma}_{i}(\theta)$ be such that
$\hat{t}_{i}(\theta, \hat{\gamma}_{i}(\theta)) = 0$. Let
$\hat{T}_{i,j} ( \theta,\gamma_{i})$ denote $\partial
\hat{t}_{i} ( \theta,\gamma_{i})/\partial \gamma_{i}' \partial \gamma_{i,j}$, for some
$0 \leq j \leq d_g + d_{\alpha},$ where $\gamma_{i,j}$ is the $j$th component of $\gamma_{i}$ and $j=0$ denotes no second derivative. Let $\hat{N}_{i} ( \theta,\gamma_{i})$ denote
$\partial \hat{t}_{i} ( \theta,\gamma_{i})/\partial \theta'$. Let
$\hat{M}_{i,j} ( \theta,\gamma_{i})$ denote $\partial
\hat{s}_{i} ( \theta,\gamma_{i})/\partial
\gamma_{i}' \partial \gamma_{i,j}$ , for some $0 \leq j \leq d_g + d_{\alpha}$. Let
$\hat{S}_{i} ( \theta,\gamma_{i})$ denote $ \partial
\hat{s}_{i} ( \theta,\gamma_{i}) / \partial \theta'$. Let
$(\hat{\theta},\{\hat{\gamma}_{i} \}_{i=1}^{n})$ be the two-step GMM
estimators.

Then, for any $\overline{\theta}$ between $\hat{\theta}$ and
$\theta_{0}$, and $\overline{\gamma}_{i}$ between $\hat{\gamma}_{i}$
and $\gamma_{i0}$,
\begin{eqnarray*}
\sqrt{T} \left( \hat{T}_{i,d}
(\overline{\theta},\overline{\gamma}_{i}) - T_{i,d} \right) &=&
o_{uP} \left( T^{1/10} \right), \ \ \sqrt{T} \left( \hat{M}_{i,j}
(\overline{\theta},\overline{\gamma}_{i}) -
M_{i,j} \right) = o_{uP} \left( T^{1/10} \right),\\
\sqrt{T} \left( \hat{N}_{i}
(\overline{\theta},\overline{\gamma}_{i}) -
 N_{i} \right) &=& o_{uP} \left( T^{1/10} \right), \ \
\sqrt{T} \left( \hat{S}_{i}
(\overline{\theta},\overline{\gamma}_{i}) -
 S_{i}  \right) = o_{uP} \left( T^{1/10} \right).
\end{eqnarray*}

\end{lemma}

\begin{proof}
Let $\hat{\gamma}_i = \hat{\gamma}_i(\hat{\theta})$ and
$\hat{\gamma}_{i0}  = \hat{\gamma}_i(\theta_0)$. First, note that
\begin{eqnarray*}
\sqrt{T} (\hat{\gamma}_{i} - \hat{\gamma}_{i0}) = \frac{\partial
\hat{\gamma}_{i}(\overline{\theta})}{\partial \theta'} \sqrt{T}
(\hat{\theta} - \theta_{0}) = - \underset{=
O_{u}(1)}{\underbrace{\left( T_{i}^{\Omega} \right)^{-1} N_{i}}}
\underset{=O_{P}(n^{-1/2})}{\underbrace{\sqrt{T} (\hat{\theta} -
\theta_{0})}} + o_{uP} \left(\sqrt{T} (\hat{\theta} - \theta_{0})
\right)= O_{uP}(n^{-1/2}).
\end{eqnarray*}
where the second equality follows from the proof of Theorem 2 and 4.
Thus, by the same argument used in the proof of Lemma \ref{lb6}, 
\begin{equation*}
\sqrt{T} (\hat{\gamma}_{i} - \gamma_{i0}) = \sqrt{T} (\hat{\gamma}_{i} - \hat{\gamma}_{i0}) + \sqrt{T} (\hat{\gamma}_{i0} - \gamma_{i0}) = o_{uP}(T^{1/10}).
\end{equation*}
Given this result and inspection of the scores and their
derivatives (see the Appendix \ref{appendixF}), the proof is similar
to the proof of the second part of Lemma \ref{lb6}.
\end{proof}


\begin{lemma}
\label{ld0}  Assume that Condition \ref{cond1} holds.  Let $h_{j}(z_{it};\theta,\alpha_{i}), \ j = 1,2$ be two
functions such that (i) $h_{j}(z_{it};\theta,\alpha_{i})$ is continuously differentiable in $(\theta,\alpha_{i}) \in
\Upsilon \subset \mathbb{R}^{d_{\theta} + d_{\alpha}}$; (ii) $\Upsilon$ is convex; (iii) there exists a function
$M(z_{it})$ such that $\left| h_{j}(z_{it};\theta,\alpha_{i})
\right| \leq M(z_{it})$ and $\left| \partial h_{j}(z_{it};\theta,\alpha_{i})/ \partial (\theta, \alpha_i)
\right| \leq M(z_{it})$ with $E \left[ M(z_{it})^{10(d_{\theta}+d_{\alpha}+6)/(1-10v) + \delta}
\right] < \infty$ for some $\delta > 0$ and $0 < v < 1/10$.  Define $\hat{F}_{i}(\theta,\alpha_{i}) :=
T^{-1} \sum_{t=1}^{T} h_1(z_{it};\theta,\alpha_{i})h_2(z_{it};\theta,\alpha_{i})$, and
$F_{i}(\theta,\alpha_{i}) := E \left[ \hat{F}_{i}(\theta,\alpha_{i})
\right]$. Let
\begin{equation*}
\alpha_{i}^{*}  = \arg \sup_{\alpha} \hat{Q}_{i}^{W} (\theta^{*},
\alpha),
\end{equation*}
such that $\alpha_{i}^{*} - \alpha_{i0} = o_{uP}(T^{a_{\alpha}})$
and $\theta^{*} - \theta_{0} = o_{P} (T^{a_{\theta}})$, with
$-2/5 \leq a \leq 0$, for $a =
\max(a_{\alpha},a_{\theta})$.  Then, for any $\overline{\theta}$ between $\theta^{*}$ and
$\theta_{0}$, and $\overline{\alpha}_{i}$ between $\alpha_{i}^{*}$ and
$\alpha_{i0}$, 
\begin{equation*}
\hat{F}_{i}(\overline{\theta},\overline{\alpha}_{i}) -
 F_{i}(\theta_{0},\alpha_{i0})  = o_{uP}(T^{a}), \ \ \sqrt{T}[\hat{F}_{i}(\overline{\theta},\overline{\alpha}_{i}) -
 F_{i}(\overline{\theta},\overline{\alpha}_{i})] = o_{uP}(T^{1/10}).
\end{equation*}
\end{lemma}

\begin{proof}
Same as for Lemma \ref{lb4}, replacing $H_{i}$ by $F_{i}$, and
$M(z_{it})$ by $M(z_{it})^{2}$.
\end{proof}


\begin{lemma}
\label{le3} Assume that Conditions \ref{cond1}, \ref{cond2},
\ref{cond3}, \ref{cond5}, \ref{cond4}, and \ref{cond7} hold. Let 
$\hat{\Omega}_{i}(\overline{\theta}, \overline{\alpha}_{i}) = T^{-1} \sum_{t=1}^{T} g(z_{it};\overline{\theta},
\overline{\alpha}_{i}) g(z_{it};\overline{\theta},
\overline{\alpha}_{i})'$  be
an estimator of the covariance function $\Omega_i = E[g(z_{it})g(z_{it})'],$
where $\overline{\theta}= \theta_0 + o_P(T^{-2/5})$ and $\overline{\alpha}_{i} = \alpha_{i0} + o_{uP}(T^{-2/5})$. Let $\hat{\Omega}_{\alpha^{d_{1}} \theta^{d_{2}}_i}
(\overline{\theta}, \overline{\alpha}_{i}) = \partial^{d_{1}+d_{2}} \hat{\Omega}_{i}(\overline{\theta},
\overline{\alpha}_{i}) / \partial^{d_{1}} \alpha_{i}
\partial^{d_{2}} \theta,$
for $0 \leq d_{1} + d_{2} \leq 2$. Then, 
\begin{eqnarray*}
\sqrt{T} \left( \hat{\Omega}_{\alpha^{d_{1}} \theta^{d_{2}} _i}
(\overline{\theta}, \overline{\alpha}_{i}) - \Omega_{\alpha^{d_{1}}
\theta^{d_{2}} _i} \right)&=& o_{up} \left( T^{1/10}
\right).
\end{eqnarray*}
\end{lemma}

\begin{proof}
Note that
\begin{multline*}
\left| g(z_{it};\overline{\theta}, \overline{\alpha}_{i})
g(z_{it};\overline{\theta}, \overline{\alpha}_{i})' - E \left[
g(z_{it};\overline{\theta}, \overline{\alpha}_{i})
g(z_{it};\overline{\theta}, \overline{\alpha}_{i})' \right]
\right| \\
\leq  d_g^{2} \max_{1 \leq k \leq l \leq d_g} \left|
g_{k}(z_{it};\overline{\theta}, \overline{\alpha}_{i})
g_{l}(z_{it};\overline{\theta}, \overline{\alpha}_{i})' - E \left[
g_{k}(z_{it};\overline{\theta}, \overline{\alpha}_{i})
g_{l}(z_{it};\overline{\theta}, \overline{\alpha}_{i})' \right]
\right|.
\end{multline*}
Then we can apply Lemma \ref{ld0} to $h_1 = g_{k}$ and $h_2 = g_{l}$ with $a = -2/5$. A similar argument applies to the derivatives, since they
are sums of products of elements that satisfy the assumption of 
Lemma \ref{ld0}.
\end{proof}


\begin{lemma}
\label{le5} Assume that Conditions \ref{cond1}, \ref{cond2},
\ref{cond3}, \ref{cond5}, \ref{cond4}, and \ref{cond7} hold, and $\ell \to \infty$ such that $\ell/T \to 0$ as $T \to \infty$. For any $\overline{\theta}$ 
 between
$\hat{\theta}$ and $\theta_{0},$ let
$\hat{\Sigma}_{\alpha_{i}}
\left(\overline{\theta}\right) =
\left[\hat{G}_{\alpha_{i}}\left(\overline{\theta}\right)'
\hat{\Omega}_{i}^{-1}
\hat{G}_{\alpha_{i}}\left(\overline{\theta}\right)
\right]^{-1},$  $\hat{H}_{\alpha_{i}}
\left(\overline{\theta}\right) =
\hat{\Sigma}_{\alpha_{i}}\left(\overline{\theta}\right)
\hat{G}_{\alpha_{i}}\left(\overline{\theta}\right)'
\hat{\Omega}_{i}^{-1},$
$\hat{P}_{\alpha_{i}}
\left(\overline{\theta}\right) =
\hat{\Omega}_{i}^{-1}
- \hat{\Omega}_{i}
^{-1}
\hat{G}_{\alpha_{i}}
\left(\overline{\theta}\right)
\hat{H}_{\alpha_{i}}
\left(\overline{\theta}\right),$ 
$\hat{\Sigma}_{\alpha_{i}}^{W}
\left(\overline{\theta}\right) =
\left[\hat{G}_{\alpha_{i}}
\left(\overline{\theta}\right)' W_{i}^{-1}
\hat{G}_{\alpha_{i}}
\left(\overline{\theta}\right)
\right]^{-1},$
$\hat{H}_{\alpha_{i}}^{W}
\left(\overline{\theta}\right) =
\hat{\Sigma}_{\alpha_{i}}^{W}
\left(\overline{\theta}\right)
\hat{G}_{\alpha_{i}}
\left(\overline{\theta}\right)' W_{i}^{-1},$
 $\hat{J}_{si}\left(\overline{\theta}\right) =
\hat{G}_{\theta_{i}}
\left(\overline{\theta}\right)'
\hat{P}_{\alpha_{i}}
\left(\overline{\theta}\right)
 \hat{G}_{\theta_{i}} \left(\overline{\theta}\right),$
$\hat{B}_{si}^C(\bar \theta) = T^{-1} \sum_{j=0}^{\ell} \sum_{t=j+1}^T \hat G_{\theta_{it}}(\bar \theta)'
\hat P_{\alpha_{i}}(\bar \theta) \hat g_{i,t-j}(\bar \theta),$ and $\hat B_{si}^{B}(\bar \theta) = - \hat G_{\theta_{i}}(\bar \theta)' [
\hat B_{\lambda_{i}}^{I}(\bar \theta)
 + \hat B_{\lambda_{i}}^{G}(\bar \theta) + \hat B_{\lambda_{i}}^{\Omega}(\bar \theta)+ \hat B_{\lambda_{i}}^{W}(\bar \theta)]$, where
\begin{footnotesize}\begin{eqnarray*}
\hat B_{\lambda_{i}}^{I}(\theta) &=& -  \hat P_{\alpha_{i}}(\theta) \sum_{j=1}^{d_{\alpha}} \hat G_{\alpha\alpha_{i,j}}(\theta)
\hat \Sigma_{\alpha_{i}}(\theta)/2
   + \hat P_{\alpha_{i}}(\theta) \sum_{j=0}^{\ell} T^{-1} \sum_{t=j+1}^T \hat G_{\alpha_{it}}(\theta) \hat H_{\alpha_{i}}(\theta) \hat g_{i,t-j}(\theta),\\
\hat B_{\lambda_{i}}^{G}(\theta) &=&  \hat H_{\alpha_{i}}(\theta)' \sum_{j=0}^{\infty}  T^{-1} \sum_{t=j+1}^T  \hat G_{\alpha_{it}}(\theta)' \hat P_{\alpha_{i}}(\theta) \hat g_{i,t-j}(\theta),   \\
\hat B_{\lambda_{i}}^{\Omega}(\theta) &=&   \hat P_{\alpha_{i}}(\theta) \sum_{j=0}^{\ell}  T^{-1} \sum_{t=j+1}^T \hat g_{it}(\theta) \hat g_{it}(\theta)'\hat P_{\alpha_{i}}(\theta) \hat g_{i,t-j}(\theta),\\
\hat B_{\lambda_{i}}^{W}(\theta) &=&   \hat P_{\alpha_{i}}(\theta) \sum_{j=1}^{d_{\alpha}} \hat \Omega_{\alpha_{i,j}}[ \hat H_{\alpha_{i,j}}^{W'}(\theta) - \hat H_{\alpha_{i,j}}^{'}(\theta)],
\end{eqnarray*}\end{footnotesize}be estimators of
$\Sigma_{\alpha_{i}}$, $H_{\alpha_i}, P_{\alpha_i},$ $\Sigma_{\alpha_{i}}^{W},$
$H_{\alpha_{i}}^{W},$ $J_{si},$ $B_{si}^{C}$
and $B_{si}^{B}.$
Let $\hat{F}_{\alpha^{d_{1}} \theta^{d_{2}}i}
(\theta, \hat{\alpha}_{i}(\theta))$ and $F_{\alpha^{d_{1}}
\theta^{d_{2}}i}(\theta, \alpha_i)$, with $F \in \left\{ \Sigma, H, P, \Sigma^{W},
H^{W}, J_{si}, B_{si}^C,  B_{si}^B \right\}$ denote their derivatives for $0 \leq
d_{1} + d_{2} \leq 1$. Then,
$$
\sqrt{T} \left( \hat{F}_{\alpha^{d_{1}} \theta^{d_{2}}i}
\left(\overline{\theta},\hat{\alpha}_{i}(\overline{\theta}) \right) -
F_{\alpha^{d_{1}} \theta^{d_{2}}i} \right) = o_{uP} \left(
T^{1/10} \right).
$$
where $F_{\alpha^{d_{1}} \theta^{d_{2}}i} := F$ if $d_1+d_2=0$.
\end{lemma}

\begin{proof}
The results follow by Theorem \ref{th4} and Lemma \ref{le2}, using the algebraic properties
of the $o_{uP}$ orders and Lemma 12 of HK to show the properties of the estimators of the spectral expectations.
\end{proof}


\begin{lemma}
\label{le8} Assume that Conditions \ref{cond1}, \ref{cond2},
\ref{cond3}, \ref{cond5}, \ref{cond4}, and \ref{cond7} hold. Then, for any $\overline{\theta}$ between $\hat{\theta}$ and
$\theta_{0},$
$$
\hat{J}_{s} \left( \overline{\theta} \right) =  J_{s} + o_P(T^{-2/5}).
$$
\end{lemma}

\begin{proof}
Note that 
$$
\sqrt{T} \left[\hat G_{\theta_i}(\overline \theta)' \hat P_{\alpha_i} (\overline \theta) \hat G_{\theta_i}(\overline \theta) - G_{\theta_i}'P_{\alpha_i}G_{\theta_i}\right] = o_{uP}(T^{1/10}),
$$
by Theorem \ref{th4} and Lemmas \ref{le2} and \ref{le5}, using the algebraic properties
of the $o_{uP}$ orders. The result then follows by a CLT for independent sequences since
$$
\hat J_s(\overline \theta) - J_s = \widehat{\bar{E}}[\hat G_{\theta_i}(\overline \theta)' \hat P_{\alpha_i} (\overline \theta) \hat G_{\theta_i}(\overline \theta)] - \bar{E}[G_{\theta_i}'P_{\alpha_i}G_{\theta_i}] = n^{-1} \sum_{i=1}^n \left(G_{\theta_i}'P_{\alpha_i}G_{\theta_i} - \bar{E}[G_{\theta_i}'P_{\alpha_i}G_{\theta_i}] \right) + o_{uP}(T^{-2/5}).
$$
\end{proof}

\begin{lemma}
\label{le9} Assume that Conditions \ref{cond1}, \ref{cond2},
\ref{cond3}, \ref{cond5}, \ref{cond4}, and \ref{cond7} hold. Then, for any $\overline{\theta}$ between $\hat{\theta}$ and
$\theta_{0},$
$$
\hat{B}_{s} \left( \tilde{\theta} \right) =  B_{s} + o_P(T^{-2/5}).
$$
\end{lemma}

\begin{proof}
Analogous to the proof of Lemma \ref{le8} replacing $J_{s}$ by
$B_{s}$.
\end{proof}

\begin{lemma}
\label{le11} Assume that Conditions \ref{cond1}, \ref{cond2},
\ref{cond3}, \ref{cond5}, \ref{cond4}, and \ref{cond7} hold. Then, for any $\overline{\theta}$ between $\hat{\theta}$ and
$\theta_{0},$ and $ \mathcal{B}  = - J_s^{-1} B_s,$
\begin{equation*}
\hat{\mathcal{B}} (\hat{\theta}) = - \hat{J}_{s} (\hat{\theta})^{-1}
\hat{B}_{s} (\hat{\theta}) = \mathcal{B} + o_{P}(T^{-2/5}).
\end{equation*}
\end{lemma}

\begin{proof}
The result follows from Lemmas \ref{le8} and \ref{le9}, using a Taylor expansion argument.
\end{proof}

\subsection{Proof of Theorem 5}

\begin{proof}
\textbf{Case I: C = BC}. By Lemmas \ref{le8} and \ref{li5}
\begin{eqnarray*}
\sqrt{nT} \left( \hat{\theta} - \theta_{0} \right) &=& - \hat J_{s}
\left( \overline{\theta} \right)^{-1} \hat{s} (\theta_{0}) = -
J_{s}^{-1} \hat{s} (\theta_{0}) + o_{P}(T^{-2/5}) O_{P} \left(
\sqrt{\frac{n}{T}} \right) = - J_{s}^{-1} \hat{s} (\theta_{0})
+ o_{P}(1).
\end{eqnarray*}
Then, by Lemmas \ref{le11} and \ref{li5}
\begin{eqnarray*}
\sqrt{nT} \left( \hat{\theta}^{BC} - \theta_{0} \right) &=&
\sqrt{nT} \left( \hat{\theta} - \theta_{0} \right) - \sqrt{nT}
\frac{1}{T} \hat{\mathcal{B}} \left( \hat{\theta} \right) = - J_{s}^{-1}
\hat{s} (\theta_{0}) + \sqrt{\frac{n}{T}}
J_{s}^{-1} B_{s}  + o_{P}(1) \\
&=& - J_{s}^{-1} \left[ \frac{1}{\sqrt{n}} \sum_{i=1}^{n}
\tilde{\psi}_{si} + \sqrt{\frac{n}{T}}  B_{s} - \sqrt{\frac{n}{T}}
B_{s} \right] + o_{P}(1) \ind N(0, J_{s}^{-1}).
\end{eqnarray*}
\textbf{Case II: C = SBC}. First, note that since the correction of
the score is of order $O_{P}(T^{-1})$, 
$\hat{\theta}^{SBC} - \hat \theta = O_P(T^{-1})$. Then, by a Taylor
expansion of the corrected FOC around $\hat{\theta}^{SBC} = \theta_0$
\begin{eqnarray*}
0 &=& \hat{s} \left( \hat{\theta}^{SBC} \right) - T^{-1}
\hat{B}_{s} \left( \hat{\theta}^{SBC} \right) = \hat{s}
(\theta_{0}) + \hat J_{s} \left( \overline{\theta} \right)
(\hat{\theta}^{SBC} - \theta_{0}) - T^{-1} B_{s} + o_{P}(T^{-2}),
\end{eqnarray*}
where $\overline{\theta}$ lies between $\hat{\theta}^{SBC}$ and
$\theta_{0}$. Then by Lemma \ref{li5}
\begin{eqnarray*}
\sqrt{nT} \left( \hat{\theta}^{SBC} - \theta_{0} \right) &=& -
\hat J_{s} \left( \overline{\theta} \right)^{-1} \left[ \sqrt{nT}
 \hat{s} (\theta_{0}) - n^{1/2} T^{-1/2} B_{s} \right]
 + o_{P}(1) \\
  &=& - \hat J_{s} \left( \overline{\theta} \right)^{-1}
 \left[ \frac{1}{\sqrt{n}} \sum_{i=1}^{n}
\tilde{\psi}_{si} +  \sqrt{\frac{n}{T}}   B_{s} -  \sqrt{\frac{n}{T}} 
B_{s} \right] + o_{P}(1) \ind N(0, J_{s}^{-1}).
\end{eqnarray*}

\textbf{Case III: C = IBC}. A similar argument applies to the
estimating equation (\ref{ibc}), since $\hat{\theta}^{IBC}$ is in a
$O(T^{-1})$ neighborhood of $\theta_{0}$.
\end{proof}

\section{Stochastic Expansion for $\tilde{\gamma}_{i0} = \tilde{\gamma}_{i}(\theta_{0})$ and $\hat{\gamma}_{i0} = \hat{\gamma}_{i}(\theta_{0})$} \label{appendixF}
We characterize the stochastic expansions up to second order for one-step and two-step
estimators of the individual effects given the true common
parameter. We only provide detailed proofs of the results for the two-step estimator $\hat{\gamma}_{i0},$ 
because the proofs  the one-step estimator
$\tilde{\gamma}_{i0}$ follow by similar arguments.  Lemmas 1
and 2 in the main text are corollaries of these expansions. The expressions for the scores and their derivatives in the components of the expansions are given in Appendix \ref{appendixF}.


\begin{lemma}
\label{lf2} Suppose that Conditions \ref{cond1}, \ref{cond2},
\ref{cond3}, and \ref{cond5} hold. Then
\begin{equation*}
\sqrt{T} (\tilde{\gamma}_{i0} - \gamma_{i0}) = \tilde{\psi}_{i}^{W}
+ T^{-1/2}  R_{1i}^{W} \ind N(0, V_i^W),
\end{equation*}
where
\begin{eqnarray*}
\tilde{\psi}_{i}^{W} &=& \frac{1}{\sqrt{T}} \sum_{t=1}^{T}
\psi_{it}^{W} = - \left( T_{i}^{W} \right) ^{-1} \sqrt{T}
\hat{t}_{i}^{W} = o_{uP}(T^{1/10}), \ \
R_{1i}^{W} = o_{uP}(T^{1/5}), \ \ V_i^W = E[\tilde{\psi}_{i}^{W}\tilde{\psi}_{i}^{W'}].
\end{eqnarray*}
Also
\begin{equation*}
\frac{1}{\sqrt{n}} \sum_{i=1}^{n} \tilde \psi_{i}^{W} = O_{P}(1).
\end{equation*}
\end{lemma}
\begin{proof}
We just show the part of the remainder term because the rest of the proof is similar to the proof of Lemma \ref{lh2}. By the proof of Lemma \ref{lb6}, $\sqrt{T} (\tilde{\gamma}_{i0} - \gamma_{i0}) = o_{uP}(T^{1/10})$ and 
$$
R_{1i}^{W} = -
\underset{=O_{u}(1)}{\underbrace{\left( T_{i}^{W} \right)^{-1}}}
\underset{=o_{uP}(T^{1/10})}{\underbrace{\left(
\hat{T}_{i}^{W}(\theta_0, \overline{\gamma}_{i0}) - T_{i}^{W} \right)}}
\underset{=o_{uP}(T^{1/10})}{\underbrace{\sqrt{T} (\tilde{\gamma}_{i0} - \gamma_{i0})}}
= o_{uP}(T^{1/5}).
$$ 
\end{proof}


\begin{lemma}
\label{lf3} Suppose that Conditions \ref{cond1}, \ref{cond2},
\ref{cond3}, and \ref{cond5} hold. Then,
\begin{equation*}
\sqrt{T} (\tilde{\gamma}_{i0} - \gamma_{i0}) = \tilde{\psi}_{i}^{W}
+T^{-1/2} Q_{1i}^{W} + T^{-1} R_{2i}^{W},
\end{equation*}
where
\begin{eqnarray*}
Q_{1i}^{W} &=& - \left( T_{i}^{W} \right) ^{-1} \left[
\tilde{A}_{i}^{W} \tilde{\psi}_{i}^{W} + \frac{1}{2}
\sum_{j=1}^{d_g+d_{\alpha}} \tilde{\psi}_{i,j}^{W} T_{i,j}^{W}
\tilde{\psi}_{i}^{W} 
\right] = o_{uP}(T^{1/5}), \\
\tilde{A}_{i}^{W} &=& \sqrt{T} (\hat{T}_{i}^{W} -
T_{i}^{W}) = o_{uP} (T^{1/10}),  \ \ 
R_{2i}^{W} = o_{uP}(T^{3/10}).
\end{eqnarray*}
Also,
\begin{equation*}
\frac{1}{n} \sum_{i=1}^{n}  Q_{1i}^{W} = O_P(1).
\end{equation*}
\end{lemma}

\begin{proof}
Similar to the proof of Lemma \ref{lh3}. 
\end{proof}


%


\begin{lemma}
\label{lf5} Suppose that Conditions \ref{cond1}, \ref{cond2},
\ref{cond3}, and \ref{cond5} hold. Then, 
$$
\frac{1}{\sqrt{n}} \sum_{i=1}^{n} \tilde  \psi_{i}^{W} \ind N(0, \bar{E}[V_i^{W}]), \
\ \frac{1}{n} \sum_{i=1}^{n}  Q_{1i}^{W}  \inp
\bar{E}[B_{\gamma_{i}}^{W,I} + B_{\gamma_{i}}^{W,G}
+ B_{\gamma_{i}}^{W,1S}] =:  B_{\gamma}^{W}, 
$$
where
\begin{eqnarray*}
V_i^{W} &=& \left(%
\begin{array}{cc}
  H_{\alpha_{i}}^{W}  \\
  P_{\alpha_{i}}^{W}  \\
\end{array}%
\right) \Omega_{i} \left(H_{\alpha_{i}}^{W'}, P_{\alpha_{i}}^{W}\right),
 \\
B_{\gamma_{i}}^{W,I} &=& \left(%
\begin{array}{c}
  B_{\alpha_{i}}^{W,I} \\
  B_{\lambda_{i}}^{W,I} \\
\end{array}%
\right) = \left(%
\begin{array}{c}
  H_{\alpha_{i}}^{W}  \\
  P_{\alpha_{i}}^{W}  \\
\end{array}%
\right) \left(\sum_{j=-\infty}^{\infty} E \left[ G_{\alpha_{i}}(z_{it}) H_{\alpha_{i}}^{W} g(z_{i,t-j}) \right] -    \sum_{j=1}^{d_{\alpha}} G_{\alpha \alpha_{i,j}} H_{\alpha_{i}}^{W} \Omega_{i} H_{\alpha_{i}}^{W'}/2\right), \\
B_{\gamma_{i}}^{W,G} &=& \left(%
\begin{array}{c}
  B_{\alpha_{i}}^{W,G} \\
  B_{\lambda_{i}}^{W,G} \\
\end{array}%
\right) = \left(%
\begin{array}{c}
  -\Sigma_{\alpha_{i}}^{W}
      \\
  H_{\alpha_{i}}^{W'}\\
\end{array}%
\right)\sum_{j=-\infty}^{\infty}  E \left[ G_{\alpha_{i}}(z_{it})' P_{\alpha_{i}}^{W} g(z_{i,t-j})\right], \\
B_{\gamma_{i}}^{W,1S} &=& \left(%
\begin{array}{c}
  B_{\alpha_{i}}^{W,1S} \\
  B_{\lambda_{i}}^{W,1S} \\
\end{array}%
\right) = \left(%
\begin{array}{c}
  \Sigma_{\alpha_{i}}^{W} \\
  -  H_{\alpha_{i}}^{W'}  \\
\end{array}%
\right) \left(\sum_{j=1}^{d_{\alpha}} G_{\alpha \alpha_{i,j}}' P_{\alpha_{i}}^{W} \Omega_{i}
  H_{\alpha_{i}}^{W'}/2
  +    \sum_{j=1}^{d_g}
   G_{\alpha \alpha_{i}}' (I_{d_{\alpha}} \otimes e_{j}) H_{\alpha_{i}}^{W} \Omega_{i} P_{\alpha_{i},j}^{W}/2\right),    \\ 
& & \quad  \quad \quad \quad \quad \quad + 
\left(%
\begin{array}{c}
   H_{\alpha_{i}}^{W}  \\
  P_{\alpha_{i}}^{W}   \\
\end{array}%
\right)\sum_{j=-\infty}^{\infty}  E \left[ \xi_i(z_{it}) P_{\alpha_{i}}^{W} g(z_{i,t-j}) \right],
\end{eqnarray*}
for $\Sigma_{\alpha_{i}}^{W} = \left( G_{\alpha_{i}}' W_{i}^{-1}
G_{\alpha_{i}} \right)^{-1},$ $H_{\alpha_{i}}^{W} =
\Sigma_{\alpha_{i}}^{W} G_{\alpha_{i}}' W_{i}^{-1},$ and
$P_{\alpha_{i}}^{W} = W_{i}^{-1} - W_{i}^{-1} G_{\alpha_{i}}
H_{\alpha_{i}}^{W}$.
\end{lemma}

\begin{proof}
The results follow from Lemmas \ref{lf2} and \ref{lf3}, noting that
\begin{eqnarray*}
\left( T_{i}^{W} \right)^{-1} &=& - \left(%
\begin{array}{cc}
  -\Sigma_{\alpha_{i}}^{W} & H_{\alpha_{i}}^{W} \\
  H_{\alpha_{i}}^{W'} & P_{\alpha_{i}}^{W} \\
\end{array}%
\right), \ \
\psi_{it}^{W} = - \left(%
\begin{array}{c}
  H_{\alpha_{i}}^{W} \\
  P_{\alpha_{i}}^{W} \\
\end{array}%
\right)g(z_{it}), \\
E \left[ \tilde \psi_{i}^{W} \tilde  \psi_{i}^{W'} \right] &=& \left(%
\begin{array}{cc}
  H_{\alpha_{i}}^{W}  \\
  P_{\alpha_{i}}^{W}  \\
\end{array}%
\right) \Omega_{i} \left(H_{\alpha_{i}}^{W'}, P_{\alpha_{i}}^{W}\right),
 \\ E \left[ \tilde{A}_{i}^{W} \tilde{\psi}_{i}^{W} \right]
&=&  \sum_{j=-\infty}^{\infty}  \left(%
\begin{array}{c}
   E \left[ G_{\alpha_{i}}(z_{it})' P_{\alpha_{i}}^{W} g(z_{i,t-j}) \right] \\
    E \left[ G_{\alpha_{i}}(z_{it})' H_{\alpha_{i}}^{W} g(z_{i,t-j}) \right]  + E \left[ \xi_i(z_{it}) P_{\alpha_{i}}^{W} g(z_{i,t-j}) \right]\\
\end{array}%
\right), \\
E \left[ \tilde{\psi}_{i,j}^{W} T_{i,j}^{W}
\tilde{\psi}_{i}^{W} \right] &=& \left\{%
\begin{array}{ll}
    -\left(%
\begin{array}{c}
  G_{\alpha \alpha_{i,j}}' P_{\alpha_{i}}^{W} \Omega_{i} H_{\alpha_{i}}^{W'} \\
  G_{\alpha \alpha_{i,j}}' H_{\alpha_{i}}^{W} \Omega_{i} H_{\alpha_{i}}^{W'} \\
\end{array}%
\right), & \hbox{if $j\le d_{\alpha}$;} \\
    G_{\alpha \alpha_{i}}' (I_{d_{\alpha}} \otimes e_{j-d_{\alpha}}) H_{\alpha_{i}}^{W} \Omega_{i} P_{\alpha_{i},j}^{W}, & \hbox{if $j>d_{\alpha}$.} \\
\end{array}%
\right. .
\end{eqnarray*}
\end{proof}


\begin{lemma}
\label{lh2} Suppose that Conditions \ref{cond1}, \ref{cond2},
\ref{cond3}, \ref{cond5}, \ref{cond4}, and \ref{cond7} hold. Then,
\begin{equation*}
\sqrt{T} (\hat{\gamma}_{i0} - \gamma_{i0}) = \tilde{\psi}_{i} +
T^{-1/2} R_{1i} \ind N(0, V_i),
\end{equation*}
where
\begin{eqnarray*}
\tilde{\psi}_{i} &=& \frac{1}{\sqrt{T}} \sum_{t=1}^{T} \psi_{it} = -
\left( T_{i}^{\Omega} \right) ^{-1} \sqrt{T}
\hat{t}_{i}^{\Omega}
= o_{uP}\left( T^{1/10} \right), \ \
R_{1i} = o_{uP} \left( T^{1/5} \right), \ \ V_i = E[\tilde{\psi}_{i}\tilde{\psi}_{i}'].
\end{eqnarray*}
Also
\begin{equation*}
\frac{1}{\sqrt{n}} \sum_{i=1}^{n} \tilde  \psi_{i} = O_{P}(1)
\end{equation*}
\end{lemma}

\begin{proof}
The statements about $\tilde{\psi}_{i}$ follow by the proof of Lemma
\ref{lb6} applied to the second stage, and the CLT in Lemma 3 of HK. From a similar argument to the proof of Lemma \ref{lb6},
\begin{eqnarray*}
R_{1i} &=& - \underset{= O_{u}(1)} { \underbrace{\left(
T_{i}^{\Omega} \right) ^{-1}}} \ \underset{= o_{uP}(T^{1/10})} {
\underbrace{ \sqrt{T} (\hat{T}_{i}^{\Omega}(\theta_0,\overline{\gamma}_{i}) -
T_{i}^{\Omega}) }} \ \underset{= o_{uP}(T^{1/10})} { \underbrace{
\sqrt{T} (\hat{\gamma}_{i0} - \gamma_{i0})}}  - \underset{=
O_{u}(1)} { \underbrace{\left( T_{i}^{\Omega} \right) ^{-1}}} \
\underset{= o_{uP}(T^{1/10})} { \underbrace{ \sqrt{T}
(\hat{T}_{i}^{R}(\theta_0,\overline{\gamma}_{i}) - T_{i}^{R}) }} \
\underset{= o_{uP}(T^{1/10})} { \underbrace{ \sqrt{T}
(\hat{\gamma}_{i0} - \gamma_{i0})}} \\
&=& o_{uP}(T^{1/5}),
\end{eqnarray*}
by Conditions \ref{cond3} and \ref{cond5}.
\end{proof}


\begin{lemma}
\label{ld1} Assume that Conditions \ref{cond1}, \ref{cond2},
\ref{cond3}, \ref{cond5} and \ref{cond4} hold. Then,
\begin{equation*}
\hat{\Omega}_{i}(\tilde{\theta},\tilde{\alpha}_{i}) = \Omega_{i} +
T^{-1/2}  \tilde{\psi}_{\Omega i}^{W} + T^{-1} R_{1 \Omega i}^{W},
\end{equation*}
where
$$
\tilde{\psi}_{\Omega i}^{W} = \sqrt{T} \left( \hat{\Omega}_{i} -
\Omega_{i} \right) + \sum_{j=1}^{d_{\alpha}} \Omega_{\alpha_{i,j}} \tilde{\psi}_{i,j}^{W}
= o_{up}(T^{1/10}), \ \
R_{1 \Omega i}^{W} = o_{up}(T^{1/5}),
$$
and $\tilde{\psi}_{i,j}^{W}$ is the $j$th element of $\tilde{\psi}_{i,j}^{W}.$
\end{lemma}

\begin{proof}
By a mean value expansion around $(\theta_{0}, \alpha_{i0})$,
\begin{eqnarray*}
\hat{\Omega}_{i}(\tilde{\theta},\tilde{\alpha}_{i}) &=&
\hat{\Omega}_{i} + \sum_{j=1}^{d_{\alpha}} \hat{\Omega}_{\alpha_{i,j}}(\overline{\theta},\overline{\alpha}_{i}) (\tilde{\alpha}_{i,j} -
\alpha_{i0,j}) + \sum_{j=1}^{d_{\theta}} \hat{\Omega}_{\theta_{j}}(\overline{\theta},\overline{\alpha}_{i}) 
(\tilde{\theta}_{j} - \theta_{0,j}) ,
\end{eqnarray*}
where $(\overline{\theta},\overline{\alpha}_{i})$ lies between
$(\tilde{\theta},\tilde{\alpha}_{i})$ and
$(\theta_{0},\alpha_{i0})$. The expressions for
$\tilde{\psi}_{\Omega i}^{W}$ can be
obtained using the expansions for $\tilde{\gamma}_{i0}$ in Lemma
\ref{lf2} since $\tilde{\gamma}_{i} - \tilde{\gamma}_{i0} = o_{uP}(T^{-3/10})$. The order of this term follows from Lemma
\ref{lf2} and the CLT for independent sequences. The remainder
term is
\begin{eqnarray*} R_{1 \Omega i}^{W} &=&
\sum_{j=1}^{d_{\alpha}} \left[ \Omega_{\alpha_{i,j}}  R_{1i,j}^{W} + \sqrt{T}
(\hat{\Omega}_{\alpha_{i,j}}(\overline{\theta},\overline{\alpha}_{i}) - \Omega_{\alpha_{i,j}})  \sqrt{T} (\tilde{\alpha}_{i,j} -
\alpha_{i0,j})\right]
+  \sum_{j=1}^{d_{\theta}} \hat \Omega_{\theta_{j}}(\overline{\theta},\overline{\alpha}_{i}) T(\tilde \theta_j - \theta_{0,j}).
\end{eqnarray*}
The uniform rate of convergence then follows by Lemmas  \ref{le3} and \ref{lf2}, and
 Theorem \ref{th1}.

\end{proof}


\begin{lemma}
\label{lh3} Suppose that Conditions \ref{cond1}, \ref{cond2},
\ref{cond3}, \ref{cond5}, and \ref{cond4} hold. Then,
\begin{equation}
\sqrt{T} (\hat{\gamma}_{i0} - \gamma_{i0}) = \tilde{\psi}_{i} +
T^{-1/2} Q_{1i} + T^{-1} R_{2i},
\end{equation}
where
\begin{eqnarray*}
Q_{1i}(\tilde{\psi}_{i}, \tilde{a}_{i}) &=& - \left( T_{i}^{\Omega}
\right) ^{-1} \left[ \tilde{A}_{i}^{\Omega} \tilde{\psi}_{i} +
\frac{1}{2} \sum_{j=1}^{d_g+d_{\alpha}} \tilde{\psi}_{i,j} T_{i,j}^{\Omega}
\tilde{\psi}_{i} + \text{diag}[0,\tilde{\psi}_{\Omega_{i}}^{W}]
\tilde{\psi}_{i}\right] 
= o_{uP} \left( T^{1/5} \right), \\
\tilde{A}_{i}^{\Omega} &=& \sqrt{T} (\hat{T}_{i}^{\Omega}
- T_{i}^{\Omega}) = o_{uP} \left( T^{1/10} \right), \ \
R_{2i} = o_{uP} \left( T^{3/10} \right).
\end{eqnarray*}
Also,
\begin{equation*}
\frac{1}{n} \sum_{i=1}^{n}   Q_{1i}  = O_{P}(1).
\end{equation*}
\end{lemma}
\begin{proof}
By a second order Taylor expansion of the FOC for
$\hat{\gamma}_{i0}$, we have
\begin{equation*}
0 = \hat{t}_{i}(\theta_0, \hat{\gamma}_{i0}) =
\hat{t}_{i}^{\Omega} + \hat{T}_{i} (\hat{\gamma}_{i0} -
\gamma_{i0}) + \frac{1}{2} \sum_{j=1}^{d_g+d_{\alpha}} (\hat{\gamma}_{i0,j} -
\gamma_{i0,j}) \hat{T}_{i,j}(\theta_0,\overline{\gamma}_{i})
(\hat{\gamma}_{i0} - \gamma_{i0}),
\end{equation*}
where $\overline{\gamma}_{i}$ is between $\hat{\gamma}_{i0}$ and
$\gamma_{i0}$. The expression for $Q_{1i}$ can be obtained in a
similar fashion as in Lemma A4 in Newey and Smith (2004). The rest of
the properties for $Q_{1i}$ follow by Lemma \ref{lb6} applied to the
second stage, Lemma \ref{lh2},  and an argument similar to the proof of Theorem 1 in HK that uses Corollary A.2 of Hall and Heide (1980, p. 278) 
and Lemma 1 of Andrews (1991). The
remainder term is
\begin{eqnarray*}
R_{2i} &=& - \left( T_{i}^{\Omega} \right) ^{-1} \left[
\tilde{A}_{i}^{\Omega} R_{1i} +  \sum_{j=1}^{d_g+d_{\alpha}} \left[ R_{1i,j}
T_{i,j}^{\Omega} \sqrt{T} (\hat{\gamma}_{i0} - \gamma_{i0}) +
 \tilde{\psi}_{i,j} T_{i,j}^{\Omega} R_{1i} \right]/2 \right]
 \\
 &-&  \left( T_{i}^{\Omega} \right) ^{-1}  \sum_{j=1}^{d_g+d_{\alpha}} \sqrt{T} (\hat{\gamma}_{i0,j} -
\gamma_{i0,j}) \sqrt{T}
(\hat{T}_{i,j}^{\Omega}(\theta_0, \overline{\gamma}_{i}) - T_{i,j}^{\Omega})
\sqrt{T} (\hat{\gamma}_{i0} - \gamma_{i0}) /2  \\
 &-&  \left( T_{i}^{\Omega} \right) ^{-1} \left[ diag[0,R_{1
 \Omega_{i}}^{W}]\sqrt{T} (\hat{\gamma}_{i0} - \gamma_{i0}) +
 \text{diag}[0,\tilde{\psi}_{\Omega_{i}}^{W}] R_{1i}
\right].
\end{eqnarray*}
The uniform rate of convergence then follows by Lemmas \ref{lb6} and
\ref{lh2}, and Conditions \ref{cond3} and \ref{cond5}.
\end{proof}

\begin{lemma}
\label{lh5} Suppose that Conditions \ref{cond1}, \ref{cond2},
\ref{cond3}, \ref{cond5}, \ref{cond4}, and \ref{cond7} hold. Then,
$$
\frac{1}{\sqrt{n}} \sum_{i=1}^{n} \tilde \psi_{i} \ind N(0, \bar{E}[V_i]), \ \ 
\frac{1}{n} \sum_{i=1}^{n}  Q_{1i}  \inp  \bar{E}[B_{\gamma_{i}}^{I} + B_{\gamma_{i}}^{G} +
B_{\gamma_{i}}^{\Omega} + B_{\gamma_{i}}^{W}] =: B_{\gamma},
$$
where
\begin{eqnarray*}
V_i &=&  \text{diag} \left(
  \Sigma_{\alpha_{i}},
  P_{\alpha_{i}} 
\right),
\\
B_{\gamma_{i}}^{I} &=& \left(%
\begin{array}{c}
  B_{\alpha_{i}}^{I} \\
  B_{\lambda_{i}}^{I} \\
\end{array}%
\right) = \left(%
\begin{array}{c}
  H_{\alpha_{i}}  \\
P_{\alpha_{i}} \\
\end{array}%
\right) \left(- \sum_{j=1}^{d_{\alpha}} G_{\alpha \alpha_{i,j}}
  \Sigma_{\alpha_{i}}/2
  + E \left[ G_{\alpha_{i}}(z_{it}) H_{\alpha_{i}} g(z_{i,t-j}) \right]\right), \\
B_{\gamma_{i}}^{G} &=& \left(%
\begin{array}{c}
  B_{\alpha_{i}}^{G} \\
  B_{\lambda_{i}}^{G} \\
\end{array}%
\right) =  \left(%
\begin{array}{c}
  -\Sigma_{\alpha_{i}} \\
  H_{\alpha_{i}}'  \\
\end{array}%
\right) \sum_{j=0}^{\infty}  E \left[ G_{\alpha_{i}}(z_{it})'P_{\alpha_{i}}g(z_{i,t-j}) \right] , \\
B_{\gamma_{i}}^{\Omega} &=& \left(%
\begin{array}{c}
  B_{\alpha_{i}}^{\Omega} \\
  B_{\lambda_{i}}^{\Omega} \\
\end{array}%
\right) = \left(%
\begin{array}{c}
  H_{\alpha_{i}} \\
  P_{\alpha_{i}} \\
\end{array}%
\right)    \sum_{j=0}^{\infty} E[g(z_{it}) g(z_{it})'P_{\alpha_{i}} g(z_{i,t-j})], \\
B_{\gamma_{i}}^{W} &=& \left(%
\begin{array}{c}
  B_{\alpha_{i}}^{W} \\
  B_{\lambda_{i}}^{W} \\
\end{array}%
\right) = \left(%
\begin{array}{c}
   H_{\alpha_{i}}  \\
   P_{\alpha_{i}}  \\
\end{array}%
\right) \sum_{j=1}^{d_{\alpha}} \Omega_{\alpha_{i,j}} \left( H_{\alpha_{i, j}}^{W'} - H'_{\alpha_{i,j}} \right),
\end{eqnarray*}
for $\Sigma_{\alpha_{i}} = \left( G_{\alpha_{i}}' \Omega_{i}^{-1}
G_{\alpha_{i}} \right)^{-1},$ $H_{\alpha_{i}} = \Sigma_{\alpha_{i}}
G_{\alpha_{i}}' \Omega_{i}^{-1},$ and $P_{\alpha_{i}} =
\Omega_{i}^{-1} - \Omega_{i}^{-1} G_{\alpha_{i}} H_{\alpha_{i}}$.
\end{lemma}

\begin{proof}
The results follow by Lemmas \ref{lh2} and \ref{lh3}, noting that
\begin{eqnarray*}
\left( T_{i}^{\Omega} \right)^{-1} &=& - \left(%
\begin{array}{cc}
  -\Sigma_{\alpha_{i}} & H_{\alpha_{i}} \\
  H_{\alpha_{i}}' & P_{\alpha_{i}} \\
\end{array}%
\right), \ \
\psi_{it} = - \left(%
\begin{array}{c}
  H_{\alpha_{i}} \\
  P_{\alpha_{i}} \\
\end{array}%
\right)g(z_{it}), \\
E \left[ \tilde \psi_{i} \tilde \psi_{i}' \right] &=& \left(%
\begin{array}{cc}
  \Sigma_{\alpha_{i}} & 0 \\
  0 & P_{\alpha_{i}} \\
\end{array}%
\right),\ \
 E \left[ \tilde{A}_{i}^{\Omega}
\tilde{\psi}_{i} \right]
=  \sum_{j=0}^{\infty} \left(%
\begin{array}{c}
   E \left[ G_{\alpha_{i}}(z_{it})'P_{\alpha_{i}}g(z_{i,t-j}) \right] \\
   E \left[ G_{\alpha_{i}}(z_{it})'H_{\alpha_{i}}g(z_{i,t-j}) \right]\\
\end{array}%
\right), \\
E \left[ \tilde{\psi}_{i,j} T_{i,j}^{\Omega}
\tilde{\psi}_{i} \right] &=& \left\{%
\begin{array}{ll}
    -\left(%
\begin{array}{c}
  0 \\
  G_{\alpha \alpha_{i,j}}' \Sigma_{\alpha_{i}} \\
\end{array}%
\right), & \hbox{if $j \leq d_{\alpha}$;} \\
    0, & \hbox{if $j>d_{\alpha}$.} \\
\end{array}%
\right. \\
E \left[ \text{diag}[0,\tilde{\psi}_{\Omega_{i}}^{W}] \tilde{\psi}_{i}
\right] &=&
\left(%
\begin{array}{cc}
  0 \\
  \sum_{j=0}^{\infty}  E[g(z_{it}) g(z_{it})'P_{\alpha_{i}} g(z_{i,t-j})]
  + \sum_{j=1}^{d_{\alpha}} \Omega_{\alpha_{i,j}} \left( H_{\alpha_{i,j}}^{W'} - H_{\alpha_{i,j}} \right)
  \\
\end{array}%
\right).
\end{eqnarray*}
\end{proof}

\section{Stochastic Expansion for $\hat{s}_{i}^{W}(\theta_{0}, \tilde \gamma_{i0})$ and $\hat{s}_{i}(\theta_{0}, \hat \gamma_{i0})$} \label{appendixG}
We characterize stochastic expansions up to second order  for one-step and two-step
profile scores of the common parameter evaluated at the true value of the
common parameter. The expressions for the scores and their derivatives in the components of the expansions are given in Appendix \ref{appendixF}.


\begin{lemma}
\label{lg1} Suppose that Conditions \ref{cond1}, \ref{cond2},
\ref{cond3}, and \ref{cond5} hold. Then,
\begin{equation*}
\hat{s}_{i}^{W} (\theta_{0}, \tilde \gamma_{i0}) = T^{-1/2}
\tilde{\psi}_{si}^{W} + T^{-1} Q_{1si}^{W} + T^{-3/2} R_{2si}^{W},
\end{equation*}
where
\begin{eqnarray*}
\tilde{\psi}_{si}^{W} &=&  M_{i}^{W} \tilde{\psi}_{i}^{W}
= o_{uP}(T^{1/10}), \ \
Q_{1si}^{W} = M_{i}^{W} Q_{1i}^{W} + \tilde{C}_{i}^{W}
\tilde{\psi}_{i}^{W} + \frac{1}{2} \sum_{j=1}^{d_g+d_{\alpha}}
\tilde{\psi}_{i,j}^{W} M_{i,j}^{W} \tilde{\psi}_{i}^{W}
 = o_{uP}(T^{1/5}), \\
\tilde{C}_{i}^{W} &=& \sqrt{T} (\hat{M}_{i}^{W} -
M_{i}^{W}) = o_{uP} (T^{1/10}), \ \
R_{2si}^{W} = o_{uP}(T^{2/5}). 
\end{eqnarray*}
Also,
$$
\frac{1}{\sqrt{n}} \sum_{i=1}^{n} \tilde{\psi}_{si}^{W}  =
O_{P}(1), \ \ \frac{1}{n} \sum_{i=1}^{n}  
Q_{1si}^{W}  = O_{P}(1).
$$
\end{lemma}

\begin{proof}
By a second order Taylor expansion of $\hat{s}_{i}^{W} (\theta_{0}, \tilde{\gamma}_{i0})$ around $\tilde{\gamma}_{i0} = \gamma_{i0}$,
\begin{eqnarray*}
\hat{s}_{i}^{W} (\theta_{0}, \tilde{\gamma}_{i0}) &=&
\hat{s}_{i}^{W} +
\hat{M}_{i}^W (\tilde{\gamma}_{i0} - \gamma_{i0})
+ \frac{1}{2} \sum_{j=1}^{d_g+d_{\alpha}} (\tilde{\gamma}_{i0,j} -
\gamma_{i0,j}) \hat{M}_{i,j}^W(\theta_0,\overline{\gamma}_{i})
(\tilde{\gamma}_{i0} - \gamma_{i0}),
\end{eqnarray*}
where $\overline{\gamma}_{i}$ is between $\tilde{\gamma}_{i0}$ and
$\gamma_{i0}$. Noting that $\hat{s}_{i}^{W} (\theta_{0},
\gamma_{i0}) = 0$ and using the expansion for $\tilde{\gamma}_{i0}$
in Lemma \ref{lf3}, we can obtain the expressions for
$\tilde{\psi}_{si}^{W}$ and $Q_{1si}^{W}$, after
some algebra. The rest of the properties for these terms follow by
the properties of $\tilde{\psi}_{i}^{W}$ and $Q_{1i}^{W}$.
The remainder term is
\begin{eqnarray*}
R_{2si}^{W} &=& M_{i}^{W}R_{2i}^{W} + \tilde{C}_{i}^{W} R_{1i}^{W} +
\frac{1}{2} \sum_{j=1}^{d_g+d_{\alpha}} \left[ R_{1i,j}^{W} M_{i,j}^{W}
\sqrt{T} (\tilde{\gamma}_{i0} - \gamma_{i0})  +
\tilde{\psi}_{i,j}^{W} M_{i,j}^{W} R_{1i}^{W} \right]
 \\
&+&  \frac{1}{2} \sum_{j=1}^{d_g+d_{\alpha}}   \sqrt{T}
(\tilde{\gamma}_{i0,j} - \gamma_{i0,j})  \sqrt{T}
(\hat{M}_{i,j}^{W}(\theta_0, \overline{\gamma}_{i}) - M_{i,j}^{W}) \sqrt{T}
(\tilde{\gamma}_{i0} - \gamma_{i0}) .
\end{eqnarray*}
The uniform order of $R_{2si}^{W}$ follows by the properties of the
components in the expansion of $\tilde{\gamma}_{i0}$, Lemma
\ref{lb6}, and Conditions \ref{cond3} and \ref{cond5}.
\end{proof}

\begin{lemma}
\label{lg2} Suppose that Conditions \ref{cond1}, \ref{cond2},
\ref{cond3}, and \ref{cond5} hold. We then have
\begin{eqnarray*}
\frac{1}{\sqrt{n}} \sum_{i=1}^{n} \tilde \psi_{si}^{W} &\ind& N(0,
V_{s}^{W}),
\ \
V_{s}^{W} = \bar{E}[
G_{\theta_{i}}' P_{\alpha_{i}}^{W} \Omega_{i} P_{\alpha_{i}}^{W} G_{\theta_{i}}], \\
\frac{1}{n} \sum_{i=1}^{n}  Q_{1si}^{W}  &\inp&  \bar{E} E \left[ Q_{1si}^{W}
\right] = \bar{E}[B_{si}^{W,B} + B_{si}^{W,C} + B_{si}^{W,V}] =: B_{s}^{W}, 
\end{eqnarray*}
where $B_{si}^{W,B} = - G_{\theta_{i}}' B_{\lambda_{i}}^{W} = -
G_{\theta_{i}}' \left( B_{\lambda_{i}}^{W,I} +
B_{\lambda_{i}}^{W,G} + B_{\lambda_{i}}^{W,1S} \right),$ $B_{si}^{W,C} = \sum_{j=-\infty}^{\infty}E \left[ G_{\theta_{i}}(z_{it})'
P_{\alpha_{i}}^{W} g(z_{i,t-j}) \right],$ $B_{si}^{W,V} =  - \sum_{j=1}^{d_{\alpha}}  G_{\theta \alpha_{i,j}}'
P_{\alpha_{i}}^{W} \Omega_{i} H_{\alpha_{i}}^{W'}/2 - \sum_{j=1}^{d_g}
G_{ \theta \alpha_{i}}' (I_{d_{\alpha}} \otimes e_{j}) H_{\alpha_{i}}^{W}
\Omega_{i} P_{\alpha_{i},j}^{W}/2,
$  $H_{\alpha_{i}}^{W} =
\Sigma_{\alpha_{i}}^{W} G_{\alpha_{i}}' W_{i}^{-1},$ $\Sigma_{\alpha_{i}}^{W} = \left( G_{\alpha_{i}}' W_{i}^{-1}
G_{\alpha_{i}} \right)^{-1},$ and
$P_{\alpha_{i}}^{W} = W_{i}^{-1} - W_{i}^{-1} G_{\alpha_{i}}
H_{\alpha_{i}}^{W}$.
\end{lemma}

\begin{proof}
The results follow by Lemmas \ref{lg1} and \ref{lf5}, noting that
\begin{eqnarray*}
E \left[ \tilde \psi_{si}^{W} \tilde  \psi_{si}^{W'} \right] &=& M_{i}^{W} \left(%
\begin{array}{cc}
  H_{\alpha_{i}}^{W} \Omega_{i} H_{\alpha_{i}}^{W'} & H_{\alpha_{i}}^{W} \Omega_{i} P_{\alpha_{i}}^{W} \\
  P_{\alpha_{i}}^{W} \Omega_{i} H_{\alpha_{i}}^{W'} & P_{\alpha_{i}}^{W} \Omega_{i} P_{\alpha_{i}}^{W} \\
\end{array}%
\right) M_{i}^{W'}, \\
E \left[ \tilde{C}_{i}^{W} \tilde{\psi}_{i}^{W} \right] &=& \sum_{j=-\infty}^{\infty}E \left[ G_{\theta_{i}}(z_{it})'
P_{\alpha_{i}}^{W} g(z_{i,t-j}) \right], \\
E \left[ \tilde{\psi}_{i,j}^{W} M_{i,j}^{W}
\tilde{\psi}_{i}^{W} \right] &=& \left\{%
\begin{array}{ll}
    - G_{\theta \alpha _{i,j}}' P_{\alpha_{i}}^{W} \Omega_{i}
H_{\alpha_{i}}^{W'}, & \hbox{if $j\le d_{\alpha}$;} \\
    - G_{\theta \alpha _{i}}' (I_{d_{\alpha}} \otimes e_{j-d_{\alpha}}) H_{\alpha_{i}}^{W} \Omega_{i}
P_{\alpha_{i},j}^{W}, & \hbox{if $j>d_{\alpha}$.} \\
\end{array}%
\right.
\end{eqnarray*}
\end{proof}


\begin{lemma}
\label{lg3} Suppose that Conditions \ref{cond1}, \ref{cond2},
\ref{cond3}, and \ref{cond5} hold. Then, for $\hat{s}^{W}(\theta_{0}) = n^{-1} \sum_{i=1}^{n}
\hat{s}_{i}^{W}(\theta_{0}, \tilde{\gamma}_{i0}),$
$$
\sqrt{nT} \hat{s}^{W}(\theta_{0}) \ind N\left( \kappa B_{s}^{W}, V_{s}^{W} \right),
$$
where $B_{s}^{W}$ and $V_{s}^{W}$ are defined in Lemma \ref{lg2}.
\end{lemma}

\begin{proof}
By Lemma \ref{lg1}, 
\begin{eqnarray*}
\sqrt{nT} \hat{s}^{W} (\theta_{0}) &=&
\underset{=O_{P}(1)}{\underbrace{\frac{1}{\sqrt{n}} \sum_{i=1}^{n}
\tilde{\psi}_{si}^{W}}} +
\underset{=O_{P}(1)}{\underbrace{\sqrt{\frac{n}{T}} \frac{1}{n}
\sum_{i=1}^{n} Q_{1si}^{W}}} +
\underset{=o_{P}(1)}{\underbrace{\sqrt{\frac{n}{T^{2}}} \frac{1}{n}
\sum_{i=1}^{n} R_{2si}^{W}}} \\
&=& \frac{1}{\sqrt{n}} \sum_{i=1}^{n} \tilde{\psi}_{si}^{W} +
\sqrt{\frac{n}{T}} \frac{1}{n} \sum_{i=1}^{n} Q_{1si}^{W} +
o_{P}(1).
\end{eqnarray*}
Then, the result follows by Lemma \ref{lg2}.
\end{proof}


\begin{lemma}
\label{li1} Suppose that Conditions \ref{cond1}, \ref{cond2},
\ref{cond3}, \ref{cond5}, \ref{cond4}, and \ref{cond7} hold. Then,
\begin{equation*}
\hat{s}_{i} (\theta_{0}, \hat{\gamma}_{i0}) = T^{-1/2}
\tilde{\psi}_{si} + T^{-1} Q_{1si} + T^{-3/2} R_{2si},
\end{equation*}
where all the terms are identical to that of Lemma \ref{lg1}
after replacing $W$ by $\Omega$. Also, 
the  properties of all the terms of the expansion are the analogous to those of Lemma \ref{lg1}.
\end{lemma}
\begin{proof}
The proof is similar to the proof of Lemma \ref{lg1}.
\end{proof}

%
%


\begin{lemma}
\label{li2} Suppose that Conditions \ref{cond1}, \ref{cond2},
\ref{cond3}, \ref{cond5}, \ref{cond4}, and \ref{cond7} hold. Then,
\begin{eqnarray*}
\frac{1}{\sqrt{n}} \sum_{i=1}^{n} \tilde \psi_{si} &\ind& N(0, J_{s}),
\ \
J_{s} = \bar{E}[
G_{\theta_{i}}' P_{\alpha_{i}} G_{\theta_{i}}] \\
\frac{1}{n} \sum_{i=1}^{n}  Q_{1si}  &\inp&  \bar{E} E \left[ Q_{1si}
\right]  = \bar{E}[B_{si}^{B} + B_{si}^{C}] =: B_{s}, 
\end{eqnarray*}
where $B_{si}^{B} = - G_{\theta_{i}}'\left( B_{\lambda_{i}}^{I} +
B_{\lambda_{i}}^{G} + B_{\lambda_{i}}^{\Omega} + B_{\lambda_{i}}^{W} \right),$ $B_{si}^{C} = \sum_{j=0}^{\infty} E \left[ G_{\theta_{i}}(z_{it})'
P_{\alpha_{i}} g(z_{i,t-j}) \right],$  $P_{\alpha_{i}} = \Omega_{i}^{-1} -
\Omega_{i}^{-1} G_{\alpha_{i}} H_{\alpha_{i}},$ $H_{\alpha_{i}} = \Sigma_{\alpha_{i}} G_{\alpha_{i}}'
\Omega_{i}^{-1}$, and $\Sigma_{\alpha_{i}} = \left( G_{\alpha_{i}}' \Omega_{i}^{-1}
G_{\alpha_{i}} \right)^{-1}$.
\end{lemma}

\begin{proof}
The results follow by Lemmas \ref{lh2}, \ref{lh3}, \ref{lh5} and
\ref{li1}, noting that
$$
E \left[ \tilde \psi_{si} \tilde \psi_{si}' \right] = M_{i}^{\Omega} \left(%
\begin{array}{cc}
  \Sigma_{\alpha_{i}} & 0 \\
  0 & P_{\alpha_{i}} \\
\end{array}%
\right) M_{i}^{\Omega'}, \ \
E \left[ \tilde{C}_{i}^{\Omega} \tilde{\psi}_{i} \right] = \sum_{j=0}^{\infty} E \left[ G_{\theta_{i}}(z_{it})'
P_{\alpha_{i}} g(z_{i,t-j}) \right], \ \
E \left[ \tilde{\psi}_{i,j} M_{i,j}^{\Omega}
\tilde{\psi}_{i} \right] = 0.
$$
\end{proof}

\begin{lemma}
\label{li5} Suppose that Conditions \ref{cond1}, \ref{cond2},
\ref{cond3}, \ref{cond4}, and \ref{cond5} hold. Then, or $\hat{s}(\theta_{0}) = n^{-1} \sum_{i=1}^{n}
\hat{s}_{i}(\theta_{0}, \hat{\gamma}_{i0}),$
$$
\sqrt{nT} \hat{s} (\theta_{0}) = \frac{1}{\sqrt{n}}
\sum_{i=1}^{n} \tilde{\psi}_{si} + \sqrt{\frac{n}{T}} B_{s}
 + o_{P}(1) \ind N\left( \kappa B_{s}, J_{s} \right),
$$
where $ \tilde{\psi}_{si}$ and $B_{s}$ are defined in Lemmas \ref{li1} and \ref{li2}, respectively.
\end{lemma}

\begin{proof}
Using the expansion form obtained in Lemma \ref{li1}, we can get the
result by examining each term with Lemma \ref{li2}. 
\end{proof}

\section{Scores and Derivatives}\label{appendixF}

\subsection{One-Step Score and Derivatives: Individual Effects}
\label{appendixK} We denote dimensions of $g(z_{it})$, $\alpha_i$,
and $\theta$ by $d_g$, $d_{\alpha}$ and $d_{\theta}$. The symbol $\otimes$ denotes
kronecker product of matrices, and $I_{d_{\alpha}}$ denotes a $d_{\alpha}$-order
identity matrix. Let $G_{\alpha \alpha_i}(z_{it}; \theta, \alpha_i)
:= (G_{\alpha \alpha_{i, 1}}(z_{it}; \theta, \alpha_i)', ...,
G_{\alpha \alpha_{i,d_{\alpha}}}(z_{it}; \theta, \alpha_i)')',$ where
$$
G_{\alpha \alpha_{i, j}}(z_{it}; \theta, \alpha_i) = \frac{\partial
G_{\alpha_i}(z_{it}; \theta, \alpha_i) }{\partial \alpha_{i,j}}.
$$
We denote derivatives of  $G_{\alpha \alpha_i}(z_{it}; \theta, \alpha_i)$ with respect to $\alpha_{i,j}$  by  $G_{\alpha \alpha, \alpha_{i,j}}(z_{it}; \theta, \alpha_i)$, and use additional subscripts for higher order derivatives.


\subsubsection{Score}

\begin{equation*}
\hat{t}_{i}^{W}(\theta, \gamma_{i}) = - \frac{1}{T}
\sum_{t=1}^{T} \left(%
\begin{array}{c}
  G_{\alpha_{i}}(z_{it};\theta, \alpha_{i})'\lambda_{i} \\
  g(z_{it};\theta,\alpha_{i}) + \hat W_{i} \lambda_{i} \\
\end{array}%
\right) = - \left(%
\begin{array}{c}
  \hat{G}_{\alpha_{i}}(\theta, \alpha_{i})'\lambda_{i} \\
  \hat{g}_{i}(\theta,\alpha_{i}) + \hat W_{i} \lambda_{i} \\
\end{array}%
\right).
\end{equation*}

\subsubsection{Derivatives with respect to the fixed effects}

\noindent \textbf{\\First Derivatives}

\begin{eqnarray*}
\hat{T}_{i}^{W}(\theta, \gamma_{i}) &=& \frac{\partial \hat{t}_{i}^{W}(\gamma_{i}, \theta)}{\partial \gamma_{i}'}= - \left(%
\begin{array}{cc}
  \hat{G}_{\alpha \alpha_{i}}(\theta,\alpha_i)'(I_{d_{\alpha}} \otimes \lambda_i) & \hat{G}_{\alpha_{i}}(\theta,\alpha_{i})' \\
  \hat{G}_{\alpha_{i}}(\theta,\alpha_{i}) & \hat W_{i} \\
\end{array}%
\right). \\
T_{i}^{W} &=& E \left[ \hat{T}_{i}^{W} \right] = - \left(%
\begin{array}{cc}
  0 & G_{\alpha_{i}}' \\
  G_{\alpha_{i}} & W_{i} \\
\end{array}%
\right). \\
\left( T_{i}^{W} \right)^{-1} &=& - \left(%
\begin{array}{cc}
  -\Sigma_{\alpha_{i}}^{W} & H_{\alpha_{i}}^{W} \\
  H_{\alpha_{i}}^{W'} & P_{\alpha_{i}}^{W} \\
\end{array}%
\right).
\end{eqnarray*}

\noindent \textbf{Second Derivatives}

\begin{eqnarray*}
\hat{T}_{i,j}^{W}(\theta, \gamma_{i}) &=& \frac{\partial^{2}
\hat{t}_{i}^{W}(\theta, \gamma_{i})}{\partial \gamma_{i,j}
\partial \gamma_{i}'} =
\left\{%
\begin{array}{ll}
    - \left(%
\begin{array}{cc}
  \hat{G}_{\alpha \alpha, \alpha_{i,j}}(\theta,\alpha_i)'(I_{d_{\alpha}} \otimes \lambda_i) & \hat{G}_{ \alpha \alpha_{i,j}}(\theta,\alpha_{i})' \\
  \hat{G}_{\alpha \alpha_{i,j} }(\theta,\alpha_{i}) & 0 \\
\end{array}%
\right), & \hbox{if $j\le d_{\alpha}$;} \\
    - \left(%
\begin{array}{cc}
  \hat{G}_{\alpha \alpha_{i}}(\theta,\alpha_i)'(I_{d_{\alpha}} \otimes e_{j-d_{\alpha}}) & 0 \\
  0 & 0 \\
\end{array}%
\right), & \hbox{if $j>d_{\alpha}$.} \\
\end{array}%
\right.      \\
T_{i,j}^{W} &=& E \left[ \hat{T}_{i,j}^{W}(\gamma_{i0}; \theta_{0}) \right] = \left\{%
\begin{array}{ll}
    - \left(%
\begin{array}{cc}
  0 & G_{\alpha \alpha_{i,j} }' \\
  G_{\alpha \alpha_{i,j} } & 0 \\
\end{array}%
\right), & \hbox{if $j\le d_{\alpha}$;} \\
    - \left(%
\begin{array}{cc}
 G_{\alpha \alpha_{i}}'(I_{d_{\alpha}} \otimes e_{j-d_{\alpha}}) & 0 \\
  0 & 0 \\
\end{array}%
\right), & \hbox{if $j>d_{\alpha}$.} \\
\end{array}%
\right.
\end{eqnarray*}

\noindent \textbf{Third Derivatives}

\begin{eqnarray*}
\hat{T}_{i,jk}^{W}(\theta, \gamma_{i}) &=& \frac{\partial^{3}
\hat{t}_{i}^{W}(\theta, \gamma_{i})}{\partial \gamma_{i,k}
\partial \gamma_{i,j}
\partial \gamma_{i}'} = \left\{%
\begin{array}{ll}
    - \left(%
\begin{array}{cc}
   \hat{G}_{\alpha \alpha, \alpha \alpha_{i,jk}}(\theta,\alpha_i)'(I_{d_{\alpha}} \otimes \lambda_i) & \hat{G}_{\alpha \alpha \alpha_{i,jk} }(\theta,\alpha_{i})' \\
  \hat{G}_{\alpha \alpha \alpha_{i,jk} }(\theta,\alpha_{i}) & 0 \\
\end{array}%
\right), & \hbox{if $j\le d_{\alpha}, k \le d_{\alpha}$;} \\
    - \left(%
\begin{array}{cc}
  \hat{G}_{\alpha \alpha, \alpha_{i,j}}(\theta,\alpha_i)'(I_{d_{\alpha}} \otimes e_{k-d_{\alpha}}) & 0 \\
  0 & 0 \\
\end{array}%
\right), & \hbox{if $j\le d_{\alpha}, k> d_{\alpha}$;} \\
        - \left(%
\begin{array}{cc}
\hat{G}_{\alpha \alpha, \alpha_{i,k}}(\theta,\alpha_i)'(I_{d_{\alpha}} \otimes e_{j-d_{\alpha}}) & 0 \\
  0 & 0 \\
\end{array}%
\right), & \hbox{if $j>d_{\alpha}, k\le d_{\alpha}$;} \\
    \left(%
\begin{array}{cc}
  0 & 0 \\
  0 & 0 \\
\end{array}%
\right), & \hbox{if $j>d_{\alpha}, k>d_{\alpha}$.} \\
\end{array}%
\right.
\\
T_{i,jk}^{W} &=& E \left[ \hat{T}_{i,jk}^{W} \right] = \left\{%
\begin{array}{ll}
    - \left(%
\begin{array}{cc}
  0 & G_{\alpha \alpha \alpha_{i,jk} }' \\
  G_{\alpha \alpha \alpha_{i,jk}} & 0 \\
\end{array}%
\right), & \hbox{if $j \le d_{\alpha}, k\le d_{\alpha}$;} \\
    - \left(%
\begin{array}{cc}
  G_{\alpha \alpha, \alpha_{i,j}}'(I_{d_{\alpha}} \otimes e_{k-d_{\alpha}}) & 0 \\
  0 & 0 \\
\end{array}%
\right), & \hbox{if $j \le d_{\alpha}, k>d_{\alpha}$;} \\
        - \left(%
\begin{array}{cc}
 G_{\alpha \alpha, \alpha_{i,k}}'(I_{d_{\alpha}} \otimes e_{j-d_{\alpha}}) & 0 \\
  0 & 0 \\
\end{array}%
\right), & \hbox{if $j>d_{\alpha}, k\le d_{\alpha}$;} \\
    \left(%
\begin{array}{cc}
  0 & 0 \\
  0 & 0 \\
\end{array}%
\right), & \hbox{if $j>d_{\alpha}, k>d_{\alpha}$.} \\
\end{array}%
\right.
\end{eqnarray*}

\subsubsection{Derivatives with respect to the common parameter}

\noindent \textbf{\\First Derivatives}

\begin{eqnarray*}
\hat{N}_{i,j}^{W}(\theta,\gamma_{i}) &=& \frac{\partial \hat{t}_{i}^{W}(\gamma_{i},\theta)}{\partial \theta_{j}}= -  \left(%
\begin{array}{c}
  \hat{G}_{\theta_{j} \alpha_{i}}(\theta,\alpha_{i})'\lambda_{i} \\
  \hat{G}_{\theta_{i,j}}(\theta,\alpha_{i}) \\
\end{array}%
\right). \\
N_{i,j}^{W} &=& E \left[ \hat{N}_{i,j}^{W}
\right]
= -  \left(%
\begin{array}{c}
  0 \\
  G_{\theta_{i,j} } \\
\end{array}%
\right).
\end{eqnarray*}


\subsection{One-Step Score and Derivatives: Common Parameters}
\label{appendixL} Let $G_{\theta \alpha_i}(z_{it}; \theta, \alpha_i)
:= $ \\ $(G_{\theta \alpha_{i, 1}}(z_{it}; \theta, \alpha_i)', \ldots,
G_{\theta \alpha_{i,d_{\alpha}}}(z_{it}; \theta, \alpha_i)')',$ where
$$
G_{\theta \alpha_{i, j}}(z_{it}; \theta, \alpha_i) = \frac{\partial
G_{\theta}(z_{it}; \theta, \alpha_i) }{\partial \alpha_{i,j}}.
$$
We denote the derivatives of  $G_{\theta \alpha_i}(z_{it}; \theta, \alpha_i)$ with respect to $\alpha_{i,j}$  by  $G_{\theta \alpha, \alpha_{i,j}}(z_{it}; \theta, \alpha_i)$, and use additional subscripts for higher order derivatives.


\subsubsection{Score}

\begin{equation*}
\hat{s}_{i}^{W} (\theta, \gamma_i) = - \frac{1}{T}
\sum_{t=1}^{T}
G_{\theta}(z_{it};\theta,\alpha_i)'\lambda_i
= - \hat{G}_{\theta_i}(\theta,\alpha_i)'\lambda_i.
\end{equation*}

\subsubsection{Derivatives with respect to the fixed effects}

\noindent \textbf{\\First  Derivatives}

\begin{eqnarray*}
\hat{M}_{i}^{W}(\theta,\gamma_i) &=& \frac{\partial \hat{s}_{i}^{W}(\theta,\gamma_i)}{\partial \gamma_i'}= - \left(%
\begin{array}{cc}
  \hat{G}_{\theta \alpha_i}(\theta,\alpha_i)' (I_{d_{\alpha}} \otimes \lambda_i) & \hat{G}_{\theta_{i}} (\theta, \alpha_i)' \\
\end{array}%
\right).
    \\
M_{i}^{W} &=& E \left[ \hat{M}_{i}^{W}
\right]
= - \left(%
\begin{array}{cc}
  0 & G_{\theta_{i}}' \\
\end{array}%
\right).
\end{eqnarray*}

\noindent \textbf{Second Derivatives}

\begin{eqnarray*}
\hat{M}_{i,j}^{W}(\theta,\gamma_i) &=&
\frac{\partial^{2}
\hat{s}_{i}^{W}(\theta,\gamma_i)}{\partial
\gamma_{i,j} \partial \gamma_i'} =  \left\{%
\begin{array}{ll}
    - \left(%
\begin{array}{cc}
  \hat{G}_{\theta \alpha, \alpha_{i,j}}(\theta,\alpha_i)' (I_{d_{\alpha}} \otimes \lambda_i) & \hat{G}_{\theta \alpha_{i,j}} (\theta, \alpha_i)' \\
\end{array}%
\right), & \hbox{if $j \le d_{\alpha}$;} \\
    - \left(%
\begin{array}{cc}
  \hat{G}_{\theta \alpha_i}(\theta,\alpha_i)' (I_{d_{\alpha}} \otimes e_{j-d_{\alpha}}) & 0 \\
\end{array}%
\right), & \hbox{if $j>d_{\alpha}$.} \\
\end{array}%
\right.
 \\
M_{i,j}^{W} &=& E \left[ \hat{M}_{i,j}^{W} (\theta_{0},\gamma_{i0})
\right]
=  \left\{%
\begin{array}{ll}
    - \left(%
\begin{array}{cc}
  0 & G_{\theta \alpha_{i,j}}' \\
\end{array}%
\right), & \hbox{if $j \le d_{\alpha}$;} \\
    - \left(%
\begin{array}{cc}
  G_{\theta \alpha_i}' (I_{d_{\alpha}} \otimes e_{j-d_{\alpha}}) & 0 \\
\end{array}%
\right), & \hbox{if $j>d_{\alpha}$.} \\
\end{array}%
\right.
\end{eqnarray*}

\noindent \textbf{Third Derivatives}

\begin{eqnarray*}
\hat{M}_{i,jk}^{W}(\theta,\gamma_i) &=&
\frac{\partial^{3}
\hat{s}_{i}^{W}(\theta,\gamma_i)}{\partial
\gamma_{i,k} \partial \gamma_{i,j}
\partial \gamma_i'} = \left\{%
\begin{array}{ll}
    - \left(%
\begin{array}{cc}
   \hat{G}_{\theta \alpha, \alpha \alpha_{i,jk}}(\theta,\alpha_i)' (I_{d_{\alpha}} \otimes \lambda_i)  & \hat{G}_{\theta \alpha \alpha_{i,jk}} (\theta, \alpha_i)' \\
\end{array}%
\right), & \hbox{if $j\le d_{\alpha}, k\le d_{\alpha}$;} \\
   - \left(%
\begin{array}{cc}
  \hat{G}_{\theta \alpha, \alpha_{i,j}}(\theta,\alpha_i)' (I_{d_{\alpha}} \otimes e_{k-d_{\alpha}})  & 0 \\
\end{array}%
\right), & \hbox{if $j \le d_{\alpha}, k>d_{\alpha}$;} \\
   - \left(%
\begin{array}{cc}
  \hat{G}_{\theta \alpha, \alpha_{i,k}}(\theta,\alpha_i)' (I_{d_{\alpha}} \otimes e_{j-d_{\alpha}})  & 0 \\
\end{array}%
\right)
, & \hbox{if $j>d_{\alpha}, k\le d_{\alpha}$;} \\
   - \left(%
\begin{array}{cc}
 0 & 0 \\
\end{array}%
\right), & \hbox{if $j>d_{\alpha}, k >d_{\alpha}$.} \\
\end{array}%
\right.
 \\
M_{i,jk}^{W} &=& E \left[ \hat{M}_{i,jk}^{W}
\right] = \left\{%
\begin{array}{ll}
    - \left(%
\begin{array}{cc}
   0  & G_{\theta \alpha \alpha_{i,jk}}' \\
\end{array}%
\right), & \hbox{if $j\le d_{\alpha}, k\le d_{\alpha}$;} \\
   - \left(%
\begin{array}{cc}
  G_{\theta \alpha, \alpha_{i,j}}' (I_{d_{\alpha}} \otimes e_{k-d_{\alpha}})  & 0 \\
\end{array}%
\right), & \hbox{if $j \le d_{\alpha}, k>d_{\alpha}$;} \\
   - \left(%
\begin{array}{cc}
  G_{\theta \alpha, \alpha_{i,k}}' (I_{d_{\alpha}} \otimes e_{j-d_{\alpha}})  & 0 \\
\end{array}%
\right)
, & \hbox{if $j>d_{\alpha}, k\le d_{\alpha}$;} \\
   - \left(%
\begin{array}{cc}
 0 & 0 \\
\end{array}%
\right), & \hbox{if $j>d_{\alpha}, k >d_{\alpha}$.} \\
\end{array}%
\right.
 \\
\end{eqnarray*}

\subsubsection{Derivatives with respect to the common parameters}

\noindent \textbf{\\First Derivatives}

\begin{eqnarray*}
\hat{S}_{i,j}^{W}(\theta,\gamma_i) &=& \frac{\partial \hat{s}_{i}^{W}(\theta,\gamma_i)}{\partial \theta_{j}}= -   \hat{G}_{\theta \theta _{i,j}} (\theta, \alpha_i)'\lambda_i.  \\
S_{i,j}^{W} &=& E \left[ \hat{S}_{i,j}^{W}
\right] =  0.
\end{eqnarray*}



\subsection{Two-Step Score and Derivatives: Fixed Effects}
\label{appendixM}

\subsubsection{Score}

\begin{eqnarray*}
\hat{t}_{i}(\theta, \gamma_{i}) &=& - \frac{1}{T}
\sum_{t=1}^{T} \left(%
\begin{array}{c}
  G_{\alpha_{i}}(z_{it};\theta, \alpha_{i})'\lambda_{i} \\
  g(z_{it};\theta,\alpha_{i}) + \hat{\Omega}_{i}(\tilde{\theta}, \tilde{\alpha}_{i}) \lambda_{i} \\
\end{array}%
\right) = - \left(%
\begin{array}{c}
  \hat{G}_{\alpha_{i}}(\theta, \alpha_{i})'\lambda_{i} \\
  \hat{g}_{i}(\theta,\alpha_{i}) + \Omega_{i} \lambda_{i} \\
\end{array}%
\right) - \left(%
\begin{array}{c}
  0 \\
  (\hat{\Omega}_{i} - \Omega_{i}) \lambda_{i} \\
\end{array}%
\right) \notag \\
 &=& \hat{t}_{i}^{\Omega} (\theta, \gamma_{i}) +
\hat{t}_{i}^{R} (\theta, \gamma_{i}).
\end{eqnarray*}
Note that the formulae for the derivatives of Appendix
\ref{appendixK} apply for $\hat{t}_{i}^{\Omega}$, replacing $\hat W$ by
$\Omega$. Hence, we only need to derive the derivatives for
$\hat{t}_{i}^{R}$.

\subsubsection{Derivatives with respect to the fixed effects}

\noindent \textbf{\\First Derivatives}

\begin{eqnarray*}
\hat{T}_{i}^{R}(\theta, \gamma_{i}) &=& \frac{\partial \hat{t}_{i}^{R}(\theta, \gamma_{i})}{\partial \gamma_{i}'}= - \left(%
\begin{array}{cc}
  0 & 0 \\
  0 & \hat{\Omega}_{i}(\tilde{\theta}, \tilde{\alpha}_{i}) - \Omega_{i} \\
\end{array}%
\right). \\
T_{i}^{R} &=& E \left[ \hat{T}_{i}^{R} \right] = - \left(%
\begin{array}{cc}
  0 & 0 \\
  0 & E \left[ \hat{\Omega}_{i} - \Omega_{i} \right] \\
\end{array}%
\right).
\end{eqnarray*}

\noindent \textbf{Second and Third Derivatives }

Since $\hat{T}_{i}^{R}(\gamma_{i}, \theta)$ does not depend on
$\gamma_{i}$, the derivatives (and its expectation) of order greater
than one are zero.

\subsubsection{Derivatives with respect to the common parameters}

\noindent \textbf{\\First Derivatives}

\begin{eqnarray*}
\hat{N}_{i}^{R}(\theta, \gamma_{i}) &=& \frac{\partial
\hat{t}_{i}^{R}(\theta, \gamma_{i})}{\partial \theta'}= 0.
\end{eqnarray*}


\subsection{Two-Step Score and Derivatives: Common Parameters}
\label{appendixN}

\subsubsection{Score}

\begin{equation*}
\hat{s}_{i} (\theta,\gamma_{i}) = - \frac{1}{T}
\sum_{t=1}^{T}
G_{\theta}(z_{it};\theta,\alpha_{i})'\lambda_{i}
= - \hat{G}_{\theta_i}(\theta,\alpha_{i})'\lambda_{i}.
\end{equation*}
Since this score does not depend explicitly on
$\hat{\Omega}_{i}(\tilde{\theta}, \tilde{\alpha}_{i})$, the formulae
for the derivatives are the same as in Appendix \ref{appendixL}.

\section*{References}

\textsc{Andrews, D. W. K.} (1991), ``Heteroskedasticity and Autocorrelation Consistent Covariance Matrix Estimation,''
\emph{Econometrica} 59, 817-858.

\textsc{Fern\'andez-Val, I. and, J. Lee} (2012), \enquote{Panel Data Models with Nonadditive Unobserved Heterogeneity: Estimation and Inference,} unpublished manuscript, Boston University.

\textsc{Hahn, J., and G. Kuersteiner} (2011), ``Bias Reduction for
Dynamic Nonlinear Panel Models with Fixed Effects,''
\emph{Econometric Theory} 27, 1152-1191.

\textsc{Hall, P., and C. Heide} (1980), \emph{Martingale Limit Theory and Applications}. Academic Press.

\textsc{Newey, W.K., and D. McFadden} (1994), \textquotedblleft
Large Sample Estimation and Hypothesis Testing,\textquotedblright\
in \textsc{R.F. Engle and D.L. McFadden,} eds., \emph{Handbook of
Econometrics, Vol. 4}. Elsevier Science. Amsterdam: North-Holland.

\textsc{Newey, W.K., and R.
Smith} (2004), ``Higher Order Properties of GMM and Generalized
Empirical Likelihood Estimators,'' \emph{Econometrica} 72,
219-255.

\end{appendix}

\begin{sidewaystable}

\vspace{15cm}

\begin{centering}
\textbf{Table A1: Common Parameter $\theta_{2}$}
\par\end{centering}
\begin{center}
\begin{tabular}{ccccccccccccccccc}
\toprule\toprule & \multicolumn{4}{c}{$\rho_1=0$} &
\multicolumn{4}{c}{$\rho_1=0.3$} & \multicolumn{4}{c}{$\rho_1=0.6$}
& \multicolumn{4}{c}{$\rho_1=0.9$}\tabularnewline
\cmidrule(rl){2-5}\cmidrule(rl){6-9}\cmidrule(rl){10-13}\cmidrule(rl){14-17}Estimator
& Bias & SD & SE/SD & p;.05 & Bias & SD & SE/SD & p;.05 & Bias & SD
& SE/SD & p;.05 & Bias & SD & SE/SD & p;.05\tabularnewline \hline
 & \multicolumn{16}{c}{$\psi=2$}\tabularnewline
$OLS-FC$ & 0.06  & 0.01  & 0.84  & 1.00  & 0.06  & 0.01  & 0.83  &
1.00  & 0.06  & 0.01  & 0.84  & 1.00  & 0.07  & 0.01  & 0.71  & 1.00
\tabularnewline $IV-FC$ & 0.00  & 0.01  & 0.90  & 0.08  & -0.01  &
0.02  & 0.84  & 0.11  & -0.01  & 0.02  & 0.78  & 0.18  & -0.01  &
0.02  & 0.63  & 0.28 \tabularnewline $OLS-RC$ & 0.04  & 0.01  & 0.97
& 1.00  & 0.04  & 0.01  & 0.99  & 1.00  & 0.04  & 0.01  & 1.02  &
1.00  & 0.04  & 0.01  & 0.96  & 1.00 \tabularnewline $BC-OLS$ & 0.04
& 0.01  & 0.97  & 1.00  & 0.04  & 0.01  & 0.99  & 1.00  & 0.04  &
0.01  & 1.02  & 1.00  & 0.04  & 0.01  & 0.96  & 1.00 \tabularnewline
$IBC-OLS$ & 0.04  & 0.01  & 0.97  & 1.00  & 0.04  & 0.01  & 0.99  &
1.00  & 0.04  & 0.01  & 1.02  & 1.00  & 0.04  & 0.01  & 0.96  & 1.00
\tabularnewline $IV-RC$ & 0.00  & 0.01  & 1.00  & 0.06  & 0.00  &
0.01  & 1.01  & 0.05  & 0.00  & 0.01  & 1.00  & 0.05  & 0.00  & 0.01
& 1.01  & 0.05 \tabularnewline $BC-IV$ & 0.00  & 0.01  & 0.99  &
0.06  & 0.00  & 0.01  & 1.01  & 0.05  & 0.00  & 0.01  & 1.00  & 0.05
& 0.00  & 0.01  & 1.00  & 0.05 \tabularnewline $IBC-IV$ & 0.00  &
0.01  & 0.99  & 0.06  & 0.00  & 0.01  & 1.01  & 0.05  & 0.00  & 0.01
& 1.00  & 0.05  & 0.00  & 0.01  & 1.00  & 0.05 \tabularnewline
 & \multicolumn{16}{c}{$\psi=4$}\tabularnewline
$OLS-FC$ & 0.12  & 0.01  & 1.09  & 1.00  & 0.12  & 0.01  & 1.03  &
1.00  & 0.12  & 0.01  & 1.10  & 1.00  & 0.12  & 0.01  & 1.07  & 1.00
\tabularnewline $IV-FC$ & 0.00  & 0.02  & 0.94  & 0.07  & -0.01  &
0.02  & 0.89  & 0.08  & -0.01  & 0.02  & 0.92  & 0.09  & -0.01  &
0.03  & 0.79  & 0.15 \tabularnewline $OLS-RC$ & 0.10  & 0.01  & 1.06
& 1.00  & 0.10  & 0.01  & 1.05  & 1.00  & 0.11  & 0.01  & 1.08  &
1.00  & 0.11  & 0.01  & 1.07  & 1.00 \tabularnewline $BC-OLS$ & 0.10
& 0.01  & 1.06  & 1.00  & 0.10  & 0.01  & 1.05  & 1.00  & 0.11  &
0.01  & 1.08  & 1.00  & 0.11  & 0.01  & 1.07  & 1.00 \tabularnewline
$IBC-OLS$ & 0.10  & 0.01  & 1.06  & 1.00  & 0.10  & 0.01  & 1.05  &
1.00  & 0.11  & 0.01  & 1.08  & 1.00  & 0.11  & 0.01  & 1.07  & 1.00
\tabularnewline $IV-RC$ & 0.00  & 0.02  & 0.98  & 0.06  & 0.00  &
0.02  & 0.96  & 0.06  & 0.00  & 0.02  & 1.01  & 0.05  & 0.00  & 0.02
& 1.00  & 0.06 \tabularnewline $BC-IV$ & 0.00  & 0.02  & 0.97  &
0.05  & 0.00  & 0.02  & 0.95  & 0.06  & 0.00  & 0.02  & 1.00  & 0.05
& 0.00  & 0.02  & 0.99  & 0.06 \tabularnewline $IBC-IV$ & 0.00  &
0.02  & 0.97  & 0.05  & 0.00  & 0.02  & 0.95  & 0.06  & 0.00  & 0.02
& 1.00  & 0.05  & 0.00  & 0.02  & 0.99  & 0.06 \tabularnewline
 & \multicolumn{16}{c}{$\psi=6$}\tabularnewline
$OLS-FC$ & 0.16  & 0.01  & 1.27  & 1.00  & 0.16  & 0.01  & 1.22  &
1.00  & 0.16  & 0.01  & 1.25  & 1.00  & 0.16  & 0.01  & 1.34  & 1.00
\tabularnewline $IV-FC$ & 0.00  & 0.03  & 0.95  & 0.06  & 0.00  &
0.03  & 0.94  & 0.06  & -0.01  & 0.03  & 0.92  & 0.08  & -0.01  &
0.04  & 0.92  & 0.08 \tabularnewline $OLS-RC$ & 0.15  & 0.01  & 1.20
& 1.00  & 0.15  & 0.01  & 1.21  & 1.00  & 0.15  & 0.01  & 1.21  &
1.00  & 0.15  & 0.01  & 1.26  & 1.00 \tabularnewline $BC-OLS$ & 0.15
& 0.01  & 1.20  & 1.00  & 0.15  & 0.01  & 1.21  & 1.00  & 0.15  &
0.01  & 1.21  & 1.00  & 0.15  & 0.01  & 1.26  & 1.00 \tabularnewline
$IBC-OLS$ & 0.15  & 0.01  & 1.20  & 1.00  & 0.15  & 0.01  & 1.21  &
1.00  & 0.15  & 0.01  & 1.21  & 1.00  & 0.15  & 0.01  & 1.26  & 1.00
\tabularnewline $IV-RC$ & 0.00  & 0.03  & 0.98  & 0.06  & 0.00  &
0.03  & 1.00  & 0.04  & 0.00  & 0.03  & 1.01  & 0.05  & 0.00  & 0.03
& 1.05  & 0.04 \tabularnewline $BC-IV$ & 0.00  & 0.03  & 0.95  &
0.06  & 0.00  & 0.03  & 0.97  & 0.04  & 0.00  & 0.03  & 0.98  & 0.05
& 0.00  & 0.03  & 1.02  & 0.04 \tabularnewline $IBC-IV$ & 0.00  &
0.03  & 0.95  & 0.06  & 0.00  & 0.03  & 0.97  & 0.04  & 0.00  & 0.03
& 0.98  & 0.05  & 0.00  & 0.03  & 1.02  & 0.04 \tabularnewline
\bottomrule\bottomrule &  &  &  &  &  &  &  &  &  &  &  &  &  &  &
& \tabularnewline
\end{tabular}
\end{center}
{\footnotesize RC/FC refers to random/fixed coefficient model.
BC/IBC refers to bias corrected/iterated bias corrected estimates.
\\Note: $1,000$ repetitions.}

\end{sidewaystable}

\begin{sidewaystable}

\vspace{15cm}

\begin{centering}
\textbf{Table A2: Mean of Individual Specific Parameter
$\mu_{1}=\bar{E}[\alpha_{1i}]$}
\par\end{centering}

\begin{center}

\begin{tabular}{ccccccccccccccccc}
\toprule\toprule & \multicolumn{4}{c}{$\rho_1=0$} &
\multicolumn{4}{c}{$\rho_1=0.3$} & \multicolumn{4}{c}{$\rho_1=0.6$}
& \multicolumn{4}{c}{$\rho_1=0.9$}\tabularnewline
\cmidrule(rl){2-5}\cmidrule(rl){6-9}\cmidrule(rl){10-13}\cmidrule(rl){14-17}Estimator
& Bias & SD & SE/SD & p;.05 & Bias & SD & SE/SD & p;.05 & Bias & SD
& SE/SD & p;.05 & Bias & SD & SE/SD & p;.05\tabularnewline \hline
 & \multicolumn{16}{c}{$\psi=2$}\tabularnewline
$OLS-FC$ & 2.33  & 1.65  & 0.35  & 0.78  & 2.58  & 1.91  & 0.31  &
0.79  & 3.01  & 1.92  & 0.30  & 0.84  & 3.68  & 2.20  & 0.27  & 0.89
\tabularnewline $IV-FC$ & 0.08  & 1.59  & 0.40  & 0.44  & 0.16  &
1.72  & 0.37  & 0.47  & 0.46  & 1.66  & 0.39  & 0.46  & 0.96  & 1.80
& 0.37  & 0.53 \tabularnewline $OLS-RC$ & 1.16  & 1.53  & 1.02  &
0.12  & 1.15  & 1.65  & 0.96  & 0.12  & 1.19  & 1.59  & 0.99  & 0.12
& 1.25  & 1.62  & 0.97  & 0.13 \tabularnewline $BC-OLS$ & 1.16  &
1.53  & 0.97  & 0.14  & 1.15  & 1.65  & 0.92  & 0.14  & 1.19  & 1.59
& 0.95  & 0.14  & 1.25  & 1.62  & 0.93  & 0.15 \tabularnewline
$IBC-OLS$ & 1.16  & 1.53  & 0.97  & 0.14  & 1.15  & 1.65  & 0.92  &
0.14  & 1.19  & 1.59  & 0.95  & 0.14  & 1.25  & 1.62  & 0.93  & 0.15
\tabularnewline $IV-RC$ & 0.01  & 1.51  & 1.07  & 0.04  & -0.01  &
1.62  & 1.00  & 0.05  & 0.02  & 1.56  & 1.04  & 0.05  & 0.08  & 1.59
& 1.03  & 0.05 \tabularnewline $BC-IV$ & -0.01  & 1.51  & 1.02  &
0.04  & -0.02  & 1.62  & 0.96  & 0.06  & 0.00  & 1.56  & 1.00  &
0.06  & 0.06  & 1.59  & 0.98  & 0.06 \tabularnewline $IBC-IV$ &
-0.01  & 1.51  & 1.02  & 0.04  & -0.03  & 1.62  & 0.96  & 0.06  &
0.00  & 1.56  & 1.00  & 0.06  & 0.06  & 1.59  & 0.98  & 0.06
\tabularnewline
 & \multicolumn{16}{c}{$\psi=4$}\tabularnewline
$OLS-FC$ & 4.15  & 1.84  & 0.52  & 0.90  & 4.43  & 1.95  & 0.49  &
0.90  & 4.90  & 2.08  & 0.46  & 0.93  & 5.45  & 2.22  & 0.43  & 0.95
\tabularnewline $IV-FC$ & 0.09  & 1.85  & 0.59  & 0.25  & 0.21  &
1.92  & 0.57  & 0.27  & 0.56  & 1.89  & 0.58  & 0.27  & 1.03  & 1.94
& 0.57  & 0.32 \tabularnewline $OLS-RC$ & 3.19  & 1.76  & 1.06  &
0.41  & 3.12  & 1.81  & 1.04  & 0.38  & 3.12  & 1.76  & 1.07  & 0.38
& 3.18  & 1.78  & 1.06  & 0.38 \tabularnewline $BC-OLS$ & 3.19  &
1.76  & 0.93  & 0.50  & 3.12  & 1.81  & 0.91  & 0.48  & 3.12  & 1.76
& 0.94  & 0.47  & 3.18  & 1.78  & 0.93  & 0.47 \tabularnewline
$IBC-OLS$ & 3.19  & 1.76  & 0.93  & 0.50  & 3.12  & 1.81  & 0.91  &
0.48  & 3.12  & 1.76  & 0.94  & 0.47  & 3.18  & 1.78  & 0.93  & 0.47
\tabularnewline $IV-RC$ & 0.06  & 1.78  & 1.15  & 0.03  & -0.01  &
1.86  & 1.10  & 0.03  & 0.03  & 1.78  & 1.15  & 0.03  & 0.10  & 1.78
& 1.15  & 0.03 \tabularnewline $BC-IV$ & 0.00  & 1.78  & 1.02  &
0.05  & -0.08  & 1.86  & 0.98  & 0.05  & -0.04  & 1.78  & 1.02  &
0.05  & 0.03  & 1.78  & 1.02  & 0.05 \tabularnewline $IBC-IV$ &
-0.01  & 1.78  & 1.02  & 0.05  & -0.08  & 1.86  & 0.98  & 0.05  &
-0.04  & 1.78  & 1.02  & 0.05  & 0.03  & 1.78  & 1.02  & 0.05
\tabularnewline
 & \multicolumn{16}{c}{$\psi=6$}\tabularnewline
$OLS-FC$ & 5.62  & 2.13  & 0.62  & 0.93  & 5.87  & 2.25  & 0.58  &
0.92  & 6.19  & 2.31  & 0.57  & 0.93  & 6.35  & 2.28  & 0.57  & 0.95
\tabularnewline $IV-FC$ & 0.14  & 2.29  & 0.69  & 0.17  & 0.26  &
2.34  & 0.68  & 0.19  & 0.53  & 2.31  & 0.69  & 0.19  & 0.80  & 2.26
& 0.70  & 0.20 \tabularnewline $OLS-RC$ & 4.69  & 2.10  & 1.08  &
0.53  & 4.59  & 2.14  & 1.07  & 0.51  & 4.52  & 2.11  & 1.09  & 0.50
& 4.30  & 2.01  & 1.15  & 0.46 \tabularnewline $BC-OLS$ & 4.69  &
2.10  & 0.88  & 0.69  & 4.59  & 2.14  & 0.88  & 0.67  & 4.52  & 2.11
& 0.89  & 0.64  & 4.30  & 2.01  & 0.94  & 0.61 \tabularnewline
$IBC-OLS$ & 4.69  & 2.10  & 0.88  & 0.69  & 4.59  & 2.14  & 0.88  &
0.67  & 4.52  & 2.11  & 0.89  & 0.64  & 4.30  & 2.01  & 0.94  & 0.61
\tabularnewline $IV-RC$ & 0.09  & 2.30  & 1.12  & 0.04  & 0.05  &
2.33  & 1.11  & 0.03  & 0.03  & 2.23  & 1.16  & 0.02  & -0.10  &
2.18  & 1.19  & 0.02 \tabularnewline $BC-IV$ & -0.05  & 2.30  & 0.95
& 0.06  & -0.10  & 2.33  & 0.94  & 0.06  & -0.12  & 2.23  & 0.97  &
0.06  & -0.26  & 2.18  & 1.00  & 0.05 \tabularnewline $IBC-IV$ &
-0.06  & 2.30  & 0.95  & 0.06  & -0.10  & 2.32  & 0.94  & 0.06  &
-0.13  & 2.23  & 0.97  & 0.06  & -0.26  & 2.18  & 1.00  & 0.05
\tabularnewline \bottomrule\bottomrule &  &  &  &  &  &  &  &  &  &
&  &  &  &  &  & \tabularnewline
\end{tabular}

\end{center}
{\footnotesize RC/FC refers to random/fixed coefficient model.
BC/IBC refers to bias corrected/iterated bias corrected estimates.
\\Note:  $1,000$ repetitions.}
\end{sidewaystable}

\begin{sidewaystable}

\vspace{15cm}

\begin{centering}
\textbf{Table A3: Standard Deviation of the Individual Specific
Parameter $\sigma_{1}=\bar{E}[(\alpha_{1i} - \mu_1)^2]^{1/2}$}
\par\end{centering}

\begin{center}

\begin{tabular}{ccccccccccccccccc}
\toprule\toprule & \multicolumn{4}{c}{$\rho_1=0$} &
\multicolumn{4}{c}{$\rho_1=0.3$} & \multicolumn{4}{c}{$\rho_1=0.6$}
& \multicolumn{4}{c}{$\rho_1=0.9$}\tabularnewline
\cmidrule(rl){2-5}\cmidrule(rl){6-9}\cmidrule(rl){10-13}\cmidrule(rl){14-17}Estimator
& Bias & SD & SE/SD & p;.05 & Bias & SD & SE/SD & p;.05 & Bias & SD
& SE/SD & p;.05 & Bias & SD & SE/SD & p;.05\tabularnewline \hline
 & \multicolumn{16}{c}{$\psi=2$}\tabularnewline
$OLS-RC$ & 0.01  & 1.06  & 1.02  & 0.05  & 0.15  & 1.06  & 1.02  &
0.05  & 0.11  & 1.08  & 0.99  & 0.06  & 0.17  & 1.06  & 0.99  & 0.06
\tabularnewline $BC-OLS$ & -0.63  & 1.10  & 1.04  & 0.10  & -0.48  &
1.11  & 1.04  & 0.09  & -0.52  & 1.12  & 1.01  & 0.10  & -0.46  &
1.11  & 1.01  & 0.09 \tabularnewline $IBC-OLS$ & -0.63  & 1.10  &
1.04  & 0.10  & -0.48  & 1.11  & 1.04  & 0.09  & -0.52  & 1.12  &
1.01  & 0.10  & -0.46  & 1.11  & 1.01  & 0.09 \tabularnewline
$IV-RC$ & 0.38  & 1.08  & 1.03  & 0.05  & 0.47  & 1.10  & 1.02  &
0.06  & 0.41  & 1.13  & 0.98  & 0.06  & 0.46  & 1.11  & 0.99  & 0.06
\tabularnewline $BC-IV$ & -0.25  & 1.13  & 1.05  & 0.06  & -0.16  &
1.14  & 1.04  & 0.06  & -0.22  & 1.18  & 1.00  & 0.07  & -0.17  &
1.16  & 1.00  & 0.06 \tabularnewline $IBC-IV$ & -0.25  & 1.13  &
1.05  & 0.06  & -0.16  & 1.14  & 1.04  & 0.06  & -0.22  & 1.18  &
1.00  & 0.07  & -0.17  & 1.16  & 1.00  & 0.06 \tabularnewline
 & \multicolumn{16}{c}{$\psi=4$}\tabularnewline
$OLS-RC$ & 0.89  & 1.21  & 1.17  & 0.04  & 0.98  & 1.20  & 1.17  &
0.05  & 1.08  & 1.16  & 1.19  & 0.06  & 1.09  & 1.05  & 1.23  & 0.08
\tabularnewline $BC-OLS$ & -1.24  & 1.46  & 1.19  & 0.08  & -1.13  &
1.44  & 1.20  & 0.08  & -1.02  & 1.41  & 1.20  & 0.05  & -0.98  &
1.23  & 1.24  & 0.03 \tabularnewline $IBC-OLS$ & -1.24  & 1.46  &
1.19  & 0.08  & -1.13  & 1.44  & 1.20  & 0.08  & -1.02  & 1.41  &
1.20  & 0.05  & -0.98  & 1.24  & 1.22  & 0.03 \tabularnewline
$IV-RC$ & 1.84  & 1.28  & 1.17  & 0.17  & 1.83  & 1.29  & 1.16  &
0.16  & 1.85  & 1.26  & 1.17  & 0.18  & 1.87  & 1.18  & 1.17  & 0.20
\tabularnewline $BC-IV$ & -0.25  & 1.52  & 1.20  & 0.03  & -0.26  &
1.52  & 1.19  & 0.02  & -0.26  & 1.51  & 1.18  & 0.03  & -0.21  &
1.37  & 1.18  & 0.02 \tabularnewline $IBC-IV$ & -0.25  & 1.52  &
1.20  & 0.03  & -0.26  & 1.52  & 1.19  & 0.02  & -0.26  & 1.51  &
1.18  & 0.03  & -0.21  & 1.38  & 1.16  & 0.02 \tabularnewline
 & \multicolumn{16}{c}{$\psi=6$}\tabularnewline
$OLS-RC$ & 2.35  & 1.33  & 1.38  & 0.14  & 2.60  & 1.40  & 1.30  &
0.21  & 2.57  & 1.37  & 1.31  & 0.21  & 2.69  & 1.31  & 1.28  & 0.38
\tabularnewline $BC-OLS$ & -2.06  & 2.04  & 1.41  & 0.00  & -1.71  &
2.14  & 1.30  & 0.01  & -1.75  & 2.06  & 1.35  & 0.00  & -1.54  &
1.78  & 1.28  & 0.00 \tabularnewline $IBC-OLS$ & -2.06  & 2.04  &
1.41  & 0.00  & -1.71  & 2.14  & 1.30  & 0.01  & -1.75  & 2.06  &
1.35  & 0.00  & -1.54  & 1.80  & 1.26  & 0.00 \tabularnewline
$IV-RC$ & 3.79  & 1.52  & 1.31  & 0.46  & 3.87  & 1.55  & 1.28  &
0.49  & 3.78  & 1.50  & 1.30  & 0.47  & 3.87  & 1.48  & 1.23  & 0.60
\tabularnewline $BC-IV$ & -0.49  & 2.13  & 1.37  & 0.00  & -0.42  &
2.23  & 1.29  & 0.01  & -0.55  & 2.14  & 1.35  & 0.00  & -0.40  &
1.96  & 1.24  & 0.01 \tabularnewline $IBC-IV$ & -0.49  & 2.13  &
1.37  & 0.00  & -0.41  & 2.23  & 1.29  & 0.01  & -0.55  & 2.14  &
1.35  & 0.00  & -0.39  & 1.97  & 1.22  & 0.02 \tabularnewline
\bottomrule\bottomrule &  &  &  &  &  &  &  &  &  &  &  &  &  &  &
& \tabularnewline
\end{tabular}

\end{center}
{\footnotesize RC/FC refers to random/fixed coefficient model.
BC/IBC refers to bias corrected/iterated bias corrected estimates.
\\Note: $1,000$ repetitions.}
\end{sidewaystable}

\end{document}